\documentclass{article}

\usepackage{main}

\usepackage[utf8]{inputenc} 
\usepackage[T1]{fontenc}    
\usepackage{hyperref}       
\usepackage{url}            
\usepackage{booktabs}       
\usepackage{amsfonts}       
\usepackage{nicefrac}       
\usepackage{microtype}      
\usepackage{lineno}
\usepackage{lipsum}
\usepackage{graphicx}
\usepackage{cancel}
\usepackage{enumerate}
\usepackage{natbib}
\usepackage{subcaption}
\usepackage{amsmath}
\usepackage{amssymb}
\usepackage{mathtools}
\usepackage{amsthm}
\usepackage{todonotes}
\usepackage[most]{tcolorbox}

\theoremstyle{plain}

\theoremstyle{definition}

\theoremstyle{remark}

\makeatletter

\def\th@plain{%
  \thm@headfont{\bfseries}
  \thm@bodyfont{\itshape}
  \thm@headpunct{.}
  \thm@preskip=12pt
  \thm@postskip=12pt
}

\def\th@remark{%
  \thm@headfont{\itshape\bfseries}
  \thm@bodyfont{\normalfont}
  \thm@headpunct{.}
  \thm@preskip=10pt
  \thm@postskip=10pt
}

\makeatother

\usepackage[most]{tcolorbox}
\tcbuselibrary{theorems,breakable}

\tcbset{
  formalblue/.style={
    enhanced,
    breakable,
    boxrule=0.6pt,
    arc=2pt,
    left=7pt,
    right=7pt,
    top=7pt,
    bottom=7pt,
    fonttitle=\bfseries,
    coltitle=black,
    colback=blue!2,
    colframe=blue!45!black,
    colbacktitle=blue!8,
    before skip=12pt,
    after skip=12pt
  }
}

\newtcbtheorem{boxtheorem}{Theorem}{formalblue}{thm}
\newtcbtheorem{boxprop}{Proposition}{formalblue}{prop}
\newtcbtheorem{boxassumption}{Assumption}{formalblue}{assm}
\newtcbtheorem{boxremark}{Remark}{formalblue}{rem}

\graphicspath{ {./docs/} }

\title{Koopman early warning signals for bifurcation and rate-induced tipping}

\author{
Juan Nathaniel\textsuperscript{1}\thanks{Corresponding author: \texttt{jn2808@columbia.edu}} \quad
Carla Roesch\textsuperscript{2} \quad
Derek DeSantis\textsuperscript{3} \quad
Parvathi Kooloth\textsuperscript{4} \quad
Hang Fan\textsuperscript{1} \\
\textbf{
Valerio Lucarini\textsuperscript{5} \quad
Anastasia Romanou\textsuperscript{6} \quad
Pierre Gentine\textsuperscript{1}
} \\
\textsuperscript{1}Columbia University, USA \quad
\textsuperscript{2}University of Edinburgh, UK \quad \\
\textsuperscript{3}Los Alamos National Laboratory, USA 
\textsuperscript{4}Pacific Northwest National Laboratory, USA \quad \\
\textsuperscript{5}University of Leicester, UK \quad
\textsuperscript{6}NASA Goddard Institute for Space Studies, USA
}

\begin{document}
\maketitle

\begin{abstract}
Abrupt transitions in complex systems are often preceded by early warning signals. However, most indicators rely on the notion of critical slowing down and do not generally extend to rate-induced tipping where transitions can occur without local loss of stability. This is problematic in stochastic, nonautonomous systems where internal variability and time-varying variables interact to shape tipping onset. We use Koopman operator theory to develop a unified early warning framework for both bifurcation and rate-induced tipping in stochastic systems. Our approach builds on residual Koopman mode decomposition that measures discrepancies between dynamics and their finite-dimensional approximation, and extends it to the control setting by augmenting the observable space with time-varying control variables. In idealized examples, the resulting indicators recover expected signatures near bifurcation points and improve detection in rate-induced regimes where classical indicators fail. We further show that learned embeddings through deep learning outperform prescribed dictionaries, especially in a high-dimensional setting. Applied to simulations of the Atlantic Meridional Overturning Circulation, our Koopman-based indicators distinguish tipping from non-tipping trajectories and reveal interpretable spectral signatures prior to critical transition.
\end{abstract}


\section*{Introduction}

Critical transitions, or tipping points, are abrupt shifts in system dynamics that can occur when a stable regime loses resilience or can no longer track changing external conditions. Such transitions have been studied across complex systems, including in ecology, biology, and climate science~\cite{scheffer2009early,scheffer2010foreseeing,lenton2012early,wunderling2024climate,ritchie2021overshooting,dakos2024tipping,panahi2023rate,hastings2026tipping,trefois2015critical,lenton2008tipping,kooloth2026time,Boers2025}. Because these transitions can be severe, irreversible, and strongly amplified by stochastic fluctuations, considerable effort has been devoted to the development of early warning signals (EWS) capable of anticipating their onset~\cite{ritchie2017probability,slyman2023rate,LucariniChekroun2023,LucariniChekroun2024}.

Most classical EWS are rooted in bifurcation-induced tipping (B-tipping), where a slowly varying parameter drives the system toward loss of local stability. Near such transitions, perturbations recover increasingly slowly, a phenomenon known as critical slowing down (CSD)~\cite{dakos2012robustness}. In stochastic systems, CSD typically manifests as increasing lag-1 autocorrelation (AR1) and variance~\cite{ives1995measuring}, making these quantities among some of the most widely used indicators. Spectral theory based on combining Koopman framework \cite{Budinisic2012} with response theory also provides a mathematical explanation for this behaviour \cite{Santos2022,LucariniChekroun2023}. A distinct mechanism arises in rate-induced tipping (R-tipping)~\cite{ashwin2012tipping,kiers2020rate,ritchie2023rate}, where transitions occur not because stability is lost, but because external conditions change too rapidly for the system to track the evolving attractor. In this case, the instantaneous dynamics may remain locally stable throughout the transition pathway, so that classical indicators based on CSD may fail to provide reliable warning signals~\cite{huang2024deep}.

In realistic stochastic and externally forced systems, different tipping mechanisms may coexist and can be difficult to distinguish across multiple interacting timescales~\cite{nathaniel2024chaosbench,herdeanu2026causaldynamics,sevinchan2026context}. This is especially relevant in climate dynamics, where internal variability, nonstationary forcing, and multiscale feedbacks jointly shape the onset of critical transitions. Existing EWS are typically tailored to either stability loss (B-tipping) or tracking failure (R-tipping) separately, raising the need for frameworks that remain informative across both regimes.

Here we address this challenge using Koopman operator theory, which represents finite-dimensional nonlinear dynamics through the linear evolution of infinite-dimensional observables~\cite{brunton2021modern,koopman1931hamiltonian,Budinisic2012,nathaniel2026generative}. Previous works have connected B-tipping and CSD to Koopman spectral structure, including spectral gap closure (Figure~\ref{fig:graphical_abstract}a), the amplification of stochastic fluctuations (Figure~\ref{fig:graphical_abstract}b)~\cite{LucariniChekroun2023,Tantet2018,chekroun2019c,Held2004,Santos2022,miyauchi2026generalized}, and the resulting deterioration of Koopman representation quantified through residual Koopman mode decomposition (ResKMD) near critical transition~\cite{miyauchi2026generalized}. In this work, we extend this idea to the rate-induced setting using a control-aware Koopman formulation~\cite{proctor2018generalizing}, addressing both tipping mechanisms under a unified framework. 

We evaluate our framework across a hierarchy of systems, ranging from prototypical models to reduced-order and fully coupled simulations of the Atlantic Meridional Overturning Circulation (AMOC), a widely discussed climate tipping element. We further show that learned embeddings through deep learning (DL), such as through multi-layer perceptron (MLP)-based autoencoders, outperform prescribed observables, such as random Fourier (RFF~\cite{rahimi2007random}) and time-delay features (TDF~\cite{abarbanel1994predicting}), in terms of better separation of tipping and non-tipping trajectories, especially in high-dimensional settings. Beyond early warning detection, our framework also provides interpretable spectral diagnostics, including decay timescales, spectral gap structure, and changes in memory preceding critical transition. Together, these results demonstrate how operator-theoretic methods can unify prediction and mechanistic interpretation of different tipping phenomena in stochastic systems.

\begin{figure}
    \centering
    \includegraphics[width=0.9\linewidth]{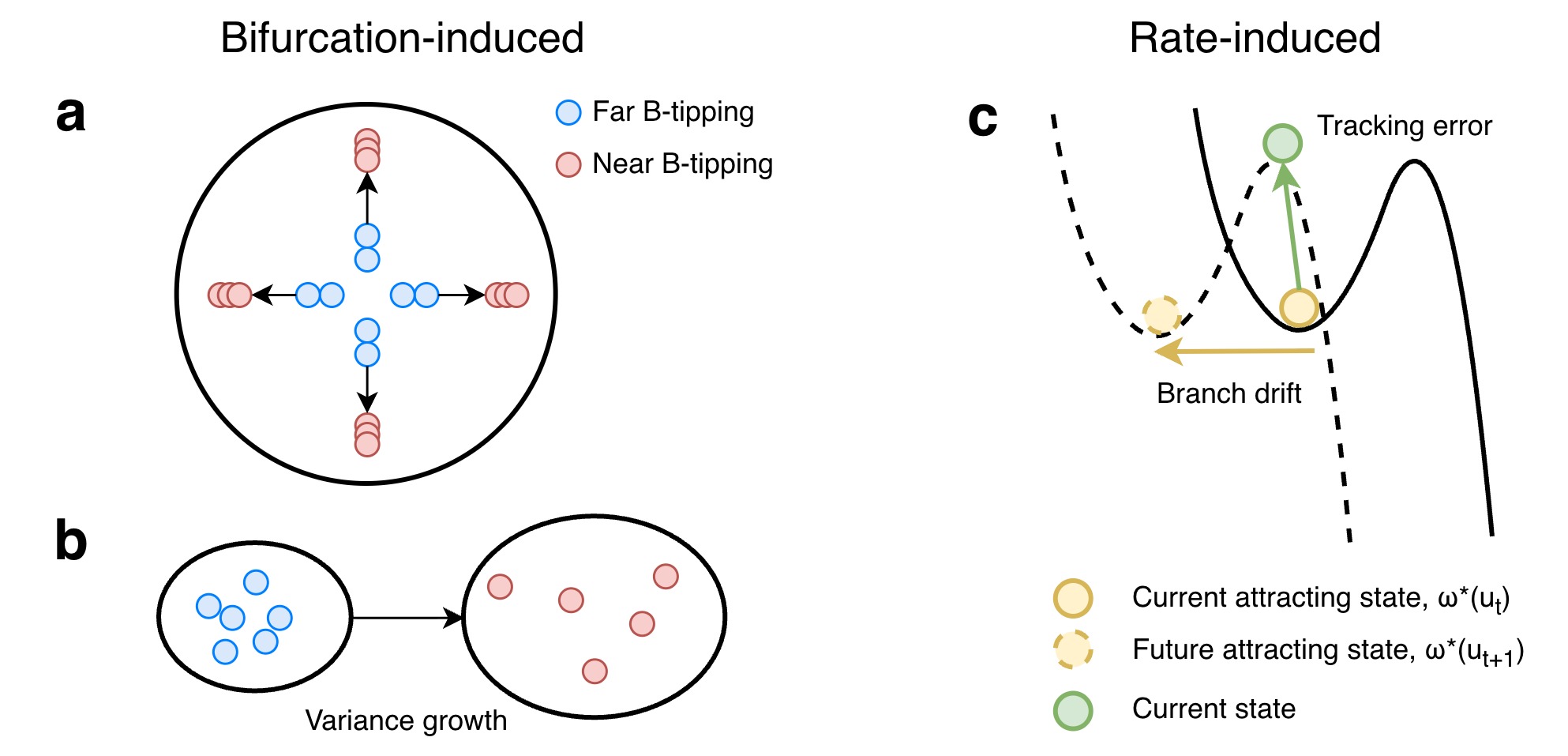}
    \caption{Schematic illustration of Koopman-based early warning mechanisms. (\textbf{a}) In B-tipping, slow Koopman modes move toward the stability boundary as the system approaches critical slowing down, reducing the spectral gap and increasing dynamical memory. (\textbf{b}) Slower recovery causes stochastic perturbations to persist longer, leading to variance growth. (\textbf{c}) In R-tipping, the frozen dynamics can remain locally stable, but the system can fail to track a moving attracting branch. The current and future attracting states, $\omega^\star(u_t)$ and $\omega^\star(u_{t+1})$, indicate how the branch shifts as the control changes. The branch drift $d_t=\omega^\star(u_t)-\omega^\star(u_{t+1})$ measures this shift, while the tracking error $\bar{\omega}_t=\omega_t-\omega^\star(u_t)$ measures the departure of the current state from the instantaneous attracting state. Together, these spectral properties and tracking error mechanisms motivate our unified Koopman-based EWS.}
    \label{fig:graphical_abstract}
\end{figure}

\section*{Generalized Koopman framework for early warning of tipping points}

The central quantity in our framework is the residual obtained through ResKMD~\cite{colbrook2024beyond,colbrook2023residual,mankovich2026disentangling}. This residual measures how accurately the observed dynamics can be represented using a finite set of Koopman modes. Far from tipping, the dynamics are often well described by a compact low-dimensional spectral representation dominated by a small number of Koopman modes. Near tipping, we argue that this representation progressively deteriorates: in the bifurcation setting because slow modes accumulate near the stability boundary, and in the rate-induced setting because the system can no longer track the moving attracting branch under changing external conditions. The Koopman residual therefore acts as a dynamical measure of loss of spectral representability (see Figure~\ref{fig:graphical_abstract}).

\paragraph{Bifurcation-induced setting.} In the autonomous setting relevant to B-tipping, we consider a stochastic dynamical system with state $\omega_t \in \mathbb{R}^N$ and observables $\boldsymbol{\psi}(\omega_t) \in \mathbb{C}^M$ evolved through the first-moment Koopman operator $\mathcal{K}^{(1)}: \mathbb{C}^M \to \mathbb{C}^M$. We approximate the dynamics using a retained Koopman subspace $\mathcal{S}_m^b$. In the bifurcation setting, the Koopman residual separates naturally into two contributions: a truncation component measuring imperfect finite-dimensional Koopman representation, and a stochastic component measuring variability in the observable evolution
\begin{equation}
    \underbrace{\mathrm{res}\!\left[\mathcal{K}^{(1)},\mathcal{K}^{(1)}_{\mathcal{S}_m^b};\boldsymbol{\psi}\right]^2}_{\text{ResKMD}}
    =
    \underbrace{
    \left\|
    \mathcal{K}^{(1)}\boldsymbol{\psi}
    -
    \mathcal{K}^{(1)}_{\mathcal{S}_m^b}\boldsymbol{\psi}
    \right\|^2
    }_{\text{truncation error}}
    +
    \underbrace{
    \int
    \mathrm{Tr}\!\left[
    \mathrm{Cov}\!\left(\boldsymbol{\psi}(\Phi_\varepsilon(\omega))\right)
    \right]\,d\mu(\omega)
    }_{\text{stochastic fluctuation}}.
    \label{eq:main_reskmd_decomp}
\end{equation}

Here $\mu$ denotes a reference measure on state space and $\Phi_\varepsilon$ the one-step stochastic map. Near B-tipping, both contributions are theoretically shown to increase as slow Koopman modes approach the stability boundary and stochastic fluctuations become amplified~\cite{miyauchi2026generalized}. In this sense, ResKMD generalizes classical CSD indicators such as AR1 and variance within a spectral operator-theoretic framework~\cite{grziwotz2023anticipating} (see discussion in Supplementary Information~\ref{si-sec:ar1_derivation}, \ref{si-sec:var_derivation}). 

\paragraph{Rate-induced setting.} The R-tipping setting requires a different interpretation. Here the key issue is not necessarily loss of local stability, but failure of the system to track a moving attracting branch under time-dependent forcing. We therefore augment the Koopman representation with a control variables describing the forcing history~\cite{colbrook2024beyond,colbrook2023residual}. Let $u_t\in\mathbb{R}^U$ denote the control variable evolving according to $u_{t+1}=g(u_t)$, let $\boldsymbol{\Psi}(\omega,u) \in \mathbb{C}^M$ be an observable dictionary on the joint state-control space, and let $\mathcal{K}^{(1)}_c: \mathbb{C}^M \to \mathbb{C}^M$ denote the corresponding first-moment Koopman operator. The resulting control-aware residual retains the same structure as in the bifurcation setting, separating truncation and stochastic contribution, but now in the joint state-control space

\begin{boxprop}{Control-aware Koopman residual decomposition}{reskmdc_decomp}
\begin{equation}
    \underbrace{\mathrm{res}\!\left[\mathcal{K}^{(1)}_c,\mathcal{K}^{(1)}_{\mathcal{S}_m^c};\boldsymbol{\Psi}\right]^2}_{\text{ResKMD control}}
    =
    \underbrace{
    \left\|
    \mathcal{K}^{(1)}_c\boldsymbol{\Psi}
    -
    \mathcal{K}^{(1)}_{\mathcal{S}_m^c}\boldsymbol{\Psi}
    \right\|^2
    }_{\text{truncation error}}
    +
    \underbrace{
    \int
    \mathrm{Tr}\!\left[
    \mathrm{Cov}\!\left(
    \boldsymbol{\Psi}(\Phi_\varepsilon(\omega,u),g(u))
    \right)
    \right]\,d\mu(\omega,u)
    }_{\text{stochastic fluctuation}}.
    \label{eq:main_reskmdc_decomp}
\end{equation}
The proof is provided in Supplementary Information~\ref{si-sec:proof-reskmdc-decomp}. 
\end{boxprop}

In order to make the rate-induced mechanism explicit, we define the tracking error relative to the moving attracting branch (see Figure~\ref{fig:graphical_abstract}c)
\begin{equation}
    \bar{\omega}_t := \underbrace{\omega_t-\omega^\star(u_t)}_{\text{traking error}},
\end{equation}

which measures the deviation of the evolving system state from the instantaneous attracting branch associated with the control variable \(u_t\). In the rate-induced setting, tipping corresponds to growth of this tracking error despite the dynamics remaining locally stable. Linearizing around $\omega^\star(u_t)$ yields (derivation in Supplementary Information~\ref{si-sec:rate_perturb_dynamics})
\begin{equation}
    \bar{\omega}_{t+1} = \underbrace{\mathcal{J}(u_t)\bar{\omega}_t}_{\text{local restoration}} + \underbrace{\omega^\star(u_t)-\omega^\star(u_{t+1})}_{\text{branch drift, }d_t} + \underbrace{\epsilon_t}_{\text{high-order terms}}.
\end{equation}

The first term describes local restoration toward the attracting branch, while the second term $d_t$ denotes the branch drift induced by the changing control variable, and $\epsilon_t$ collects stochastic and higher-order contributions. In this picture, R-tipping occurs when branch drift is not sufficiently compensated by local restoration. 

We assume that the tracking error lies in, or is well approximated by, the span of the chosen observables. Concretely, there exists a bounded linear map $C : \mathbb{C}^M \to \mathbb{R}^N$ relating observables to the tracking error $\bar{\omega}_t$. Motivated by the tracking error dynamics above, our main theoretical result connects increased control-aware residuals to unresolved branch drift.

\begin{boxtheorem}{Residual lower bound from unresolved branch drift}{rtip_lower_bound}

Assume that the local restoring dynamics are captured by the retained Koopman subspace, then the control-aware residual satisfies:

\begin{equation}
    \mathrm{res}\!\left[\mathcal{K}^{(1)}_c,\mathcal{K}^{(1)}_{\mathcal{S}_m^c};\boldsymbol{\Psi}\right]
    \ge
    \frac{1}{\|C\|}
    \left(
    \|\mathcal{Q}_m d\| - \|\mathcal{Q}_m \epsilon\|
    \right),
    \label{eq:main_drift_lower_bound}
\end{equation}

where $\mathcal{Q}_m d$ and $\mathcal{Q}_m \epsilon$ denote the components of the branch drift and high-order terms unresolved by the retained Koopman representation. Proof is provided in Supplementary Information~\ref{si-sec:proof-rtip_lower_bound}.
\end{boxtheorem}

Theorem~\ref{thm:rtip_lower_bound} shows that the control-aware residual increases when branch drift, $d_t$, contains components unresolved by the retained Koopman representation. Thus, unlike B-tipping, where residual growth reflects spectral slowing near instability, increased residuals in R-tipping reflect loss of tracking under time-dependent forcing. Nonetheless, unresolved branch drift is not the only possible source of increased residuals. Additional unresolved dynamics through e.g., local restoring dynamics not fully captured by the retained Koopman subspace, and stochastic fluctuations may also contribute. Theorem~\ref{thm:rtip_lower_bound} instead identifies unresolved branch drift as a minimal mechanism through which rate-induced tracking failure contributes to residual growth.

In summary, both of these residual quantities in Equations~\ref{eq:main_reskmd_decomp} and~\ref{eq:main_reskmdc_decomp} will be used as our Koopman-based EWS in B- and R-tipping, respectively. See Methods and Supplementary Information for complete derivation and formalism.

\section*{Koopman EWS skills on prototypical systems}\label{sec:experiments}

We first benchmark our proposed Koopman early warning indicators on four prototypical stochastic systems spanning both B- and R-tipping before applying the framework to AMOC simulation data. Aside from traditional indicators such as AR1 and variance, we also plot the Koopman spectral gap, defined as $\gamma := -\mathrm{Re}(\lambda_{L,1})$, where $\lambda_{L,1}$ is the leading non-zero Koopman generator eigenvalue, i.e., the non-zero eigenvalue with real part closest to zero. Near CSD, $\gamma$ is expected to decrease toward zero~\cite{    Tantet2018}. For B-tipping, we consider May’s harvesting model and the Rosenzweig–MacArthur consumer–resource model, which exhibit fold and Hopf bifurcations, respectively~\cite{bury2021deep}. For R-tipping, we consider a forced saddle-node normal form and a Bautin oscillator, where transitions arise from failure to track a moving attracting branch under time-dependent forcing~\cite{huang2024deep}. Additional system details, parameter choices, and implementation procedures are provided in Methods and Supplementary Information.

\begin{figure}[h!]
    \centering
    \begin{subfigure}[h]{\linewidth}
        \centering
        \includegraphics[width=\linewidth]{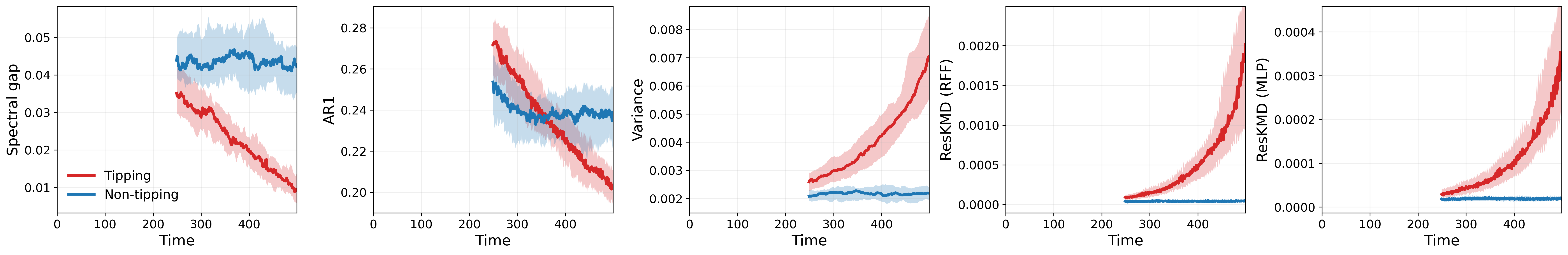}
        \caption{B-tipping: Rosenzweig-MacArthur}
        \label{fig:bifurcation_ews_rm}
    \end{subfigure}
    \hfill
    \begin{subfigure}[h]{\linewidth}
        \centering
        \includegraphics[width=\linewidth]{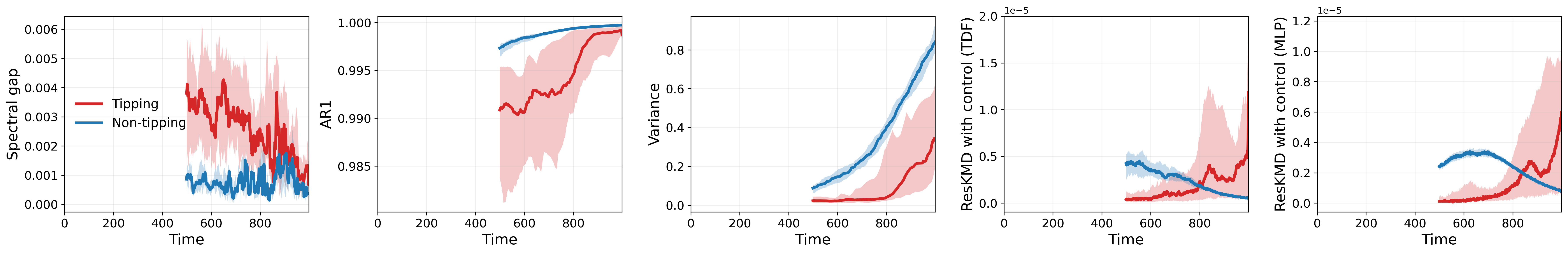}
        \caption{R-tipping: Bautin oscillator}
        \label{fig:rate_ews_bautin}
    \end{subfigure}
    \caption{Koopman EWS benchmarked against traditional metrics across tipping scenarios: In (\textbf{a}) the Rosenzweig-MacArthur model, panels show the ensemble behavior of five indicators computed on pre-transition sliding windows of length $50\%$: Spectral gap, AR1, variance, ResKMD (RFF), and ResKMD (MLP observable with a hidden size of (16, 8, 4)). In (\textbf{b}) the Bautin oscillator, panels show the ensemble behavior of five indicators computed on pre-transition sliding windows: Spectral gap, AR1, and control-aware ResKMD (TDF, MLP observable with a hidden size of (8, 4, 2)). The shaded region denotes the interquantile range around the median.}  
    \label{fig:bifurcation_ews}
\end{figure}

Figure~\ref{fig:bifurcation_ews} (see also Supplementary Section~\ref{si-sec:proto_exp} for additional results) shows representative results for the Rosenzweig–MacArthur model (B-tipping) and the Bautin oscillator (R-tipping). In the bifurcation setting, Koopman residuals increase more clearly for tipping trajectories than for non-tipping trajectories. This is consistent with the theoretical mechanism: as the system approaches the bifurcation, slow Koopman modes move toward the stability boundary, finite-dimensional spectral representations deteriorate, and stochastic fluctuations become amplified along slowly decaying directions. Learned observables obtained from MLP embeddings further improve the separation between tipping and non-tipping trajectories relative to prescribed dictionaries such as random Fourier features (RFF), consistent with the quantitative AUROC scores shown in Figure~\ref{fig:auroc}.

The rate-induced setting reveals a qualitatively different picture. Classical CSD indicators show inconsistent separation between tipping and non-tipping trajectories. For instance, AR1 and variance increase in both tipping and non-tipping cases, while the spectral gap approaches zero near the critical transition in both groups. The control-aware Koopman residual, in contrast, increases substantially more for tipping trajectories, consistent with the unresolved branch drift mechanism described earlier. This supports the interpretation that, in R-tipping, the dominant precursor is loss of tracking rather than CSD, and is also consistent with the quantitative AUROC scores shown in Figure~\ref{fig:auroc}.

\begin{figure}[h!]
    \centering
    \begin{subfigure}[h]{0.7\linewidth}
        \centering
        \includegraphics[width=\linewidth]{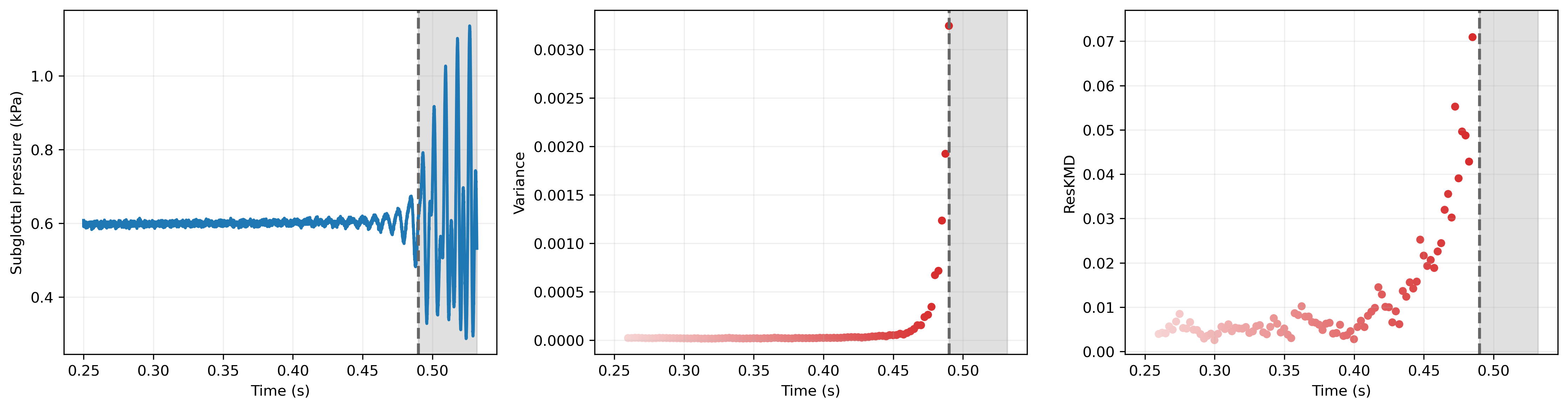}
        \caption{Voice onset}
        \label{fig:empirical_voice}
    \end{subfigure}
    \hfill
    \begin{subfigure}[h]{0.7\linewidth}
        \centering
        \includegraphics[width=\linewidth]{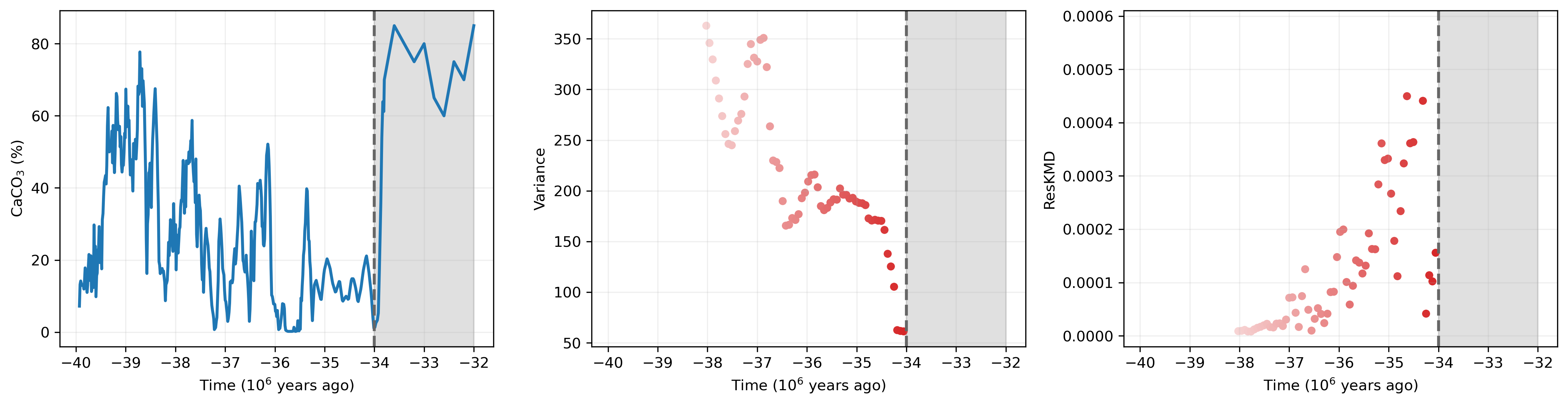}
        \caption{Paleoclimate}
        \label{fig:empirical_climate}
    \end{subfigure}
    \hfill
    \begin{subfigure}[h]{0.7\linewidth}
        \centering
        \includegraphics[width=\linewidth]{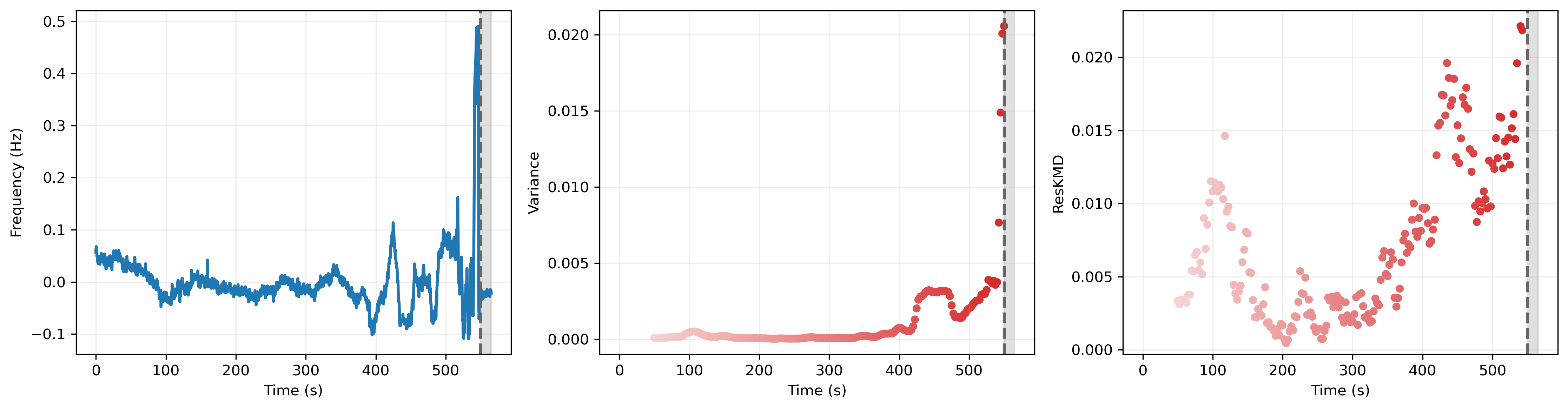}
        \caption{Electricity blackout}
        \label{fig:empirical_blackout}
    \end{subfigure}
    \caption{Koopman EWS across empirical systems: (\textbf{a}) vocal phonation exhibiting fold transition, (\textbf{b}) paleoclimate time series of calcium carbonate associated with the end of greenhouse Earth, and (\textbf{c}) electricity blackout at Western Interconnect in 1996. The grey vertical dashed line indicates identified critical transitions. We find that our Koopman EWS indicator is able to detect critical transitions earlier (e.g., voice onset and electricity blackout) and more accurately (e.g., paleoclimate) than classical CSD metrics, such as variance. For all indicators, lighter dots represent earlier time.}
    \label{fig:b_empirical_ews}
\end{figure}

\section*{Koopman EWS skills on empirical observations}
We further evaluate our Koopman EWS on a diverse collection of real-world observations spanning human and natural systems. These datasets include cyanobacteria recovery dynamics in a microcosm experiment~\cite{veraart2012recovery}, phonation onset of vocal cord~\cite{mergell1998phonation}, ATP collapse under hypoxia in plant cells~\cite{wagner2019multiparametric}, paleoclimate transitions associated with the end of greenhouse Earth~\cite{dakos2008slowing}, and electricity blackout~\cite{council1996western}. Together, these systems probe a broad range of temporal scales, noise characteristics, and underlying transition mechanisms beyond idealized dynamical models. Additional dataset descriptions and details are provided in the Methods section.

Figure~\ref{fig:b_empirical_ews} (see also Figure~\ref{si-fig:b_empirical_ews} for additional results) shows that the Koopman-based indicators remain informative across these heterogeneous empirical settings. In systems associated with classical CSD, such as the microcosm and paleoclimate examples, the Koopman residual exhibits increasing trends consistent with growing dynamical memory and spectral slowing prior to transition. In more complex or noisy settings, including the blackout and phonation examples, the residual-based indicators remain sensitive to changes in the underlying dynamical structure even when conventional indicators appear less consistent. Overall, these results suggest that the Koopman framework can extract informative dynamical signatures from empirical observations across a wide range of domains and observational conditions.

\section*{Tipping of the Atlantic Meridional Overturning Circulation}

The Atlantic Meridional Overturning Circulation (AMOC) is a major component of the climate system that transports heat and salinity through the Atlantic Ocean. Its weakening or collapse has long been discussed as a potential climate tipping phenomenon because it can reorganize large-scale ocean circulation and strongly affect regional and global climate~\cite{wunderling2024climate,lenton2008tipping}. We therefore use AMOC as a physically relevant testbed for our Koopman framework in both the reduced box-model setting and coupled global circulation model (GCM).

\begin{figure}[h!]
    \centering
    \begin{subfigure}[h]{0.8\linewidth}
        \centering
        \includegraphics[width=\linewidth]{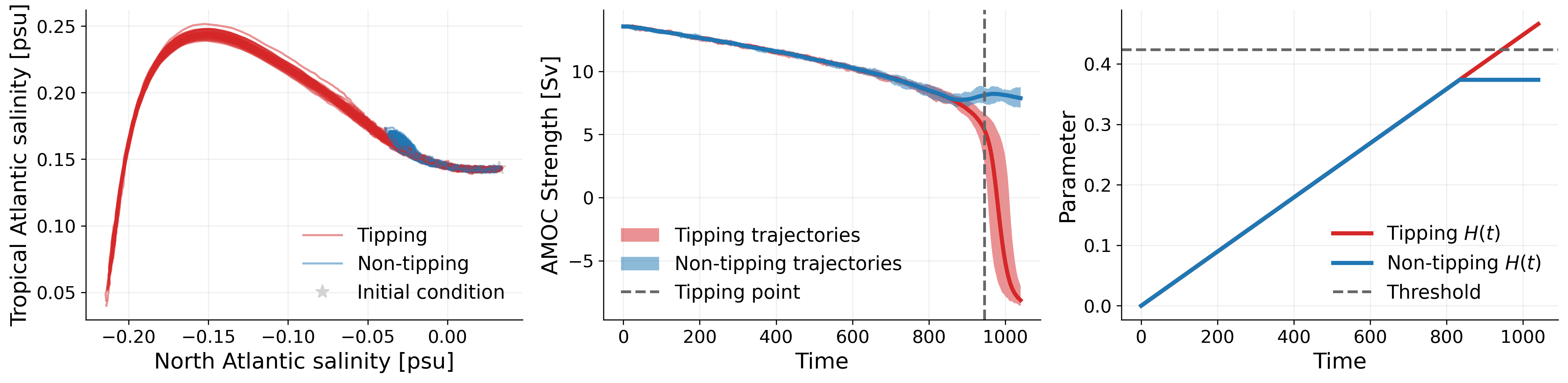}
        \caption{AMOC box-model (B-tipping)}
        \label{fig:bifurcation_examples_amoc3b}
    \end{subfigure}
    \hfill
    \begin{subfigure}[h]{0.8\linewidth}
        \centering
        \includegraphics[width=\linewidth]{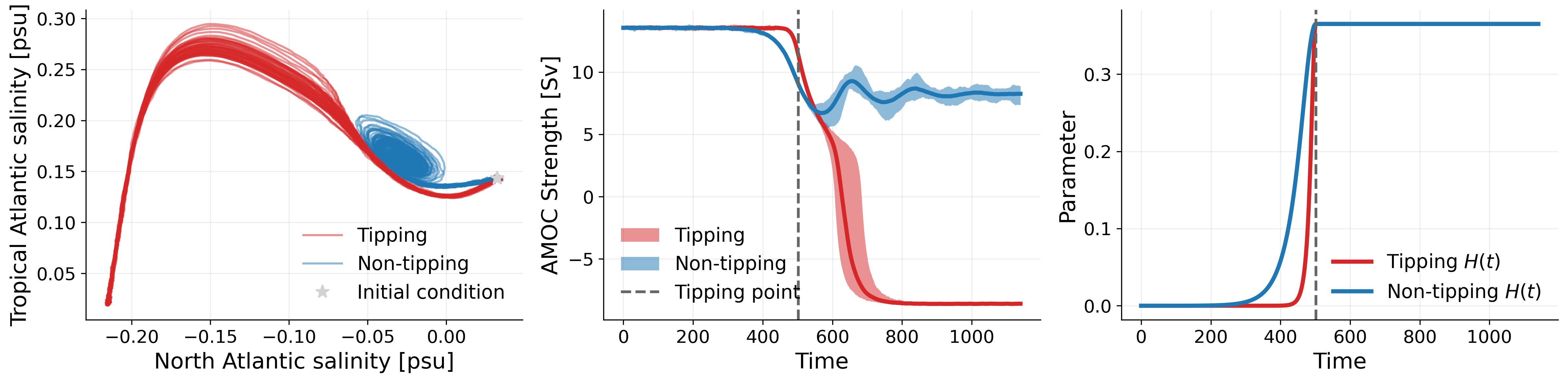}
        \caption{AMOC box-model (R-tipping)}
        \label{fig:rate_examples_amoc3b}
    \end{subfigure}
    \caption{Idealized AMOC box-model trajectories. (\textbf{a}) B-tipping and (\textbf{b}) R-tipping. Each panel shows \textbf{left}: phase-space evolution in $(S_T,S_N)$, \textbf{center}: AMOC strength $Q$, and \textbf{right}: the corresponding hosing or control schedule.}
    \label{fig:amoc3box}
\end{figure}

For the reduced setting (see Figure~\ref{fig:amoc3box} for sample realizations), we adopt the two-dimensional box-model from \cite{ritchie2023rate,alkhayuon2019basin}, obtained from a five-box ocean model by treating Southern Ocean and bottom water salinities as slow variables and diagnosing Indo-Pacific salinity from salt conservation. The resulting state variables are the North Atlantic and tropical Atlantic salinities $(S_N, S_T)$, from which the AMOC strength $Q$ is diagnosed (see Methods Equation~\ref{eq:amoc_strength}) and supplied as the scalar observable to all early warning estimators. The salinity dynamics are piecewise-defined according to the sign of $Q$ (see Methods Equations~\ref{eq:amoc_sde_posq_SN}–\ref{eq:amoc_sde_negq_ST}) and driven by surface freshwater hosing $H(r,t)$. Departing from the original deterministic formulation, we add white noise with amplitude $\sigma_N=\sigma_T=1.0$ to both salinity equations. We configure the model for both tipping classes: in the B-tipping setup, $H$ is ramped linearly past the deterministic fold at $H^\star = 0.4236$ \cite{alkhayuon2019basin}, with non-tipping runs capped just below threshold; in the R-tipping setup, tipping and non-tipping trajectories follow an identical hyperbolic-secant forcing profile (see Methods Equation~\ref{eq:amoc_rate_forcing}) differing only in the forcing rate $r$.

\begin{figure}[h!]
    \centering
    \begin{subfigure}[h]{\linewidth}
        \centering
        \includegraphics[width=\linewidth]{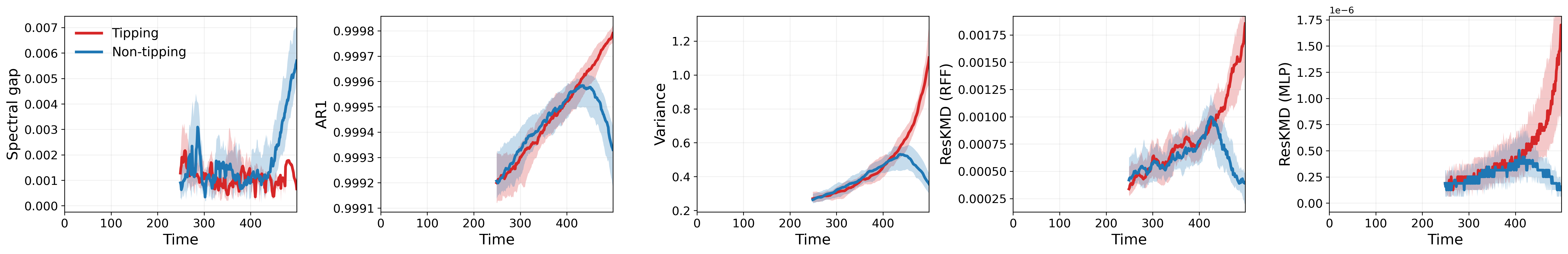}
        \caption{AMOC box-model (B-tipping)}
        \label{fig:bifurcation_ews_amoc3b}
    \end{subfigure}
    \hfill
    \begin{subfigure}[h]{\linewidth}
        \centering
        \includegraphics[width=\linewidth]{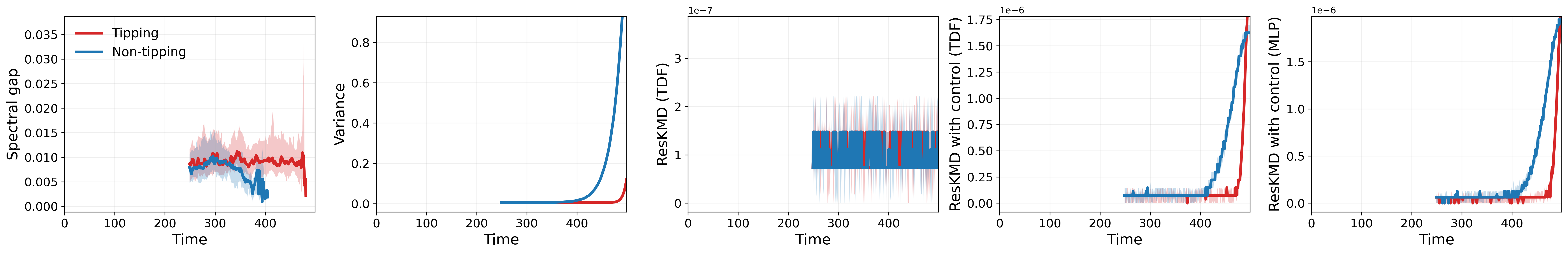}
        \caption{AMOC box-model (R-tipping)}
        \label{fig:rate_ews_amoc3b}
    \end{subfigure}
    \caption{EWS for the B- and R-tipping setups in AMOC box-model. For each setup, the panels show the ensemble behavior of five indicators computed on pre-transition sliding windows of length $50\%$. (\textbf{a}) Spectral gap, AR1, variance, ResKMD (RFF), and ResKMD (MLP observable with a hidden size of (8, 4, 2)). (\textbf{b}) Spectral gap, variance, ResKMD (TDF), and control-aware ResKMD (TDF, MLP observable with a hidden size of (8, 4, 2)). The shaded region denotes the interquantile range around the median.}
    \label{fig:amoc3box_ews}
\end{figure}

As illustrated in Figure~\ref{fig:amoc3box_ews}, the results from AMOC box-model mirror the patterns seen in the prototypical systems. In the B-tipping setting, AR1, variance and ResKMD (RFF and MLP) all exhibit increasing trend prior to transition, consistent with CSD phenomena. In the R-tipping setting, by contrast, classical indicators fail to provide useful early warning signal as the non-tipping trajectories tend to be falsely characterized as tipping (e.g., increasing variance, diminishing spectral gap). The control-aware ResKMD, however, provides a better tipping signal as the indicator abruptly rises near the tipping threshold at a much faster rate than their CSD-based counterpart. This is also evident in our quantitative AUROC measures in Figure~\ref{fig:auroc} where control-aware ResKMD outperforms across the board.

We next consider the 10-member NASA GISS-E2.1-G ensemble studied in \cite{romanou2023stochastic} under the SSP2-4.5 scenario. The ensemble exhibits branching behavior in AMOC strength at $48^\circ$N: eight members remain on a relatively strong overturning branch, whereas two evolve toward a markedly weaker state. 

\begin{figure}[h!]
    \centering
    \begin{subfigure}[h]{\linewidth}
    \includegraphics[width=\linewidth]{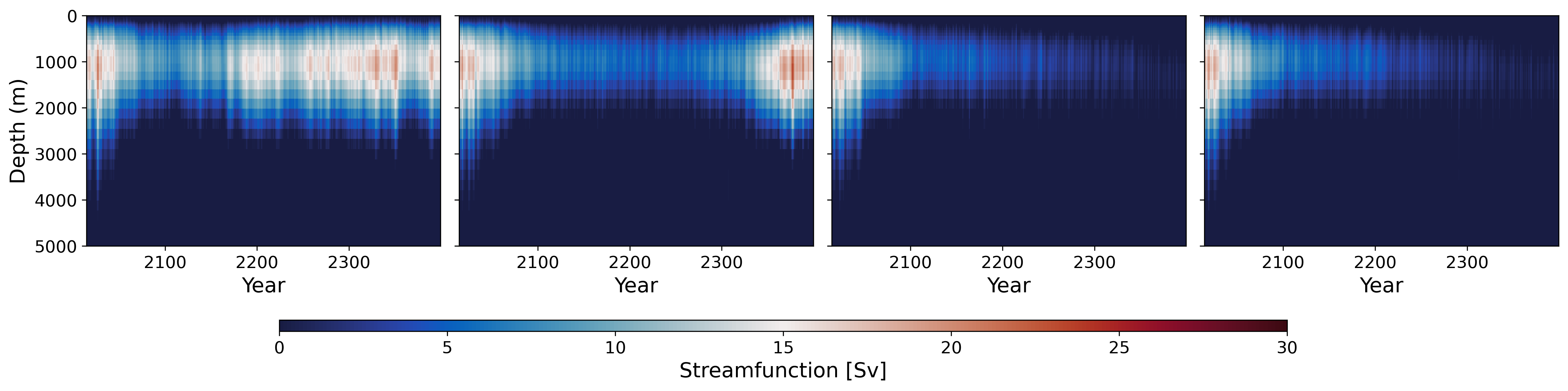}
    \caption{Depth profiles of overturning streamfunction}
    \label{fig:amoc_profile}
    \end{subfigure}
    \begin{subfigure}[h]{0.48\linewidth}
        \centering
        \includegraphics[width=\linewidth]{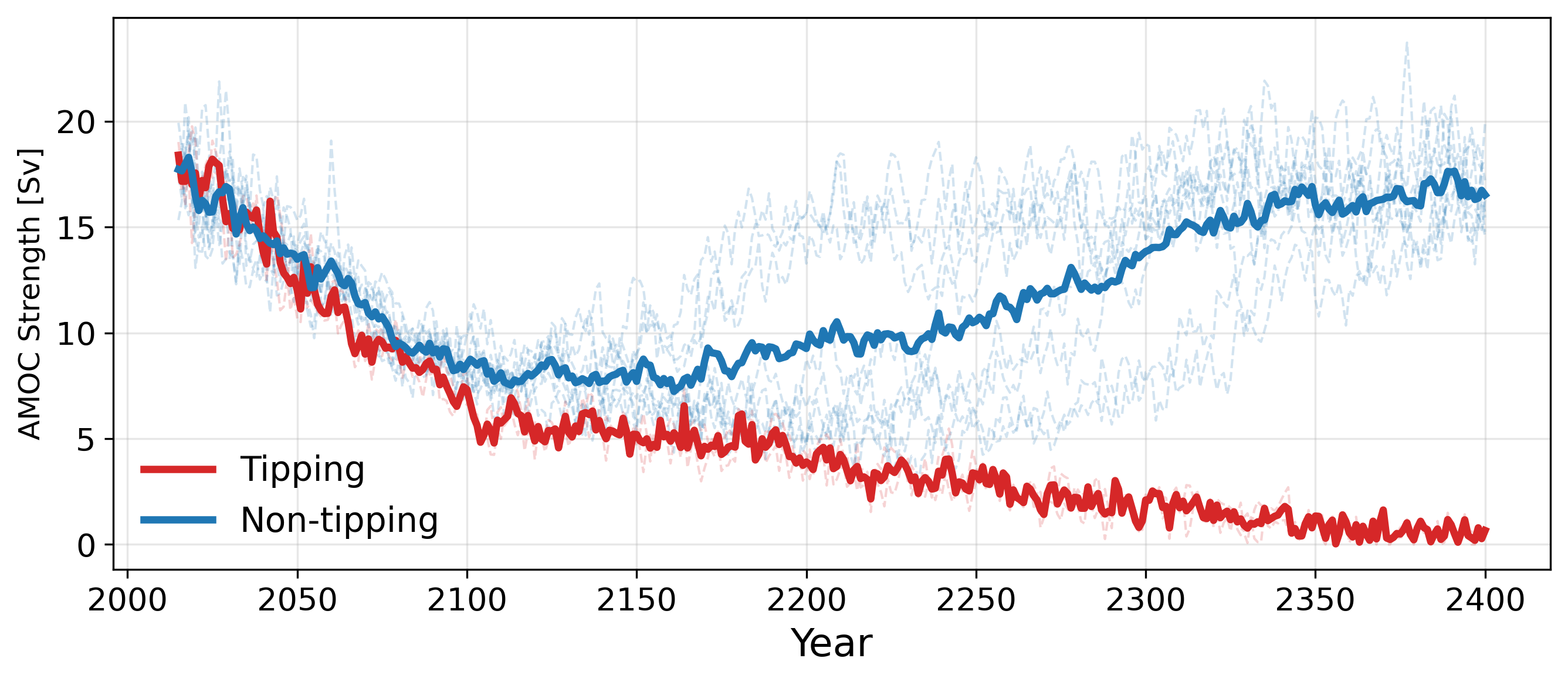}
        \caption{AMOC strength}
        \label{fig:amoc_amoc48N_ssp245_traj}
    \end{subfigure}
    \hfill
    \begin{subfigure}[h]{0.48\linewidth}
        \centering
        \includegraphics[width=\linewidth]{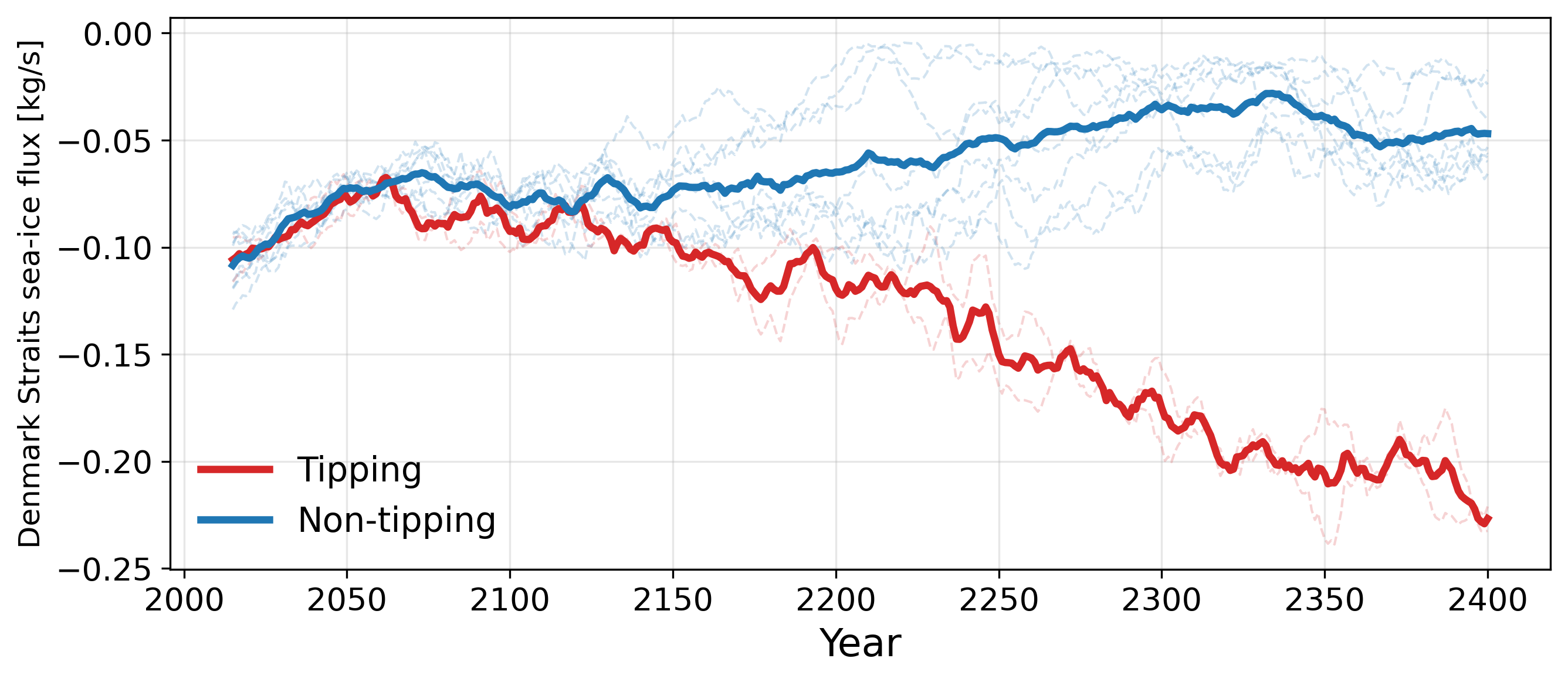}
        \caption{Denmark Strait sea-ice flux}
        \label{fig:control_si_flux_denmarkstr_ssp245_traj}
    \end{subfigure}
    \caption{AMOC tipping (red) and non-tipping (blue) runs at $48^\circ$N. (\textbf{a}) Depth profiles of Atlantic overturning streamfunction for representative coupled AMOC trajectories: the first two panels show non-tipping runs that recover to the stronger branch, while the latter two show tipping runs that evolve toward a weaker circulation state. (\textbf{b}) AMOC strength trajectories and (\textbf{c}) Denmark Strait sea-ice flux used as the control variable in the control-aware Koopman analysis. Light lines denote individual ensemble members and thick lines denote ensemble means. AMOC strength is computed as the maximum overturning streamfunction below 500\,m, and Denmark Strait sea-ice flux is smoothed with a 10-year running mean.}
    \label{fig:amoc_control_traj}
\end{figure}

Figure~\ref{fig:amoc_profile} illustrates this separation using depth profiles of the Atlantic overturning streamfunction, with non-tipping realizations retaining a stronger overturning cell and tipping realizations showing pronounced weakening, consistent with the AMOC strength trajectories in Figure~\ref{fig:amoc_amoc48N_ssp245_traj}. As a physically motivated control variable, we use annual Denmark Strait sea-ice mass flux, which \cite{romanou2023stochastic} linked to the weakening branch through enhanced freshwater release (Figure~\ref{fig:control_si_flux_denmarkstr_ssp245_traj}). For this study, we treat the weak-branch members as tipping trajectories and the remainder as non-tipping trajectories ($26^\circ$N results are reported in Supplementary Information~\ref{si-sec:amoc_26N}). Before supplying this control variable to the control-aware ResKMD, we smooth it using a 10-year running mean. Sensitivity to alternative control variables, such as maximum mixed layer depth (MLD) across different ocean basins, is reported in Supplementary Information~\ref{si-sec:ablate_amoc_forcing}.

Figure~\ref{fig:amoc_ews_48} shows that the control-aware Koopman indicators provide the clearest separation between tipping and non-tipping AMOC trajectories in the coupled NASA GISS ModelE ensemble at $48^\circ$N. Classical CSD-based metrics, especially AR1, remain highly variable across ensemble members and do not show a consistent increase prior to trajectory divergence. Variance exhibits a smoother trend, but still provides weaker separation. By contrast, control-aware ResKMD, particularly with MLP observables, increases more strongly and earlier for the tipping members, with a marked rise around year 2150, while remaining comparatively small for the non-tipping members. In the MLP case, the tipping trajectories eventually decrease and converge after the transition, consistent with relaxation toward a different attracting state. These results suggest that, in the coupled setting, the relevant precursor is not fully captured by indicators based on the AMOC time series alone. Instead, it is more clearly expressed as a mismatch between AMOC evolution and the supplied control variable in the augmented Koopman space. This interpretation is consistent with the idea proposed in \cite{Borner2025} that the strong trajectory separation observed in \cite{romanou2023stochastic} may be associated with a close encounter with an edge state, which acts as a dynamical saddle and reduces predictability; see also \cite{Lohmann2024}.

\begin{figure}[h!]
    \centering
    \begin{subfigure}[h]{\linewidth}
        \includegraphics[width=\linewidth]{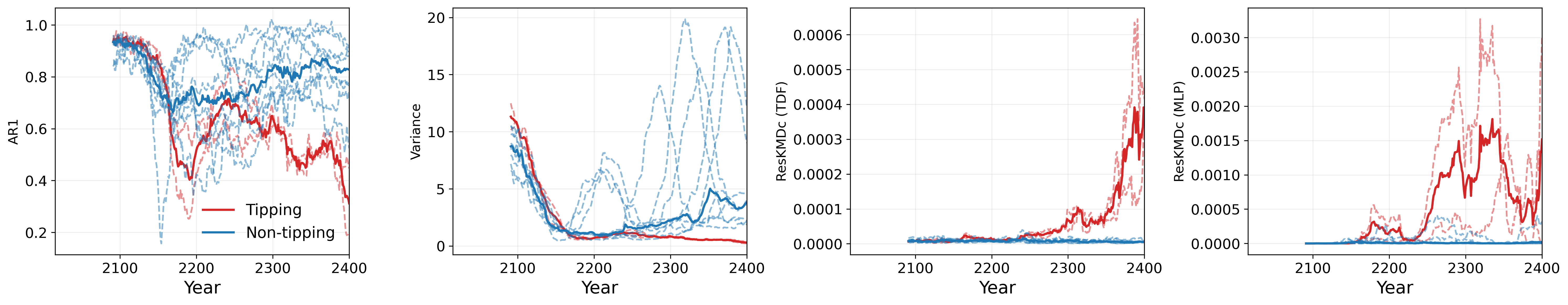}
        \caption{EWS}
        \label{fig:amoc_ews_48}
    \end{subfigure}
    \hfill
    \begin{subfigure}[h]{0.255\linewidth}
        \centering
        \includegraphics[width=\linewidth]{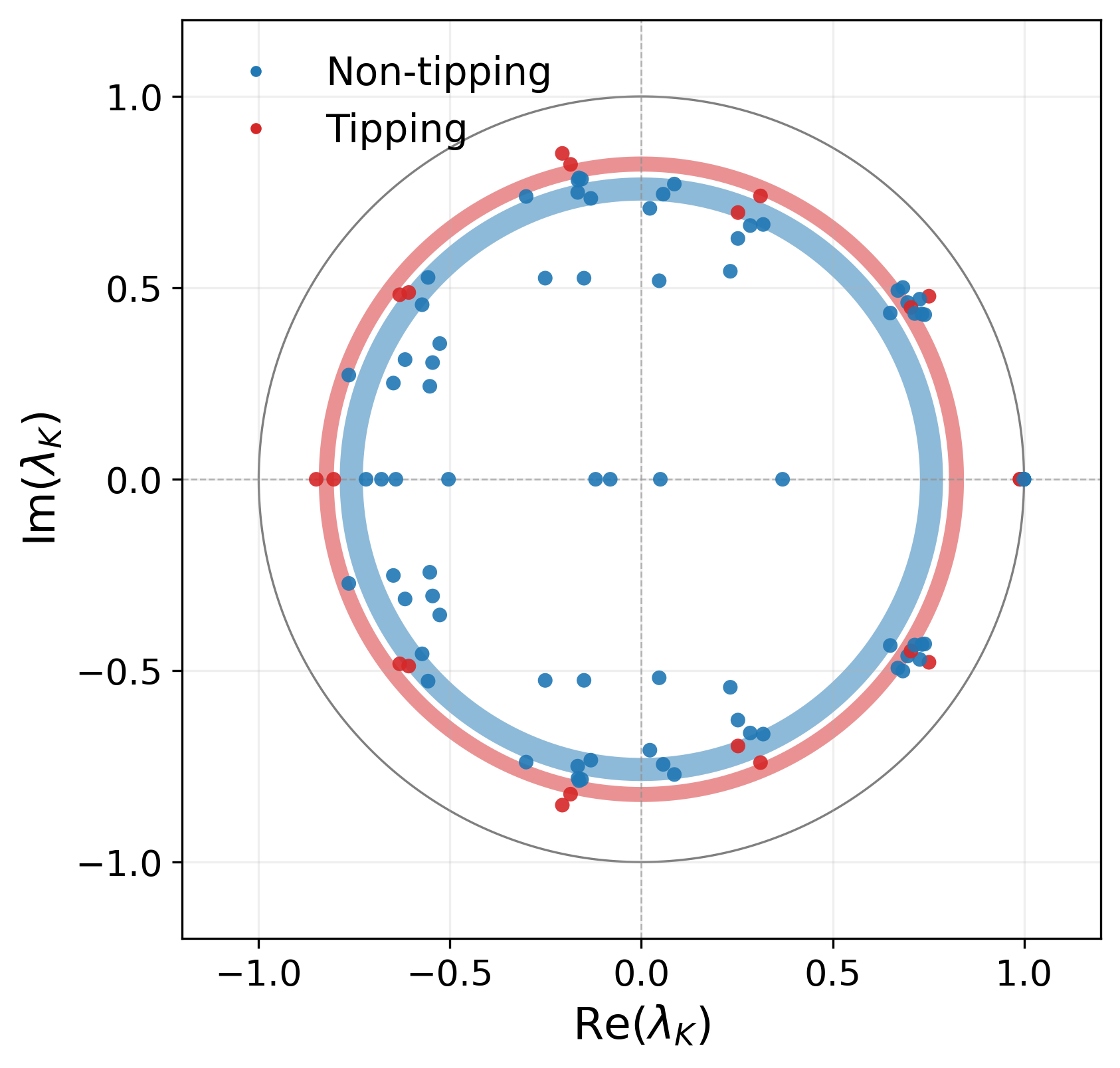}
        \caption{Discrete spectrum}
        \label{fig:amoc_amoc48N_ssp245_disc_spec}
    \end{subfigure}
    \hfill
    \begin{subfigure}[h]{0.240\linewidth}
        \centering
        \includegraphics[width=\linewidth]{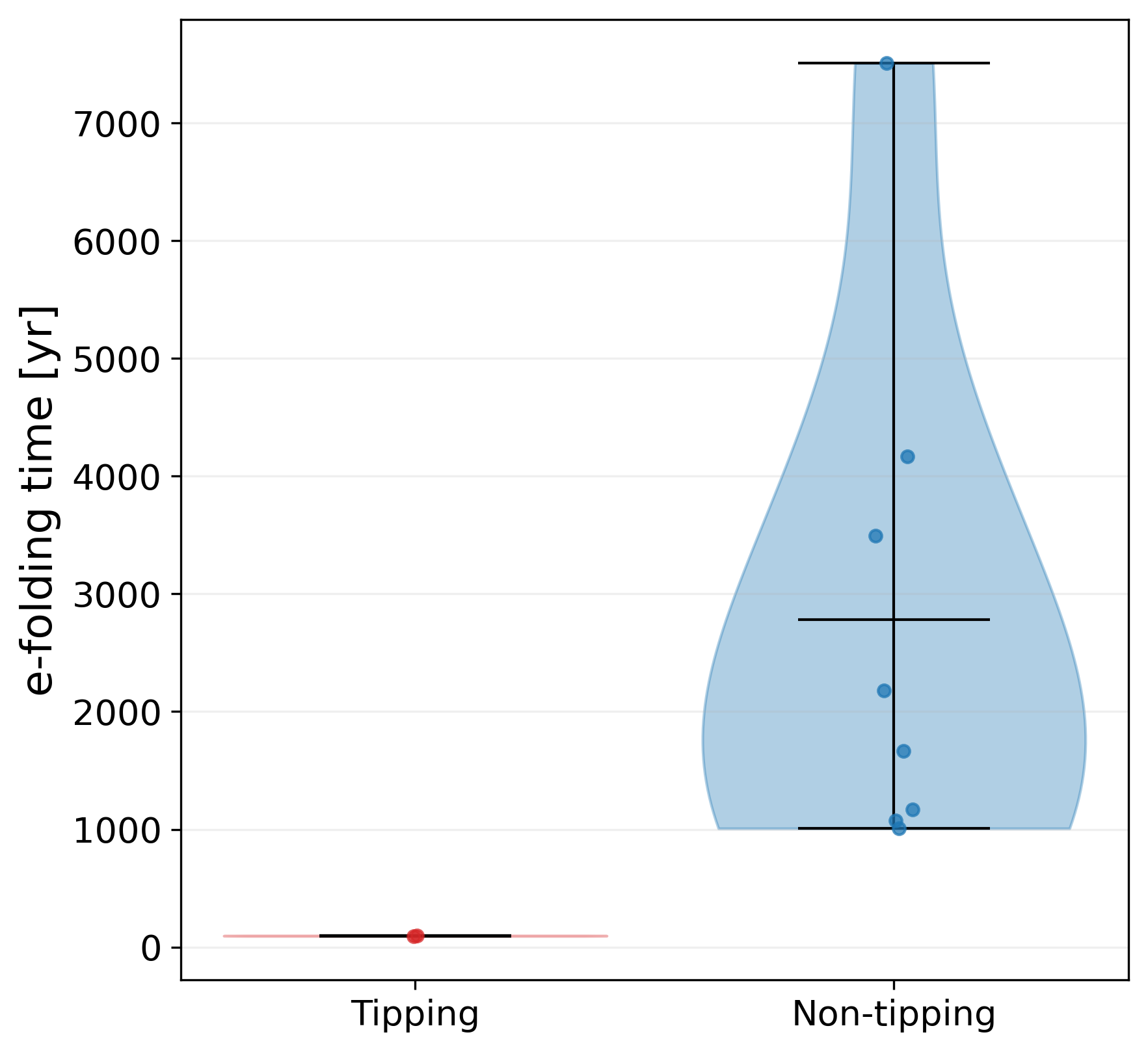}
        \caption{E-folding years}
        \label{fig:amoc_amoc48N_ssp245_efold}
    \end{subfigure}
    \hfill
    \begin{subfigure}[h]{0.41\linewidth}
        \centering
        \includegraphics[width=\linewidth]{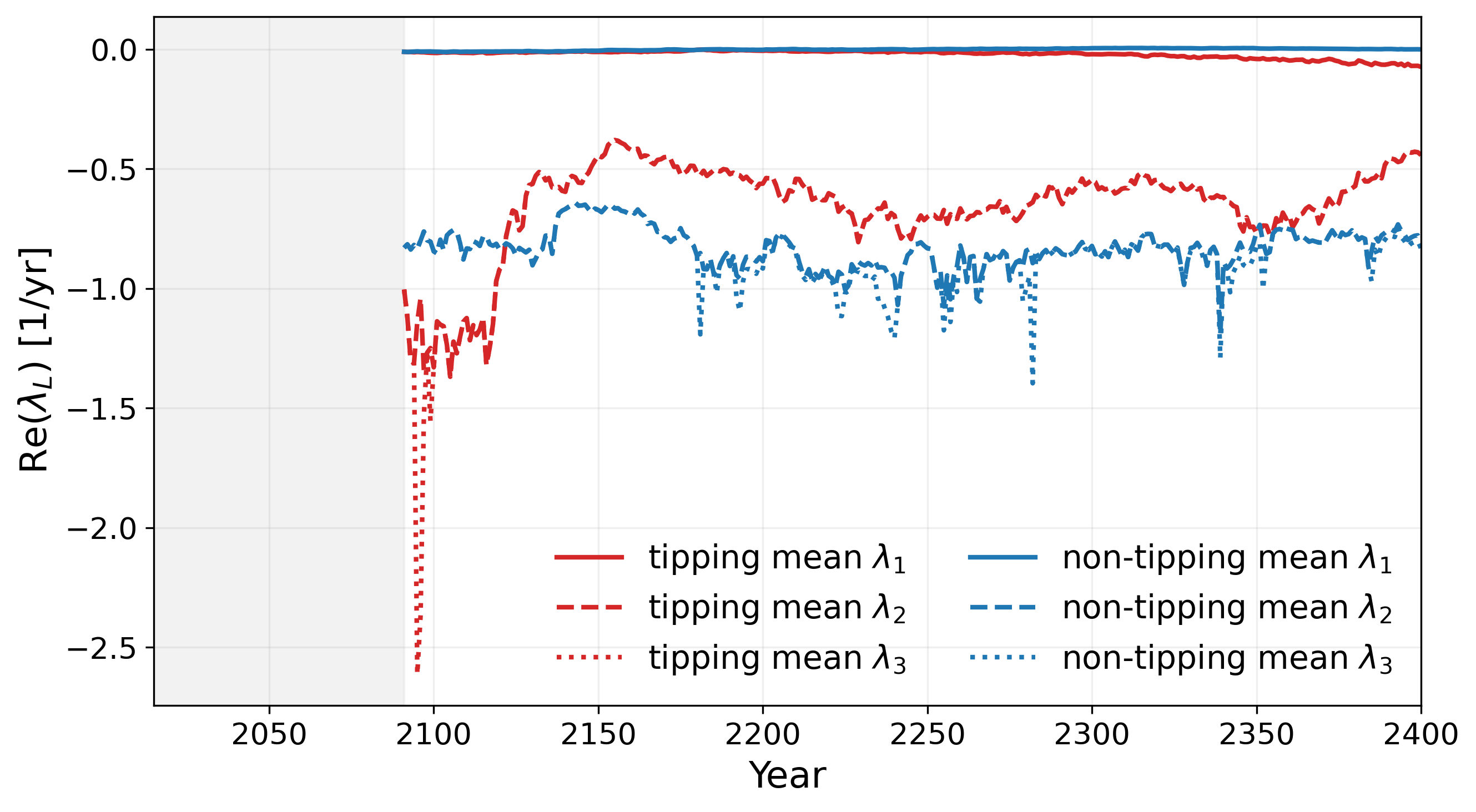}
        \caption{Rolling $\lambda_L$}
        \label{fig:amoc_amoc48N_ssp245_gap}
    \end{subfigure}
    \caption{Koopman EWS and spectral diagnostics for the coupled AMOC ensemble at $48^\circ$N. \textbf{(a)} EWS for the coupled AMOC ensemble at $48^\circ$N. The panels show the ensemble behaviour of AR1, variance, and control-aware ResKMD using TDF with delay size of 20 and MLP observables with a hidden size of of (16, 8, 4). \textbf{(b)} discrete-time Koopman eigenvalues $\lambda_K$ in the unit disk with colored annulus representing $40-60\%$-percentile distribution. \textbf{(c)} e-folding times $T_e = 1/|\mathrm{Re}(\lambda_L)|$ of the slowest decaying non-constant generator mode for each ensemble member. \textbf{(d)} rolling group means of the three generator eigenvalues with smallest $|\mathrm{Re}(\lambda_L)|$, obtained from EDMD fits on rolling windows of length $20\%$ of the full record. Red denotes tipping members and blue denotes non-tipping members.}
    \label{fig:amoc_spectral_48}
\end{figure}

Beyond early warning detection, our Koopman framework provides access to interpretable spectral diagnostics. Here we focus on the $48^\circ$N results using the TDF observables setting (see Figure~\ref{si-fig:amoc_26N} for similar results in $28^\circ$N). We first examine the distribution of the discrete-time Koopman eigenvalues $\lambda_{K,j}\in\mathbb{C}$ (Figure~\ref{fig:amoc_amoc48N_ssp245_disc_spec}). The tipping ensemble exhibits a broader and more outward radial distribution, with eigenvalues lying closer to the unit circle. Since the decay of discrete-time Koopman modes is governed by $|\lambda_{K,j}|$, this indicates slower decay, weaker damping, and thus reduced restoring stability of the stronger AMOC branch in the tipping ensemble.

To obtain physically interpretable decay rates and oscillation frequencies, we map to the continuous-time generator spectrum through $\lambda_{K,j}=e^{\lambda_{L,j}\Delta t}=e^{(\sigma_j+i\omega_j)\Delta t}$, where $\sigma_j=\mathrm{Re}(\lambda_{L,j})$ is the exponential decay rate, $\omega_j=\mathrm{Im}(\lambda_{L,j})$ is the angular frequency, and $\Delta t$ has annual resolution. From this we compute the e-folding time $T_{e,j}=1/|\sigma_j|$. We report the dominant $T_e$ associated with the slowest-decaying non-constant mode. As illustrated in Figure~\ref{fig:amoc_amoc48N_ssp245_efold}, the non-tipping ensemble tends to exhibit larger $T_e$, suggesting a more persistent slow mode associated with an eventual recovery toward the stronger AMOC branch rather than collapse.

Finally, we examine spectral gap behavior by repeating the EDMD fit in a rolling window of length $20\%$ of the available record and retaining, within each window, the three modes with smallest $|\mathrm{Re}(\lambda_{L,j})|$. Figure~\ref{fig:amoc_amoc48N_ssp245_gap} shows that the tipping ensemble tends to exhibit slow modes that move more quickly toward the stability boundary $\mathrm{Re}(\lambda_{L,j})=0$ and remain closer to zero than in the non-tipping ensemble. This indicates spectral densification and a reduced spectral gap in the tipping case, consistent with a weakening of the isolated point-spectrum description. It is also consistent with the larger residuals, which indicate that the dynamics are less well captured by a small number of isolated Koopman modes prior to AMOC weakening.

\begin{figure}[h!]
    \centering
    \begin{subfigure}[h]{0.3\linewidth}
        \includegraphics[width=\linewidth]{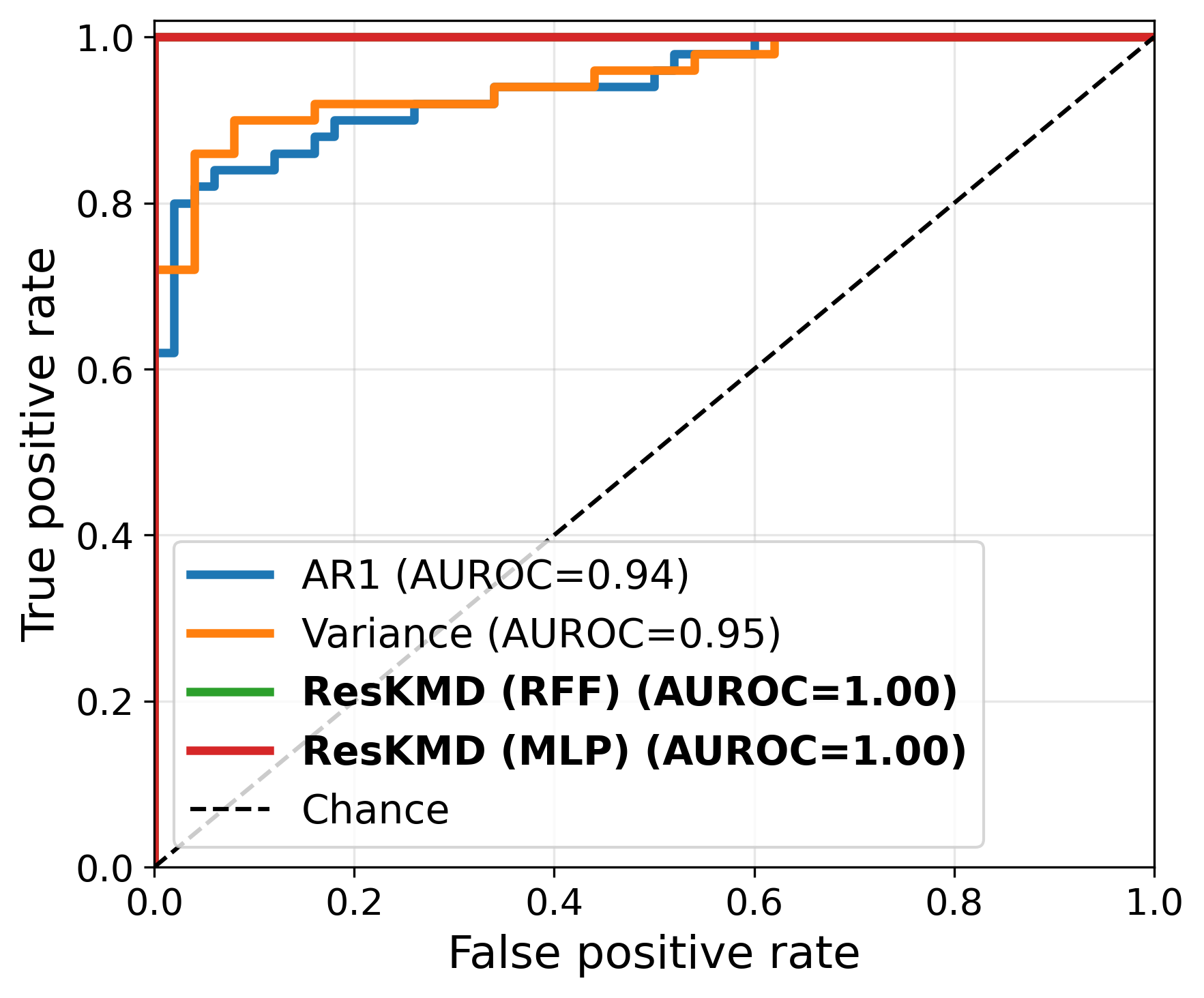}
        \caption{May's harvesting}
        \label{fig:b_harvesting_roc}
    \end{subfigure}
    \hfill
    \begin{subfigure}[h]{0.3\linewidth}
        \centering
        \includegraphics[width=\linewidth]{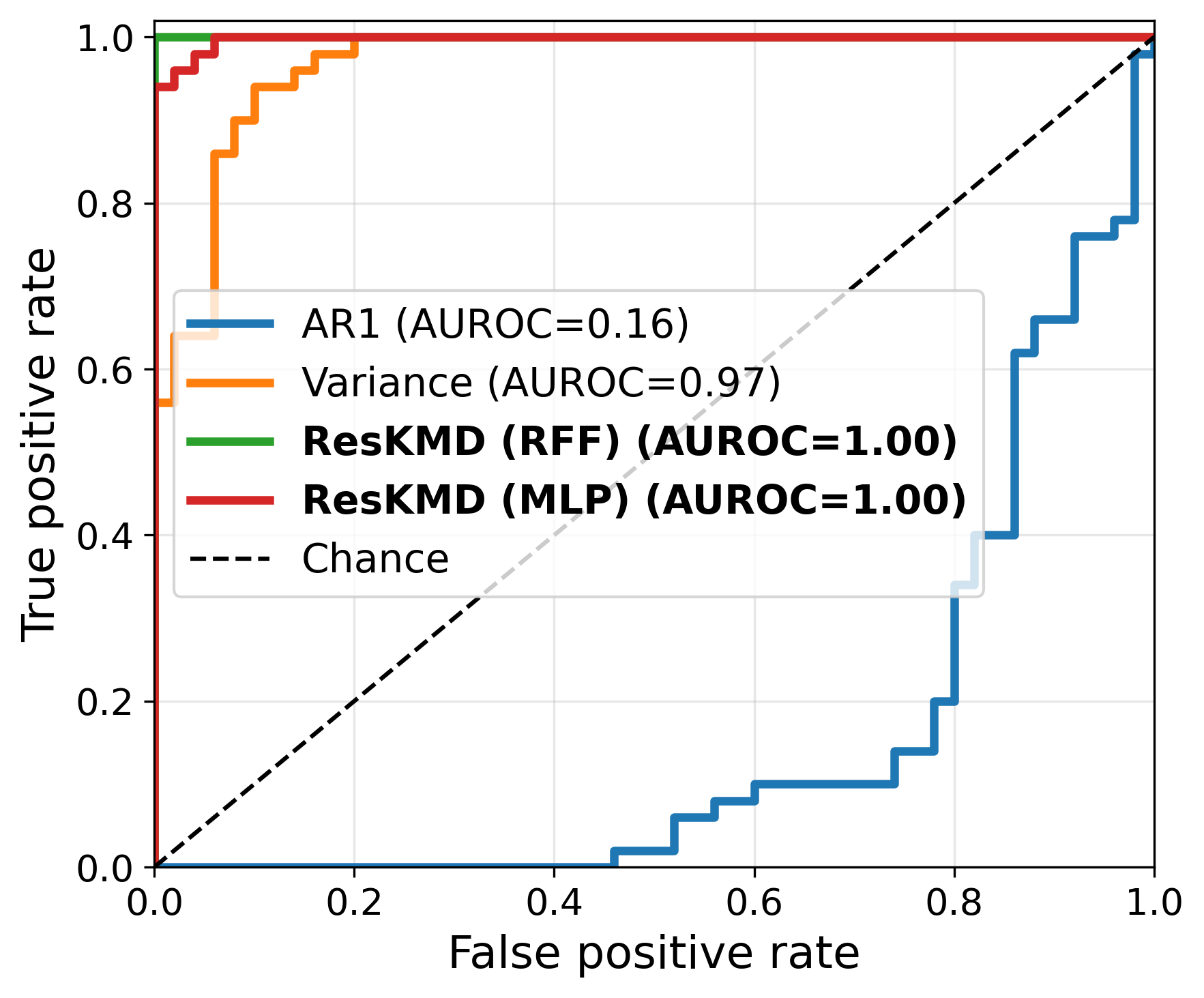}
        \caption{Rosenzweig-MacArthur}
        \label{fig:b_rosenzweig_macarthur_roc}
    \end{subfigure}
    \hfill
    \begin{subfigure}[h]{0.3\linewidth}
        \centering
        \includegraphics[width=\linewidth]{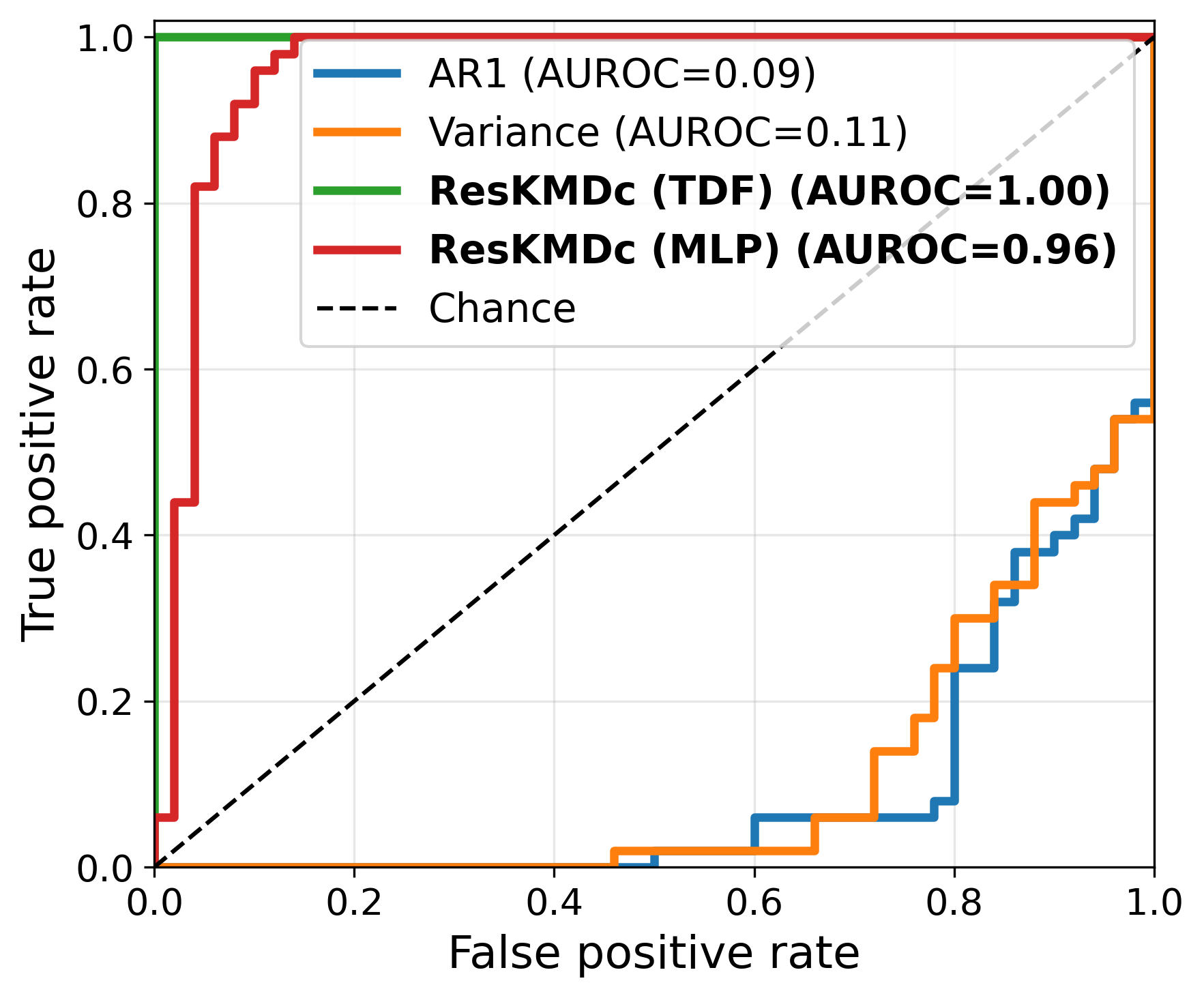}
        \caption{Saddle-node}
        \label{fig:r_saddle_node_roc}
    \end{subfigure}
    \hfill
    \begin{subfigure}[h]{0.3\linewidth}
        \centering
        \includegraphics[width=\linewidth]{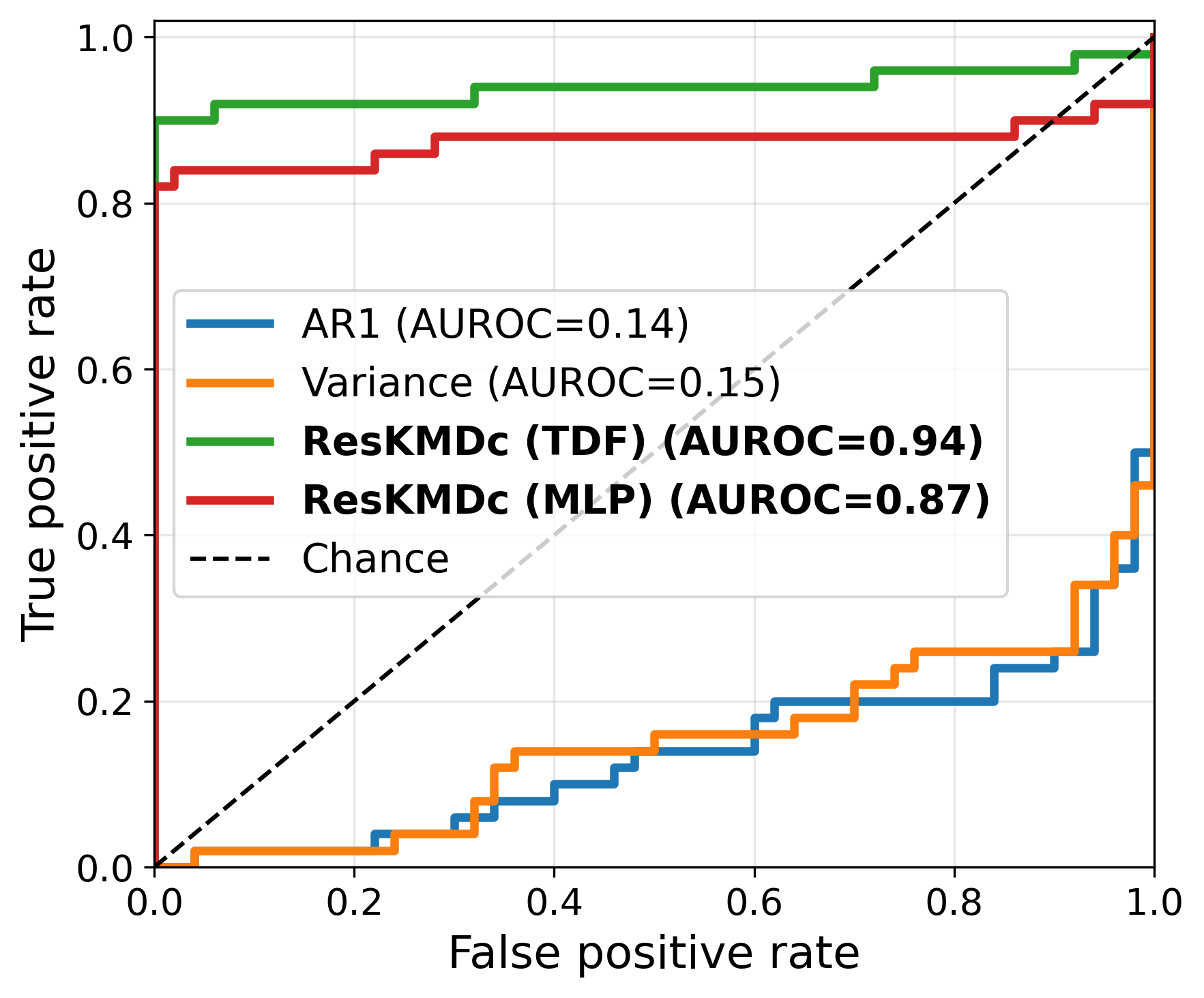}
        \caption{Bautin}
        \label{fig:r_bautin_roc}
    \end{subfigure}
    \hfill
    \begin{subfigure}[h]{0.3\linewidth}
        \centering
        \includegraphics[width=\linewidth]{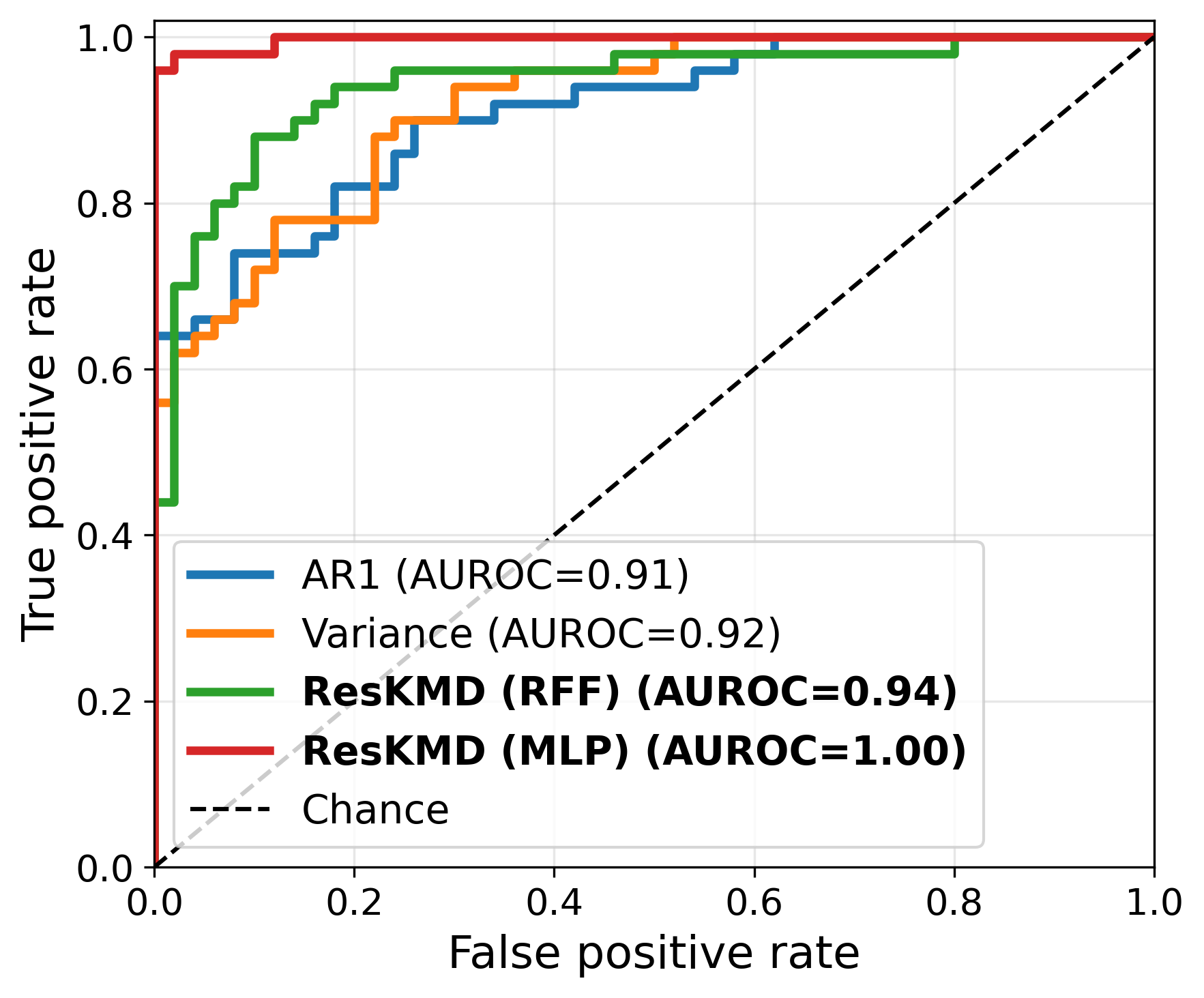}
        \caption{B-tipping: AMOC box}
        \label{fig:b_amoc3b_roc}
    \end{subfigure}
    \hfill
    \begin{subfigure}[h]{0.3\linewidth}
        \centering
        \includegraphics[width=\linewidth]{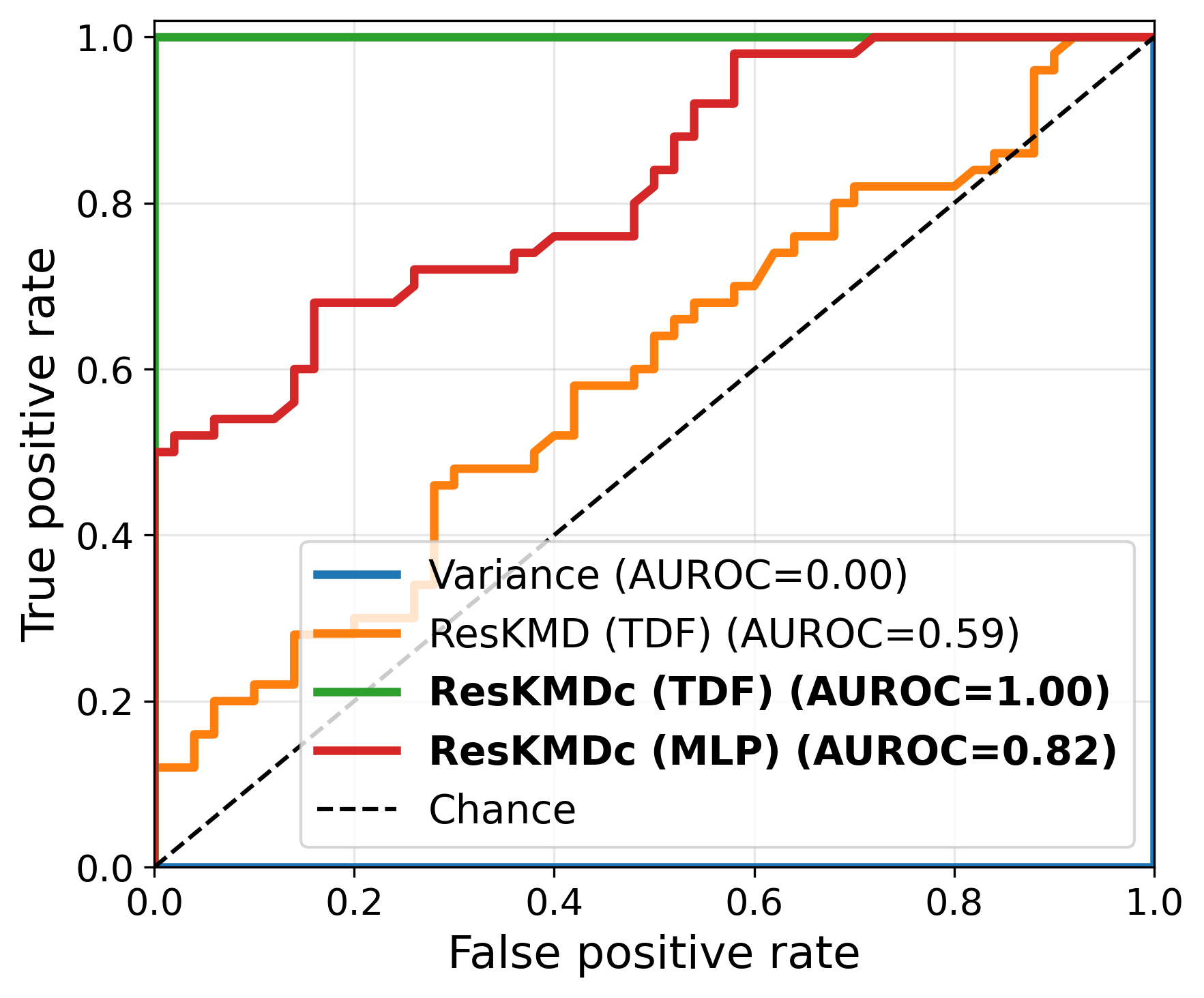}
        \caption{R-tipping: AMOC box}
        \label{fig:r_amoc3b_roc}
    \end{subfigure}
    \caption{Quantitative benchmarking of early warning indicators using AUROC. Panels correspond to (\textbf{a}) May’s harvesting model, (\textbf{b}) Rosenzweig–MacArthur model, (\textbf{c}) forced saddle-node system, (\textbf{d}) Bautin oscillator, (\textbf{e}) B-tipping AMOC box-model, and (\textbf{f}) R-tipping AMOC box-model. Higher AUROC values indicate improved discrimination between tipping and non-tipping trajectories. Across both tipping regimes, Koopman-based indicators generally outperform classical CSD indicators, with control-aware ResKMD outperforms in the R-tipping regimes.}
    \label{fig:auroc}
\end{figure}

\section*{Discussion}
\label{sec:discussion}
In this work, we develop a Koopman-based framework for EWS in stochastic systems undergoing either B-tipping or R-tipping. In the B-tipping setting, we built on recent results connecting ResKMD to CSD and showed how the residual provides a principled indicator that combines truncation error and stochastic fluctuation within a single spectral quantity. With a control setup we generalise this framework to R-tipping scenarios, yielding a control-aware ResKMD that captures loss of tracking relative to a moving attracting branch. Unlike the B-tipping case, where residual growth is tied to eigenvalues approaching the unit circle, the R-tipping mechanism is instead governed by unresolved drift associated with the time-varying control variables.

Across a hierarchy of examples, from low-dimensional prototypical systems to reduced box-model and fully coupled AMOC simulations, we found that the proposed Koopman-based indicators recover the expected warning behavior in B-tipping problems and provide useful discrimination in R-tipping regimes where classical EWS such as AR1 and variance can be weak or inconsistent. In the AMOC applications, the framework also yielded useful spectral diagnostics beyond detection itself, including $T_e$, spectral gaps, and the distribution of slow Koopman modes. These quantities provide additional physical interpretation of resilience loss and timescale separation, and highlight a central advantage of the operator-theoretic viewpoint: the same EDMD/ResDMD machinery used to construct the EWS also gives access to interpretable spectral structure. In real-world application where only a single realization is available, our ResKMD indicator can provide robust uncertainty estimate by performing multiple runs with different observable configuration.

Some considerations point to a number of promising active research directions. The practical performance of ResKMD and its control-aware extension depends on the choice of observables, truncation level, and the supplied control or driver variable. As seen in the AMOC experiments, some control choices are substantially more informative than others, and identifying the most dynamically relevant forcing channel is itself a nontrivial modeling problem. Furthermore, in the AMOC R-tipping box-model experiments, although the control-aware ResKMD indicator exhibits a substantially faster increase in the tipping trajectories, this increase alone may be difficult to interpret in isolation when non-tipping reference trajectories are unavailable. This highlights an important limitation for practical deployment in real-world systems when R-tipping dominates and counterfactual non-tipping realizations are typically inaccessible. Future work should therefore investigate more robust calibration strategies, uncertainty quantification, and methods for identifying absolute rather than purely comparative signatures of rate-induced loss of tracking. 

While the B-tipping case admits a clearer theoretical divergence result, the R-tipping case is more problem dependent: increased residuals arise through unresolved branch drift rather than through a universal spectral route (see our discussion in Remark~\ref{rem:no_general_divergence}). Our results also rely on finite-sample EDMD- and ResDMD-type approximations, whose statistical error, dictionary dependence, and robustness in high-dimensional settings deserve further study. On the theory side, sharper guarantees for control-aware ResKMD under partial observability, model misspecification, and nonstationary noise would further strengthen the framework's foundations, as would connecting the present framework more explicitly to Koopman generators, continuous spectrum, and ResDMD with control in more general settings. Understanding how residual growth relates to tracking failure in nonautonomous stochastic systems is a particularly interesting open question. Just like B-tipping can be seen as a manifestation of the closure of the spectral gap and of the divergence of response operators in the autonomous setting \cite{Santos2022,LucariniChekroun2023}, it is tempting to investigate whether recent advances on linear response for nonautonomous systems   \cite{Lucarini2026,GL26} can shed light on R-tipping and be related to the results presented here. On the application side, an important next step is to test these EWS more broadly across other climate tipping elements and ESM simulations, where multiple coupled control variables, spatial heterogeneity, and partial observations~\cite{nathaniel2023metaflux,kim2024spatiotemporal} make tipping detection especially challenging.

Overall, to our knowledge our results present the first  framework of an unified EWS analysis across B-tipping and R-tipping in stochastic systems using operator-theoretic formalism. Rather than treating these two tipping mechanisms with separate heuristic indicators, ResKMD offers a common language in which both loss of stability and loss of tracking can be studied through spectral approximation error. This study provides a significant step toward EWS methods that remain flexible and physically interpretable in the stochastic, high-dimensional, and strongly driven systems that arise in AMOC and other complex applications.

\section*{Methods}
\label{sec:methods}

\subsection*{Stochastic dynamics}

Consider a discrete stochastic dynamical system with state $\omega_t\in\Omega\subseteq\mathbb{R}^N$:
\begin{equation}
    \omega_{t+1}=\Phi(\omega_t)+\varepsilon_t,
    \qquad
    \Phi_\varepsilon(\omega):=\Phi(\omega)+\varepsilon.
    \label{eq:sde}
\end{equation}

Here, $\Phi:\Omega\to\Omega$ is the deterministic flow map and $\varepsilon_t$ denotes stochastic forcing. Unless otherwise stated, we assume that $\varepsilon_t$ has zero conditional mean and finite second-moment, so that the relevant expectations and covariances are well defined. In the synthetic experiments, we use independent Gaussian noise with covariance $\Sigma_\varepsilon$. Expectations below are taken with respect to the probability law induced by the noise sequence:
\begin{align}
    m_t &:= \mathbb{E}[\omega_t],\\
    \mathrm{Cov}(\omega_{t_1},\omega_{t_2})
    &:= \mathbb{E}\!\left[(\omega_{t_1}-m_{t_1})(\omega_{t_2}-m_{t_2})^\top\right].
\end{align}

\subsection*{Bifurcation-induced tipping}
To define bifurcation-induced tipping (B-tipping), consider a parameterized family of deterministic dynamical systems $\omega_{t+1} = \Phi_\beta(\omega_t)$, where $\beta \in \mathbb{R}$ is a slowly varying parameter. Suppose that, for each $\beta < \beta^\star$, the system admits a stable fixed point, $\omega^\star(\beta)$, satisfying $\Phi_\beta(\omega^\star(\beta)) = \omega^\star(\beta)$. The set of equilibria $\omega^\star(\beta)$ defines an attracting equilibrium branch parameterized by $\beta$. To characterize local stability along this branch, we linearize the dynamics around $\omega^\star(\beta)$. The corresponding Jacobian is
\begin{equation}
    \mathcal{J}(\beta) := \left.\frac{\partial \Phi_\beta}{\partial \omega}\right|_{\omega=\omega^\star(\beta)}.
    \label{eq:b_jacobian}
\end{equation}

Let $\{\lambda_{\mathcal{J},i}(\beta)\}_{i=1}^N$ denote the eigenvalues of $\mathcal{J}(\beta)$ ordered so that $|\lambda_{\mathcal{J},1}(\beta)|\ge \cdots \ge |\lambda_{\mathcal{J},N}(\beta)|$. We define the spectral radius of the Jacobian as $\rho(\mathcal{J}(\beta)) := \max_i |\lambda_{\mathcal{J},i}(\beta)|$. Local asymptotic stability requires $\rho(\mathcal{J}(\beta))<1$, so the spectrum lies strictly inside the unit circle. As the system evolves along the equilibrium branch toward the critical parameter value $\beta\to\beta^\star$, the leading spectral modes approach the unit circle, signaling weakening restoring stability and the onset of CSD.

For a fixed parameter value $\beta$, linearizing the perturbation 
$\bar{\omega}_t:=\omega_t-\omega^\star(\beta)$ gives
\begin{equation}
    \bar{\omega}_{t+1}\approx \mathcal{J}(\beta)\bar{\omega}_t+\epsilon_t,
\end{equation}

where $\epsilon_t$ collects stochastic forcing and neglected higher-order terms. Let $\ell_{\mathcal{J},1}(\beta)$ denote the left eigenvector associated with the dominant eigenvalue $\lambda_{\mathcal{J},1}(\beta)$, and define
\begin{align}
    z_t &:= \ell_{\mathcal{J},1}(\beta)^\top\bar{\omega}_t, \nonumber\\
    z_{t+1} &= \lambda_{\mathcal{J},1}(\beta)z_t+\eta_t,
    \label{eq:linearized_sde}
\end{align}

with $\eta_t:=\ell_{\mathcal{J},1}(\beta)^\top\epsilon_t$ denoting the projected stochastic innovation. This yields the familiar spectral interpretation of CSD. Assuming stationarity, $\mathrm{Var}(z_{t+1})=\mathrm{Var}(z_t)$, and uncorrelated innovations, $\mathrm{Cov}(z_t,\eta_t)=0$, we obtain AR1 (derivation in Supplementary Information~\ref{si-sec:ar1_derivation})
\begin{equation}
    \mathrm{Corr}(z_t,z_{t+1})=\lambda_{\mathcal{J},1}(\beta),
    \label{eq:csd_ar1}
\end{equation}

and variance (derivation in Supplementary Information~\ref{si-sec:var_derivation})
\begin{equation}
    \mathrm{Var}(z_t)=\frac{\mathrm{Var}(\eta_t)}{1-|\lambda_{\mathcal{J},1}(\beta)|^2}.
    \label{eq:csd_var}
\end{equation}

These relations are exact for one-dimensional linear systems and provide a useful leading mode approximation when the local dynamics are dominated by a single slow mode. If several modes decay on comparable timescales, AR1 and variance should instead be interpreted as effective scalar summaries rather than exact spectral diagnostics. Thus, as $\rho(\mathcal{J}(\beta))\to 1$, perturbations decay more slowly, and both AR1 and variance increase.

\subsection*{Rate-induced tipping}
Rate-induced tipping (R-tipping) refers to the loss of tracking of a moving attracting branch caused by sufficiently rapid variation in a time-dependent control variable, even when the corresponding frozen dynamics remain locally stable~\cite{ashwin2012tipping,kiers2020rate}. We consider the stochastic control-dependent system
\begin{equation}
    \omega_{t+1}=\Phi(\omega_t,u_t)+\varepsilon_t,
    \qquad
    \Phi_\varepsilon(\omega,u):=\Phi(\omega,u)+\varepsilon,
    \label{eq:sde_control}
\end{equation}

where $u_t\in\mathbb{R}^U$ is a time-dependent control variable. For each fixed value of $u$, we associate the frozen system $\omega_{t+1}=\Phi(\omega_t,u)$, obtained by holding the control constant. Assume that, for each $u$ in the regime of interest, this frozen system admits an equilibrium $\omega^\star(u)$ satisfying $\Phi(\omega^\star(u),u)=\omega^\star(u)$, so that $u\mapsto \omega^\star(u)$ defines a frozen equilibrium branch. We further assume that this branch is locally stable for each fixed $u$, with
\begin{equation}
    \mathcal{J}(u):=\left.\frac{\partial \Phi(\omega,u)}{\partial \omega}\right|_{\omega=\omega^\star(u)},
    \qquad
    \rho(\mathcal{J}(u))<1.
    \label{eq:r_jacobian}
\end{equation}

Thus, if the control variable is frozen at any single value, nearby trajectories would relax locally toward the corresponding equilibrium $\omega^\star(u)$. The key issue in the non-autonomous setting is whether the evolving trajectory $\omega_t$ remains close to the moving equilibrium branch $\omega^\star(u_t)$ as the control variable changes in time. To quantify this, we define the tracking error, $\bar{\omega}_t:=\omega_t-\omega^\star(u_t)$. Linearizing around the moving branch $\omega^\star(u_t)$ yields (derivation in Supplementary Information~\ref{si-sec:rate_perturb_dynamics})
\begin{equation}
    \bar{\omega}_{t+1} = \underbrace{\mathcal{J}(u_t)\bar{\omega}_t}_{\text{local restoration}} + \underbrace{\omega^\star(u_t)-\omega^\star(u_{t+1})}_{\text{branch drift, }d_t} + \underbrace{\epsilon_t}_{\text{high-order terms}}.
    \label{eq:rtip}
\end{equation}

The term $d_t:=\omega^\star(u_t)-\omega^\star(u_{t+1})$ is the control-induced drift of the frozen equilibrium branch, that is, the displacement of the instantaneous attracting state caused by the change in the control variable from $u_t$ to $u_{t+1}$. Equation~\ref{eq:rtip} therefore shows that R-tipping arises through competition between two effects: local restoration toward the current frozen equilibrium, represented by $\mathcal{J}(u_t)\bar{\omega}_t$, and control-driven drift, represented by $d_t$. In particular, tipping can occur even when $\rho(\mathcal{J}(u))<1$ for every fixed $u$, provided the control changes rapidly enough that the trajectory can no longer track the moving branch.

\subsection*{Stochastic Koopman operators}
We now define the Koopman-theoretic framework used throughout the paper. Koopman theory represents nonlinear state evolution through linear operators acting on observables $\psi:\Omega\to\mathbb{C}$~\cite{brunton2021modern,koopman1931hamiltonian}. We let $\mu$ be a finite positive measure on the state space $\Omega$. Consider the space $L^2(\Omega,\mu)$ equipped with the standard inner product
\begin{equation}
    \langle \psi_1,\psi_2\rangle_{L^2(\Omega,\mu)}
    :=
    \int_\Omega \psi_1(\omega)\overline{\psi_2(\omega)}\,d\mu(\omega).
\end{equation}

The first-moment stochastic Koopman operator  $\mathcal{K}^{(1)}$ acts on observables of $L^2(\Omega,\mu)$ by expectation
\begin{equation}
    [\mathcal{K}^{(1)}\psi](\omega)
    :=
    \mathbb{E}\!\left[\psi(\Phi_\varepsilon(\omega))\right].
    \label{eq:stochastic_koopman}
\end{equation}

As a linear operator on $L^2(\Omega,\mu)$, $\mathcal{K}^{(1)}$ can be studied through its spectrum, whose structure depends on the underlying dynamics. When Koopman eigenfunctions exist, the corresponding eigenpairs $(\lambda,\phi)$ satisfy
\begin{equation}
    \mathcal{K}^{(1)}\phi=\lambda\phi.
\end{equation}

To connect Koopman spectral structure with stochastic variability and classical CSD indicators such as variance, we additionally define a second-moment Koopman operator
\begin{equation}
    [\mathcal{K}^{(2)}(\psi_1,\psi_2)](\omega)
    :=
    \mathbb{E}\!\left[
        \psi_1(\Phi_\varepsilon(\omega))
        \overline{\psi_2(\Phi_\varepsilon(\omega))}
    \right].
\end{equation}

The corresponding one-step variance is
\begin{equation}
    \mathrm{Var}\!\left[\psi(\Phi_\varepsilon(\omega))\right]
    =
    [\mathcal{K}^{(2)}(\psi,\psi)](\omega)
    -
    \left|[\mathcal{K}^{(1)}\psi](\omega)\right|^2.
\end{equation}

In the data driven setting, we use observable dictionaries built from random Fourier features (RFF~\cite{rahimi2007random}), time-delay features (TDF~\cite{abarbanel1994predicting}), and deep embeddings, such as using multi-layer perceptron (MLP)-based autoencoders~\cite{nathaniel2025deep}.

\subsection*{Residual Koopman framework for B-tipping}
We now formalize the residual Koopman framework used for the autonomous, bifurcation-induced setting. Throughout this subsection, the sub/superscript $b$ denotes the bifurcation setting. Let $\boldsymbol{\psi}=[\psi_1,\ldots,\psi_M]^\top$ be a vector-valued observable dictionary of dimension $M$. When non-constant Koopman eigenfunctions are available, let $\{(\lambda_n^b,\phi_n^b)\}_{n\ge 1}$ denote the retained eigenvalue--eigenfunction pairs, ordered by decreasing magnitude. We exclude the constant invariant eigenfunction with eigenvalue one, since it does not represent decay or early warning behavior. For a truncation level $m\in\mathbb{N}^+$, we define:
\begin{equation}
    \mathcal{S}_m^b:=\mathrm{span}\{\phi_1^b,\ldots,\phi_m^b\}.
\end{equation}

Following the standard Koopman mode decomposition~\cite{brunton2021modern}, the truncated Koopman representation of the vector observable $\boldsymbol{\psi}$ on $\mathcal{S}_m^b$ is:
\begin{equation}
    [\mathcal{K}^{(1)}_{\mathcal{S}_m^b}\boldsymbol{\psi}](\omega)
    =
    \sum_{n=1}^m \lambda_n^b v_n^b \phi_n^b(\omega),
    \label{eq:kmd}
\end{equation}

where $v_n^b\in\mathbb{C}^M$ is the Koopman mode associated with eigenfunction $\phi_n^b$. We then define the residual Koopman mode decomposition (ResKMD) by
\begin{equation}
    \mathrm{res}\!\left[
        \mathcal{K}^{(1)},
        \mathcal{K}^{(1)}_{\mathcal{S}_m^b};
        \boldsymbol{\psi}
    \right]^2
    :=
    \int_\Omega
    \mathbb{E}\!\left[
        \left\|
            \boldsymbol{\psi}(\Phi_\varepsilon(\omega))
            -
            [\mathcal{K}^{(1)}_{\mathcal{S}_m^b}\boldsymbol{\psi}](\omega)
        \right\|_2^2
    \right]
    d\mu(\omega).
    \label{eq:reskmd}
\end{equation}

This residual separates into a deterministic approximation component and a stochastic fluctuation component. Specifically
\begin{equation}
    \mathrm{res}\!\left[
        \mathcal{K}^{(1)},
        \mathcal{K}^{(1)}_{\mathcal{S}_m^b};
        \boldsymbol{\psi}
    \right]^2
    =
    \underbrace{
    \left\|
        \mathcal{K}^{(1)}\boldsymbol{\psi}
        -
        \mathcal{K}^{(1)}_{\mathcal{S}_m^b}\boldsymbol{\psi}
    \right\|_{L^2(\Omega,\mu;\mathbb{C}^M)}^2
    }_{\text{truncation error}}
    +
    \underbrace{
    \int_\Omega
    \mathrm{Tr}\!\left[
        \mathrm{Cov}\!\left(\boldsymbol{\psi}(\Phi_\varepsilon(\omega))\right)
    \right]d\mu(\omega)
    }_{\text{stochastic fluctuation}}.
    \label{eq:reskmd_decomp}
\end{equation}

Here $\mathrm{Cov}$ denotes the covariance of the next-step observable $\boldsymbol{\psi}(\Phi_\varepsilon(\omega))$ induced by stochasticity at fixed current state $\omega$. The first term in Equation~\ref{eq:reskmd_decomp} measures the error incurred by representing the expected observable evolution using only the retained Koopman subspace \(\mathcal{S}_m^b\), while the second term measures the stochastic spread of the lifted one-step evolution.

To interpret the truncation term, consider the simplifying case in which the relevant Koopman eigenfunctions form an orthonormal basis for the observable subspace of interest. Then the observable can be expanded as $\boldsymbol{\psi} =  \sum_{n\ge 1} v_n^b \phi_n^b$, and the truncation error becomes
\begin{equation}
    \left\|
        \mathcal{K}^{(1)}\boldsymbol{\psi}
        -
        \mathcal{K}^{(1)}_{\mathcal{S}_m^b}\boldsymbol{\psi}
    \right\|_{L^2(\Omega,\mu;\mathbb{C}^M)}^2
    =
    \sum_{n>m}
    |\lambda_n^b|^2
    \|v_n^b\|_2^2.
    \label{eq:trunc_modal}
\end{equation}

With the orthonormal-basis and eigenvalue-ordering assumptions above, the truncation error is bounded below by the first omitted Koopman mode and above by a worst-case tail contribution controlled by the largest omitted eigenvalue~\cite{miyauchi2026generalized}
\begin{equation}
    \underbrace{
    |\lambda^b_{m+1}|^2\|v^b_{m+1}\|_2^2
    }_{\text{first omitted mode}}
    \le
    \left\|
        \mathcal{K}^{(1)}\boldsymbol{\psi}
        -
        \mathcal{K}^{(1)}_{\mathcal{S}_m^b}\boldsymbol{\psi}
    \right\|_{L^2(\Omega,\mu;\mathbb{C}^M)}^2
    \le
    \underbrace{
    |\lambda^b_{m+1}|^2\|\boldsymbol{\psi}\|_{L^2(\Omega,\mu;\mathbb{C}^M)}^2
    }_{\text{worst-case omitted tail}}.
    \label{eq:trunc_bounds}
\end{equation}

Moreover, for each non-constant Koopman eigenfunction, the stochastic one-step evolution can be written as
\begin{equation}
    \phi_n^b(\omega_{t+1})
    =
    \lambda_n^b\phi_n^b(\omega_t)+\eta_{n,t},
    \qquad
    \eta_{n,t}
    :=
    \phi_n^b(\Phi_\varepsilon(\omega_t))
    -
    \lambda_n^b\phi_n^b(\omega_t).
    \label{eq:eigenfunction_ar1}
\end{equation}

Here $\eta_{n,t}$ is the one-step stochastic innovation associated with the Koopman eigenfunction $\phi_n^b$. Under the same stationarity and innovation-independence assumptions used in Equations~\ref{eq:csd_ar1} and \ref{eq:csd_var}, this gives
\begin{equation}
    \mathrm{Var}\!\left[\phi_n^b(\omega_t)\right]
    =
    \frac{\mathrm{Var}(\eta_{n,t})}{1-|\lambda_n^b|^2}.
    \label{eq:koopman_coordinate_var}
\end{equation}

From Equation~\ref{eq:trunc_bounds}, the contribution of the first omitted mode can increase as its eigenvalue approaches the unit circle. Similarly, Equation~\ref{eq:koopman_coordinate_var} further shows that stochastic fluctuations become larger when a non-constant Koopman eigenvalue approaches the unit circle. Intuitively, this happens because perturbations decay more slowly, so noise persists for longer. Thus, near B-tipping, ResKMD can increase through two complementary mechanisms: poorer finite-dimensional spectral approximation and stronger stochastic amplification. This mechanism is formalized for the bifurcation setting by Theorem~2 of \cite{miyauchi2026generalized}, which shows, under its stated assumptions, that
\begin{equation}
    \mathrm{res}\!\left[
        \mathcal{K}^{(1)},
        \mathcal{K}^{(1)}_{\mathcal{S}_m^b};
        \boldsymbol{\psi}
    \right]^2
    \to\infty
    \qquad\text{as }\beta\to\beta^\star.
    \label{eq:reskmd_diverge}
\end{equation}

\subsection*{Control-aware residuals for R-tipping}
We now apply the Koopman residual framework to the augmented state-control system introduced above. Throughout this subsection, the sub/superscript $c$ denotes the control setting. Using $u_{t+1}=g(u_t) \in \mathbb{R}^U$, the corresponding one-step map on the joint state-control space $\Omega\times\mathbb{R}^U$ has the form
\begin{equation}
    \Phi_{c,\varepsilon}(\omega,u)
    :=
    \bigl(\Phi_\varepsilon(\omega,u),g(u)\bigr).
    \label{eq:augmented_control_map}
\end{equation}

This rewrites the control-dependent dynamics as an autonomous stochastic system on the joint state-control space. We assume that the control update $g: \mathbb{R}^U \to \mathbb{R}^U$ is deterministic, so that stochasticity enters only through the state transition $\Phi_\varepsilon(\omega,u)$. This keeps the one-step uncertainty collected in a single stochastic term for notational simplicity.

Let $\mu$ be a finite reference measure on $\Omega\times\mathbb{R}^U$, such as the empirical measure induced by an observed state-control trajectory. We define state-only, joint state-control, and control-only observable blocks by
\[
    \boldsymbol{\psi}^{(\omega)}:\Omega\to\mathbb{C}^{M_\omega},
    \qquad
    \boldsymbol{\psi}^{(\omega u)}:\Omega\times\mathbb{R}^U\to\mathbb{C}^{M_{\omega u}},
    \qquad
    \boldsymbol{\psi}^{(u)}:\mathbb{R}^U\to\mathbb{C}^{M_u}.
\]

The concatenated control-aware dictionary is
\begin{equation}
    \boldsymbol{\Psi}(\omega,u)
    :=
    \begin{bmatrix}
        \boldsymbol{\psi}^{(\omega)}(\omega) \\
        \boldsymbol{\psi}^{(\omega u)}(\omega,u) \\
        \boldsymbol{\psi}^{(u)}(u)
    \end{bmatrix}
    \in\mathbb{C}^M,
    \qquad
    M=M_\omega+M_{\omega u}+M_u .
    \label{eq:psi_control}
\end{equation}

The three blocks need not have the same dimension. Their specific construction is implementation-dependent. In our numerical experiments, we construct the control variable directly i.e., using the identity observable, and the joint observable using simple products between the state observables and the control variable, following~\cite{proctor2018generalizing}.

The control-aware first-moment Koopman operator acts on observables of the augmented state-control space as
\begin{equation}
    [\mathcal{K}^{(1)}_c\Psi](\omega,u)
    :=
    \mathbb{E}\!\left[\Psi(\Phi_\varepsilon(\omega,u),g(u))\right].
    \label{eq:kic_first_moment}
\end{equation}

Thus, $\mathcal{K}^{(1)}_c$ is simply the stochastic Koopman operator associated with the augmented map in Equation~\ref{eq:augmented_control_map}. The control-aware formulation therefore applies the same Koopman residual idea as before, but now to the pair $(\omega,u)$ rather than to $\omega$ alone.

Let $\{(\lambda_n^c,\phi_n^c)\}_{n\ge 1}$ denote the retained non-constant eigenvalue-eigenfunction pairs of $\mathcal{K}^{(1)}_c$, ordered by decreasing magnitude. We assume that such point-spectrum components exist. For a truncation level $m \in \mathbb{N}^+$, we define:
\begin{equation}
    \mathcal{S}_m^c:=\mathrm{span}\{\phi_1^c,\ldots,\phi_m^c\}.
\end{equation}

Following the standard Koopman mode decomposition, the truncated control-aware representation is
\begin{equation}
    [\mathcal{K}^{(1)}_{\mathcal{S}_m^c}\boldsymbol{\Psi}](\omega,u)
    =
    \sum_{n=1}^m \lambda_n^c v_n^c \phi_n^c(\omega,u),
    \label{eq:kmd_control}
\end{equation}

where $v_n^c\in\mathbb{C}^{M}$ is the Koopman mode associated with $\phi_n^c$. We define the control-aware residual by
\begin{equation}
    \mathrm{res}\!\left[
        \mathcal{K}^{(1)}_c,
        \mathcal{K}^{(1)}_{\mathcal{S}_m^c};
        \boldsymbol{\Psi}
    \right]^2
    :=
    \int_{\Omega\times\mathbb{R}^U}
    \mathbb{E}\!\left[
        \left\|
            \boldsymbol{\Psi}(\Phi_\varepsilon(\omega,u),g(u))
            -
            [\mathcal{K}^{(1)}_{\mathcal{S}_m^c}\boldsymbol{\Psi}](\omega,u)
        \right\|_2^2
    \right]
    d\mu(\omega,u),
    \label{eq:reskmd_control}
\end{equation}

Because this is the usual ResKMD residual applied to the augmented state-control system, the same decomposition into approximation error and stochastic fluctuation carries over. This gives rise to Proposition~\ref{prop:reskmdc_decomp}, and detailed proof is given in Supplementary Information~\ref{si-sec:proof-reskmdc-decomp}.

\begin{boxremark}{No universal residual divergence in the R-tipping setting}{no_general_divergence}
Unlike in the B-tipping case, one should not expect a universal residual-divergence result for R-tipping, because rate-induced transitions need not involve loss of local stability or Koopman eigenvalues approaching the unit circle. If the augmented state-control dynamics are represented exactly in the chosen observable space, the control-aware residual can remain zero even when the trajectory loses track of the moving equilibrium branch; a counterexample is given in Supplementary Information~\ref{si-sec:counterexample-rtipping}. Thus, the residual should be interpreted as detecting unresolved tracking dynamics in a finite-dimensional data-driven representation, rather than tracking failure in a perfectly resolved model.
\end{boxremark}

We now connect the control-aware residual to the lower bound mechanism in Theorem~\ref{thm:rtip_lower_bound}. Recall the tracking error, $\bar{\omega}(\omega,u) := \omega-\omega^\star(u)$ and branch drift, $d(\omega,u) := \omega^\star(u)-\omega^\star(g(u))$ as defined earlier, but now adapted using the control dynamics $g$.

\begin{boxassumption}{Tracking error observability}{tracking_repr}
There exists a bounded linear map $C:\mathbb{C}^M\to\mathbb{R}^N$, such that $\bar{\omega}(\omega,u) = C\boldsymbol{\Psi}(\omega,u)$.
\end{boxassumption}

Assumption~\ref{assm:tracking_repr} says that the chosen observables contain enough information to reconstruct the tracking error. In practice this may hold only approximately, but it makes explicit the link between residuals in the augmented observable space and tracking errors in state space.

Using the tracking error dynamics from Equation~\ref{eq:rtip}, we write the first-moment evolution of $\bar{\omega}$ as
\begin{equation}
    [\mathcal{K}^{(1)}_c\bar{\omega}](\omega,u)
    =
    \underbrace{\mathcal{J}(u)\bar{\omega}(\omega,u)}_{\text{local restoration}}
    +
    \underbrace{d(\omega,u)}_{\text{branch drift}}
    +
    \underbrace{\epsilon(\omega,u)}_{\text{high-order terms}}.
    \label{eq:first_moment_decomp}
\end{equation}

Equation~\ref{eq:first_moment_decomp} separates the expected tracking error evolution into local restoration, branch drift, and remaining higher-order terms from linearization. The lower bound follows by projecting this decomposition onto the part not represented by the retained Koopman subspace. Let $\mathcal{P}_m$ denote projection onto $\mathcal{S}_m^c$, and let $\mathcal{Q}_m := I-\mathcal{P}_m$ denote the unresolved projection. Under the tracking error observability Assumption~\ref{assm:tracking_repr}, and assuming that the local restoration term is captured by the retained Koopman subspace, the control-aware residual satisfies
\begin{equation}
    \mathrm{res}\!\left[
        \mathcal{K}^{(1)}_c,
        \mathcal{K}^{(1)}_{\mathcal{S}_m^c};
        \boldsymbol{\Psi}
    \right]
    \ge
    \frac{1}{\|C\|}
    \left(
    \|\mathcal{Q}_m d\|_{L^2(\Omega\times\mathbb{R}^U,\mu;\mathbb{R}^N)}
    -
    \|\mathcal{Q}_m\epsilon\|_{L^2(\Omega\times\mathbb{R}^U,\mu;\mathbb{R}^N)}
    \right).
    \label{eq:drift_lower_bound}
\end{equation}

If $\mathcal{Q}_m\epsilon=0$, this further reduces to
\begin{equation}
    \mathrm{res}\!\left[
        \mathcal{K}^{(1)}_c,
        \mathcal{K}^{(1)}_{\mathcal{S}_m^c};
        \boldsymbol{\Psi}
    \right]
    \ge
    \frac{1}{\|C\|}
    \|\mathcal{Q}_m d\|_{L^2(\Omega\times\mathbb{R}^U,\mu;\mathbb{R}^N)}.
    \label{eq:drift_lower_bound_clean}
\end{equation}

The proof is provided in Supplementary Information~\ref{si-sec:proof-rtip_lower_bound}. Theorem~\ref{thm:rtip_lower_bound} shows that unresolved branch drift provides a lower-bound contribution to the control-aware residual. This does not mean that all residual growth is caused by branch drift: unresolved state dynamics, stochastic fluctuations, imperfect observability, and finite-dimensional approximation error may also contribute. Rather, the theorem identifies one mechanism by which rate-induced loss of tracking can appear as increased Koopman residual.

\subsection*{Extended dynamic mode decomposition}
We estimate finite-dimensional autonomous and control-aware Koopman operators from trajectory data using extended dynamic mode decomposition (EDMD)~\cite{li2017extended}. Given a trajectory $\{\omega_t\}_{t=1}^{T+1}$, in the autonomous setting, we form one-step snapshot pairs $(\omega_t,\omega_{t+1})$ and lift each observation into a higher-dimensional feature space using a dictionary of observable $\boldsymbol{\psi} \in \mathbb{C}^M$ constructing two data matrices:
\begin{equation}
    \Psi_X :=
    \begin{bmatrix}
        \boldsymbol{\psi}(\omega_1) & \cdots & \boldsymbol{\psi}(\omega_T)
    \end{bmatrix},
    \qquad
    \Psi_Y :=
    \begin{bmatrix}
        \boldsymbol{\psi}(\omega_2) &  \cdots & \boldsymbol{\psi}(\omega_{T+1})
    \end{bmatrix},
\end{equation}

The Koopman approximation is then given as $K \approx \Psi_Y \Psi_X^\dagger$, where ${}^\dagger$ is the Moore-Penrose pseudoinverse \cite{Barata2012}. In the control-aware setting, we proceed analogously on the augmented trajectory $\{(\omega_t,u_t)\}_{t=1}^{T+1}$ using dictionary $\boldsymbol{\Psi}(\omega,u)$ (see Equation~\ref{eq:psi_control}):
\begin{equation}
    \Psi_{X_c} :=
    \begin{bmatrix}
        \boldsymbol{\Psi}(\omega_1,u_1) & \cdots & \boldsymbol{\Psi}(\omega_T,u_T)
    \end{bmatrix},
    \qquad
    \Psi_{Y_c} :=
    \begin{bmatrix}
        \boldsymbol{\Psi}(\omega_2,u_2) & \cdots & \boldsymbol{\Psi}(\omega_{T+1},u_{T+1})
    \end{bmatrix},
\end{equation}

from which we estimate $K_c \approx \Psi_{Y_c} \Psi_{X_c}^\dagger$ jointly captures the state-control dynamics. 

We consider both prescribed (random Fourier features (RFF)~\cite{rahimi2007random} and time-delay features (TDF)~\cite{abarbanel1994predicting}) and deep learning-based (latent embeddings from multi-layer perceptron (MLP) autoencoders~\cite{rupe2024causal,nathaniel2025deep}) for our choice of observables. Unless otherwise stated, we retain Koopman modes explaining 90\% of the variability. By default, we also use RFF with 500-dimension and TDF with 20-step delay, unless otherwise specified. For MLP observables, we use the default hyperparameters of learning rate $10^{-3}$, batch size 32, and training epochs of 1000. Training is performed once over the first window and the frozen encoder is applied over the subsequent sliding windows. In summary, this yields a fully data-driven procedure that does not require explicit knowledge of the governing equations while retaining the spectral interpretation of the Koopman framework. 

\subsection*{Quantitative metrics}
In addition to the qualitative rolling-window analysis of the indicators, we quantitatively benchmark the proposed Koopman EWS using Kendall’s $\tau$~\cite{kendall1938new} to measure monotonic trends prior to tipping, and receiver operating characteristic (ROC) analysis~\cite{fawcett2006introduction,hanley1982meaning} together with the corresponding area under the ROC curve (AUROC)~\cite{peterson1954theory,bradley1997use} to evaluate discrimination between tipping and non-tipping trajectories. 

We consider a binary classification setting where each sample is labeled as either tipping ($y=1$) or non-tipping ($y=0$). Let $\hat{s}: \mathcal{D} \to \mathbb{R}$ denote a scoring function that assigns a confidence score to each observed trajectory, where higher values indicate a higher likelihood of a critical transition. A prediction $\hat{y}=1$ is made whenever $\hat{s} > \tau$ for a threshold $\tau$.

For each threshold $\tau$, we define the true positive rate (TPR) and false positive rate (FPR) as
\begin{equation}
\text{TPR}(\tau) = \frac{|\{i : y_i = 1,\ \hat{s}_i > \tau\}|}{|\{i : y_i = 1\}|}, \quad
\text{FPR}(\tau) = \frac{|\{i : y_i = 0,\ \hat{s}_i > \tau\}|}{|\{i : y_i = 0\}|}.
\end{equation}

The receiver operating characteristic (ROC) curve \cite{fawcett2006introduction,hanley1982meaning} is obtained by plotting $\text{TPR}(\tau)$ against $\text{FPR}(\tau)$ as the threshold $\tau$ varies over all possible values. The ROC curve characterizes the trade-off between correctly detecting tipping events and incorrectly classifying non-tipping trajectories.

The Area Under the ROC Curve (AUROC) \cite{peterson1954theory,bradley1997use} is defined as
\begin{equation}
\text{AUROC} = \int_0^1 \text{TPR}\bigl(\text{FPR}^{-1}(u)\bigr)\,du,
\end{equation}

which corresponds to the probability that a randomly chosen tipping trajectory is assigned a higher score than a randomly chosen non-tipping trajectory. An AUROC of $0.5$ indicates random performance, while $1.0$ corresponds to perfect discrimination.

\subsection*{Prototypical tipping systems}
This section provides the data-generation protocols and model specifications for the prototypical systems and AMOC experiments presented in the main text. 

For bifurcation-induced tipping, we consider the following systems:

We first model the \textit{May's single-species harvesting model} for a biomass variable $x$, corresponding to the one-dimensional case $\omega = x \in \mathbb{R}_+$. Its continuous-time dynamics are
\begin{equation}
    \frac{dx}{dt} = r x\left(1-\frac{x}{k}\right) - h\frac{x^2}{s^2+x^2} + \sigma_x \xi_x(t),
    \label{eq:may_harvesting}
\end{equation}

where $r$ is the intrinsic growth rate, $k$ is the carrying capacity, $s$ sets the saturation scale of harvesting, $h$ is the harvesting rate, and $\xi_x(t)$ denotes Gaussian white noise. The harvesting term counteracts logistic growth, and as $h$ increases the stable positive-biomass equilibrium eventually disappears in a fold bifurcation. In our experiments, we set $r=1$, $k=1$, $s=0.1$, and $\sigma=0.01$, for which the deterministic system loses stability at the critical harvesting rate $h^\star = 0.26$. Tipping runs are generated by ramping $h$ linearly from $0.15$ to $0.27$, while non-tipping runs use the same initial ramp but capped at $h=0.25$, keeping the system below the fold. 

We next consider the \textit{Rosenzweig-MacArthur consumer-resource model}, for which the state is two-dimensional, $\omega = (x,y)^\top \in \mathbb{R}_+^2$, with $x$ the resource population and $y$ the consumer population. The continuous-time dynamics are
\begin{align}
    \frac{dx}{dt} &= r x\left(1-\frac{x}{k}\right) - \frac{axy}{1+ahx} + \sigma_x \xi_x(t), \nonumber \\
    \frac{dy}{dt} &= \frac{eaxy}{1+ahx} - my + \sigma_y \xi_y(t),
    \label{eq:rm_hopf}
\end{align}

where $r$ is the intrinsic growth rate of the resource, $k$ is its carrying capacity, $a(t)$ is the consumer attack rate, $e$ is the conversion efficiency, $h$ is the handling time, $m$ is the consumer mortality rate, and $\xi_x,\xi_y$ are independent Gaussian white noise processes. We use $r=4$, $k=1.7$, $e=0.5$, $h=0.15$, and $m=2$, for which the coexistence equilibrium destabilizes at the Hopf threshold $a^\star = 15.69$. Tipping runs are produced by ramping $a$ from $12$ to $17$, thereby crossing $a^\star$, whereas non-tipping runs use the same ramp but capped at $a=14.38$, so the system remains on the stable side of the Hopf bifurcation.

For both systems, we generate 50 tipping and 50 non-tipping stochastic trajectories with noise amplitude $\sigma=0.01$. Integration uses the Euler--Maruyama scheme with internal step size $\Delta t=0.01$. EWS are computed on pre-transition time series of length 500 using sliding windows of size $50\%$ and stride 1. The data-generation protocol follows \cite{bury2021deep}, except that for the non-tipping cases we use the same parameter schedule as the tipping cases before capping the bifurcating parameter slightly below the critical threshold, thereby making the discrimination task more challenging.

For rate-induced tipping, we consider the following systems:

We first consider a one-dimensional \textit{saddle-node normal form} with state $\omega = x \in \mathbb{R}$ and control variable $u(t)=\lambda(t)$. The continuous-time dynamics are
\begin{equation}
    \frac{dx}{dt} = (x+\lambda(t))^2 - 1 + \sqrt{2\sigma_x}\,\xi_x(t),
    \qquad
    \lambda(t) = \frac{\lambda_{\max}}{2} \left[\tanh\left(\frac{\lambda_{\max}\epsilon t}{2}\right)+1 \right],
    \label{eq:rate_saddle}
\end{equation}

where $\lambda_{\max}$ sets the forcing amplitude and $\epsilon$ controls the forcing rate. For each fixed $\lambda$, the frozen system has the usual saddle-node equilibrium structure, but the trajectory can still fail to track the attracting branch if $\lambda(t)$ increases too rapidly. In our experiments, we use $\lambda_{\max}=3$ and compare a tipping case with $\epsilon=1.25$ against a non-tipping case with the same forcing profile at half the rate, $\epsilon=0.625$. We set $\sigma_x = 0.008$. We use the scalar observable $x(t)$ for EWS estimation, and identify tipping when the trajectory crosses the threshold $x \geq 0$.

We next consider a \textit{stochastic Bautin system} with complex state $z_t = x_t + i y_t$ and forcing variable $u(t)=\Lambda(t)$. The continuous-time dynamics are
\begin{align}
    \frac{dz}{dt} ={}& (a+i\omega)(z-\Lambda(t))
    - b(z-\Lambda(t))^2(z-\Lambda(t))
    + (z-\Lambda(t))^4(z-\Lambda(t))
    + \sigma_z\,\xi_z(t), \nonumber\\
    \Lambda(t) ={}& \frac{\Lambda_{\max}}{2} \left[\tanh\left(\frac{\Lambda_{\max} r t}{2}\right)+1\right],
    \label{eq:rate_bautin}
\end{align}

where $a$ is the linear growth parameter, $\omega$ is the angular frequency, $b$ is the cubic coefficient, $\sigma_z$ sets the noise magnitude, and $r$ controls the forcing rate. In this example the forcing translates the underlying Bautin normal form through state space, and sufficiently rapid variation in $\Lambda(t)$ can trigger rate-induced escape even when the dynamics remain locally stable. We use the parameter values $a=0.1$, $\omega=3$, $b=1$, $\sigma_z=0.2$, and $\Lambda_{\max}=8$. Tipping runs use $r=0.10$, whereas non-tipping runs use the same forcing shape with half the rate, $r=0.05$. For EWS estimation we use the scalar radius observable $\rho(t)=\sqrt{x_t^2+y_t^2}$, and identify tipping when $\rho(t) \geq 10$.

For the \textit{reduced AMOC experiments}, we use the box-model analyzed in \cite{ritchie2023rate,alkhayuon2019basin}. This system is obtained from the original five-box model by treating the Southern Ocean (S) and bottom water (B) salinities as slow or fixed and diagnosing the Indo-Pacific (IP) salinity from salt conservation, yielding a two-dimensional dynamical system with state $\omega=(S_N,S_T)^\top\in\mathbb{R}^2$, where $S_N$ and $S_T$ denote the North Atlantic and tropical Atlantic salinity variables, respectively. The AMOC strength is diagnosed as
\begin{equation}
    Q = \frac{\lambda\left[\alpha(T_S-T_0)+\beta(S_N-S_S)\right]}{1+\lambda\alpha\mu},
    \label{eq:amoc_strength}
\end{equation}

and the observable supplied to the EWS estimators is the scalar time series $Q(t)$. The model is driven by freshwater hosing through the surface freshwater fluxes. In the reduced formulation, the salinity dynamics are piecewise-defined according to the sign of $Q$.

For $Q \ge 0$,
\begin{align}
    V_N \frac{dS_N}{dt} &= Q(S_T-S_N) + K_N(S_T-S_N) - F_N(H(r,t))S_0 + \sigma_N\xi_N(t), \label{eq:amoc_sde_posq_SN}\\
    V_T \frac{dS_T}{dt} &= Q\big(\gamma S_S + (1-\gamma)S_{IP} - S_T\big)
    + K_S(S_S-S_T) + K_N(S_N-S_T) - F_T(H(r,t))S_0 + \sigma_T\xi_T(t),
    \label{eq:amoc_sde_posq_ST}
\end{align}

For $Q < 0$,
\begin{align}
    V_N \frac{dS_N}{dt} &= |Q|(S_B-S_N) + K_N(S_T-S_N) - F_N(H(r,t))S_0 + \sigma_N\xi_N(t),
    \label{eq:amoc_sde_negq_SN}\\
    V_T \frac{dS_T}{dt} &= |Q|(S_N-S_T)
    + K_S(S_S-S_T) + K_N(S_N-S_T) - F_T(H(r,t))S_0 + \sigma_T\xi_T(t).
    \label{eq:amoc_sde_negq_ST}
\end{align}

Here, $F_N$ and $F_T$ denote the freshwater-flux terms induced by the hosing protocol $H(r,t)$, with parameters given in \cite{ritchie2023rate}. Different from earlier formulation, we transform the deterministic box-model as stochastic system where noise amplitudes are set to $\sigma_N=\sigma_T=1.0$.

For the B-tipping configuration, we use a linear schedule for the bifurcation parameter $H$: tipping runs ramp $H$ from $0.0$ to $0.4736$, whereas non-tipping runs follow the same ramp but are capped at $H=0.3736$. The deterministic fold occurs at $H^\star=0.4236$~\cite{alkhayuon2019basin}. For the R-tipping configuration, we use
\begin{equation}
    H(r,t)=
    \begin{cases}
        H_0 + \Delta H\,\mathrm{sech}\!\big(r(t-t_{\mathrm{crit}})\big), & t<t_{\mathrm{crit}},\\
        H_0 + \Delta H, & t\ge t_{\mathrm{crit}},
    \end{cases}
    \label{eq:amoc_rate_forcing}
\end{equation}

so that tipping and non-tipping trajectories differ only through the forcing rate $r$. We use $r=0.017$ for tipping and $r=0.005$ for non-tipping, with $t_{\mathrm{crit}}=500$. 

\subsection*{Empirical observations}

\paragraph{microcosm} We analyze a cyanobacteria population in chemostats subjected to dilution events under gradually increasing light levels~\cite{veraart2012recovery}. Population density is inferred from the light attenuation coefficient, computed from continuous measurements of outgoing light intensity. The full dataset comprises 7,784 observations over 28.86 days at a sampling interval of 5 minutes. The time series is divided into six segments separated by dilution events, which are external perturbations and not part of the intrinsic population dynamics. To ensure consistency, we restrict analysis to recovery phases between dilution events. Specifically, for each segment, we use the final 250 time points (approximately one day) prior to the next dilution event, yielding a total of 1,500 data points across all segments.

\paragraph{voice} We analyze experimental time series of phonation onset, corresponding to the emergence of vocal fold oscillations, which can be interpreted as a Hopf bifurcation~\cite{mergell1998phonation,murray2012vibratory}. The data are obtained from a physical replica of human vocal folds (EPI model) with a multilayer body–cover structure that reproduces realistic tissue mechanics. Oscillations are induced by gradually increasing airflow, while subglottal pressure is recorded using a pressure transducer. The available time series (from \cite{grziwotz2023anticipating}) consists of 7,501 data points over 0.375\,s with a sampling interval of $5 \times 10^{-5}$\,s. Tipping corresponds to the transition from a stable, non-oscillatory state to sustained oscillations as the flow rate crosses a critical threshold. We use the measured pressure signal as the observable for tipping detection.

\paragraph{cellular\_atp} We analyze experimental time series of cytosolic ATP (i.e., the readily available energy pool) dynamics in living plant cells, measured via a genetically encoded FRET sensor sensitive to MgATP$^{2-}$~\cite{wagner2019multiparametric}. The data capture the response of leaf tissue under gradually increasing hypoxia, induced by sealing samples in a dark environment where respiration depletes available oxygen. ATP levels remain relatively stable during oxygen decline and then undergo a sudden collapse, indicating a critical transition in cellular energy state. The dataset consists of 271 observations over 3 minutes with a sampling interval of 0.011 minutes. We use the fluorescence ratio signal as the observable, with tipping identified as the abrupt drop in ATP concentration.

\paragraph{greenhouse\_earth} We analyze a paleoclimate time series of calcium carbonate (CaCO$_3$) associated with the end of greenhouse Earth, marking the transition from an ice-free state to the formation of polar ice caps~\cite{dakos2008slowing}. The data exhibit increasing autocorrelation prior to the climate shift, consistent with critical slowing down. The dataset consists of 462 observations spanning 5.9 million years with a sampling interval of 0.013 million years. Tipping is identified as the transition to a glaciated climate state.

\paragraph{blackout\_frequency} We analyze bus voltage frequency data preceding the Western Interconnect blackout of August 1996, measured within the Bonneville Power Administration network~\cite{council1996western}. The time series captures system dynamics leading up to grid separation and exhibits critical fluctuations prior to the blackout, consistent with EWS of instability. The dataset consists of 11,301 observations over 565 seconds with a sampling interval of 0.05\,s. Tipping is identified as the transition to system-wide failure, i.e., blackout.

\section*{Declarations}
\subsection*{Acknowledgements}
JN, HF, PG acknowledge funding, computing, and storage resources from the NSF Science and Technology Center (STC) Learning the Earth with Artificial Intelligence and Physics (LEAP) (Award \#2019625). The material is based upon work supported by NASA under award No 80NSSC25K0062. VL acknowledges the partial support provided by the Horizon Europe Projects Past2Future (Grant No. 101184070) and ClimTIP (Grant No. 100018693), by the ARIA SCOP-PR01-P003 - Advancing Tipping Point Early Warning AdvanTip project, and by the European Space Agency Project PREDICT (Contract 4000146344/24/I-LR). Further support was provided by the U.S. Department of Energy, Office of Science, Office of Advanced Scientific Computing Research and Office of Biological and Environmental Research, Scientific Discovery through Advanced Computing (SciDAC) program under Award Number 9233218CNA000001. PK acknowledges the support from U.S. Department of Energy (DOE), Office of Science, Office of Biological and Environmental Research, Regional and Global Model Analysis program area as part of the HiLAT-RASM project.

\subsection*{Competing interests}
All authors declare no financial or non-financial competing interests. 

\subsection*{Ethics approval and consent to participate}
Not applicable.

\subsection*{Consent for publication}
Not applicable.

\subsection*{Data Availability}
The toy datasets have been generated with custom code (see the Code Availability statement) and simulation data to reproduce the AMOC results is freely available from the reference in the text. 

\subsection*{Materials availability}
Not applicable.

\subsection*{Code Availability}
The code for this study will be available at \url{https://github.com/juannat7/kews}.

\subsection*{Author contributions}
JN and CR conceived of the study and performed the analysis under the guidance of PG. DD, PK, and VL assisted in the methodological formulation. HF and AR provided domain science expertise to frame and interpret the results. The first draft of the manuscript was written by JN and CR, and all authors commented on the manuscript.

\bibliographystyle{unsrt}
\bibliography{references}

\clearpage
\appendix

\begin{center}
    {\Large\bfseries Supplementary Information for\\[0.5em]
    Koopman early warning signals for bifurcation and rate-induced tipping}

    \vspace{1.5em}

    {\normalsize
    Juan Nathaniel\textsuperscript{1,*} \quad
    Carla Roesch\textsuperscript{2} \quad
    Derek DeSantis\textsuperscript{3} \quad
    Parvathi Kooloth\textsuperscript{4} \quad
    Hang Fan\textsuperscript{1} \\[0.5em]
    Valerio Lucarini\textsuperscript{5} \quad
    Anastasia Romanou\textsuperscript{6} \quad
    Pierre Gentine\textsuperscript{1}
    }

    \vspace{1em}

    {\small
    \textsuperscript{1}Columbia University, USA \quad
    \textsuperscript{2}University of Edinburgh, UK \quad
    \textsuperscript{3}Los Alamos National Laboratory, USA \\[0.25em]
    \textsuperscript{4}Pacific Northwest National Laboratory, USA \quad
    \textsuperscript{5}University of Leicester, UK \quad
    \textsuperscript{6}NASA Goddard Institute for Space Studies, USA \\[0.5em]
    \textsuperscript{*}Corresponding author: \texttt{jn2808@columbia.edu}
    }
\end{center}

\vspace{2em}

\setcounter{section}{0}
\setcounter{equation}{0}
\setcounter{figure}{0}
\setcounter{table}{0}

\renewcommand{\thesection}{S\arabic{section}}
\renewcommand{\theequation}{S\arabic{equation}}
\renewcommand{\thefigure}{S\arabic{figure}}
\renewcommand{\thetable}{S\arabic{table}}

\renewcommand{\theHsection}{supp.section.\arabic{section}}
\renewcommand{\theHequation}{supp.equation.\arabic{equation}}
\renewcommand{\theHfigure}{supp.figure.\arabic{figure}}
\renewcommand{\theHtable}{supp.table.\arabic{table}}

\section{Spectral view of classical indicators}
The classical CSD interpretation of AR1 and variance derived below is inherently equilibrium-centered: one linearizes the dynamics around a stable fixed point and studies the dominant eigenvalue of the resulting Jacobian. This provides the mechanism by which slowing recovery implies increasing autocorrelation and variance near bifurcation. However, this reasoning assumes that the local stability structure is well approximated by a nearly fixed equilibrium over the analysis window. For methods to estimate a local, time-varying Jacobian from data around the observed state itself, see discussion in \cite{grziwotz2023anticipating}. This class of methods extend Jacobian-based warning analysis beyond the strict requirement of first locating the equilibrium. The underlying bifurcation intuition remains the same, but the estimation is performed locally along the trajectory rather than exclusively at an equilibrium.

\subsection{Spectral view of AR1}
\label{si-sec:ar1_derivation}
\begin{align}
    \mathrm{Corr}(z_t, z_{t+1})
    &= \frac{\mathrm{Cov}(z_t, z_{t+1})}{\sqrt{\mathrm{Var}(z_t)\mathrm{Var}(z_{t+1})}} \nonumber\\
    &= \frac{\mathrm{Cov}(z_t, \lambda_{\mathcal{J},1}(\beta)z_t + \eta_t)}{\mathrm{Var}(z_t)} \nonumber\\
    &= \frac{\lambda_{\mathcal{J},1}(\beta)\mathrm{Var}(z_t) + \mathrm{Cov}(z_t,\eta_t)}{\mathrm{Var}(z_t)} \nonumber \\
    &= \lambda_{\mathcal{J},1}(\beta),
\end{align}

where the second line applies stationarity to simplify the denominator, and the third uses uncorrelatedness of innovation to eliminate the cross term. As $\rho(\mathcal{J}(\beta)) \to 1$, perturbations decay increasingly slowly, indicating greater persistence. For a real dominant eigenvalue, this implies $|\mathrm{Corr}(z_t, z_{t+1})| \to 1$. In particular, $\mathrm{Corr}(z_t,z_{t+1}) \to 1$ when the leading eigenvalue approaches $+1$, as in a saddle-node (fold) bifurcation, whereas $\mathrm{Corr}(z_t,z_{t+1}) \to -1$ when the leading eigenvalue approaches $-1$, as in a flip (period-doubling) bifurcation.
A complete derivation of the link between closure of the spectral gap and  occurrence of $|\mathrm{Corr}(z_t,z_{t+1})| \to 1$ is presented in \cite{LucariniChekroun2023,Tantet2018,chekroun2019c}.

\subsection{Spectral view of variance}
\label{si-sec:var_derivation}
\begin{align}
    \mathrm{Var}(z_{t+1})
    &= \mathrm{Var}(\lambda_{\mathcal{J},1}(\beta)z_t + \eta_t)\nonumber\\
    &= |\lambda_{\mathcal{J},1}(\beta)|^2\mathrm{Var}(z_t)
       + \mathrm{Var}(\eta_t)
       + 2\,\mathrm{Re}\!\big(\lambda_{\mathcal{J},1}(\beta)\mathrm{Cov}(z_t,\eta_t)\big)\nonumber\\
    &= |\lambda_{\mathcal{J},1}(\beta)|^2\mathrm{Var}(z_t) + \mathrm{Var}(\eta_t)\nonumber\\
    (1 - |\lambda_{\mathcal{J},1}(\beta)|^2)\mathrm{Var}(z_t) &= \mathrm{Var}(\eta_t)\nonumber\\
    \mathrm{Var}(z_t) &= \frac{\mathrm{Var}(\eta_t)}{1 - |\lambda_{\mathcal{J},1}(\beta)|^2},
\end{align}

where the second line uses the uncorrelatedness of innovation to drop the cross term, and the third applies stationarity. As $\rho(\mathcal{J}(\beta)) \to 1$, the variance term increases and diverges in the linear approximation. A more general derivation of this property is presented in \cite{Lucarini2012}.

\section{Derivation and Proof}
\subsection{Derivation of the R-tipping's tracking error}
\label{si-sec:rate_perturb_dynamics}

Given $\bar{\omega}_t:=\omega_t-\omega^\star(u_t)$, we have
\begin{equation}
\begin{aligned}
    \bar{\omega}_{t+1}
    &= \omega_{t+1} - \omega^\star(u_{t+1}) \\
    &= \Phi(\omega_t,u_t) + \varepsilon_t - \omega^\star(u_{t+1}) \\
    &= \Phi(\omega^\star(u_t)+\bar{\omega}_t,u_t) + \varepsilon_t - \omega^\star(u_{t+1}) \\
    &= \Big(\Phi(\omega^\star(u_t)+\bar{\omega}_t,u_t)
           - \Phi(\omega^\star(u_t),u_t)\Big)
           + \Phi(\omega^\star(u_t),u_t)
           - \omega^\star(u_{t+1})
           + \varepsilon_t\\
    &= \Big(\Phi(\omega^\star(u_t)+\bar{\omega}_t,u_t)
           - \Phi(\omega^\star(u_t),u_t)\Big)
           + \underbrace{\omega^\star(u_t)-\omega^\star(u_{t+1})}_{d_t}
           + \varepsilon_t,
\end{aligned}
\end{equation}

where the fourth line adds and subtracts $\Phi(\omega^\star(u_t),u_t)$ and the fifth line uses the equilibrium identity $\Phi(\omega^\star(u_t),u_t)=\omega^\star(u_t)$. Linearizing $\Phi(\cdot,u_t)$ around $\omega^\star(u_t)$ gives
\begin{equation}
    \Phi(\omega^\star(u_t)+\bar{\omega}_t,u_t)
    =
    \Phi(\omega^\star(u_t),u_t)
    +
    \left.\frac{\partial \Phi}{\partial \omega}\right|_{\omega=\omega^\star(u_t)}\bar{\omega}_t
    +
    \mathcal{O}(\|\bar{\omega}_t\|^2),
\end{equation}

and hence
\begin{equation}
    \bar{\omega}_{t+1}
    =
    \underbrace{\mathcal{J}(u_t)\bar{\omega}_t}_{\text{local restoration}}
    +
    \underbrace{d_t}_{\text{branch drift}}
    +
    \underbrace{\varepsilon_t+\mathcal{O}(\|\bar{\omega}_t\|^2)}_{\epsilon_t}.
\end{equation}

\subsection{Proof of Proposition~\ref{prop:reskmdc_decomp}}
\label{si-sec:proof-reskmdc-decomp}

\begin{proof}
First we define
\begin{equation}
    Y(\omega,u):=\boldsymbol{\Psi}(\Phi_\varepsilon(\omega,u),g(u)),
    \qquad
    \hat{Y}(\omega,u):=[\mathcal{K}^{(1)}_{\mathcal{S}_m^c}\boldsymbol{\Psi}](\omega,u).
\end{equation}

For fixed $(\omega,u)$, add and subtract $\mathbb{E}[Y\mid\omega,u]$
\begin{equation}
    Y-\hat{Y}
    =
    \bigl(Y-\mathbb{E}[Y\mid\omega,u]\bigr)
    +
    \bigl(\mathbb{E}[Y\mid\omega,u]-\hat{Y}\bigr).
\end{equation}

Expanding the squared Euclidean norm gives
\begin{align}
    \|Y-\hat{Y}\|_2^2
    &=
    \left\|Y-\mathbb{E}[Y\mid\omega,u]\right\|_2^2
    +
    \left\|\mathbb{E}[Y\mid\omega,u]-\hat{Y}\right\|_2^2 \nonumber\\
    &\quad
    +
    2\operatorname{Re}
    \left\langle
        Y-\mathbb{E}[Y\mid\omega,u],
        \mathbb{E}[Y\mid\omega,u]-\hat{Y}
    \right\rangle.
\end{align}

Conditioning on $(\omega,u)$ and taking expectations yields
\begin{align}
    \mathbb{E}\!\left[\|Y-\hat{Y}\|_2^2\mid\omega,u\right]
    &=
    \mathbb{E}\!\left[
        \left\|Y-\mathbb{E}[Y\mid\omega,u]\right\|_2^2
        \middle|\omega,u
    \right]
    +
    \left\|\mathbb{E}[Y\mid\omega,u]-\hat{Y}\right\|_2^2 \nonumber\\
    &\quad
    +
    2\operatorname{Re}
    \left\langle
        \mathbb{E}\!\left[Y-\mathbb{E}[Y\mid\omega,u]\middle|\omega,u\right],
        \mathbb{E}[Y\mid\omega,u]-\hat{Y}
    \right\rangle.
\end{align}

The cross term vanishes because
\begin{equation}
    \mathbb{E}\!\left[Y-\mathbb{E}[Y\mid\omega,u]\middle|\omega,u\right]=0.
\end{equation}

Therefore,
\begin{equation}
    \mathbb{E}\!\left[\|Y-\hat{Y}\|_2^2\mid\omega,u\right]
    =
    \left\|\mathbb{E}[Y\mid\omega,u]-\hat{Y}\right\|_2^2
    +
    \mathbb{E}\!\left[
        \left\|Y-\mathbb{E}[Y\mid\omega,u]\right\|_2^2
        \middle|\omega,u
    \right].
    \label{eq:bias_variance_control}
\end{equation}

Since $\mathbb{E}[Y\mid \omega,u] = [\mathcal{K}^{(1)}_c\boldsymbol{\Psi}](\omega,u)$ by the definition of the control-aware first-moment Koopman operator in Equation~\ref{eq:kic_first_moment}, the first term in Equation~\ref{eq:bias_variance_control} is exactly the squared truncation error at $(\omega,u)$. 
\begin{equation}
    \left\|\mathbb{E}[Y\mid \omega,u]-\hat{Y}\right\|_2^2
    =
    \left\|[\mathcal{K}^{(1)}_c\boldsymbol{\Psi}](\omega,u) - [\mathcal{K}^{(1)}_{\mathcal{S}^c_m}\boldsymbol{\Psi}](\omega,u)\right\|_2^2
\end{equation}

The second term equals the trace of the conditional covariance,
\begin{equation}
    \mathbb{E}\left[\left\|Y-\mathbb{E}[Y\mid \omega,u]\right\|_2^2 \middle|\omega,u \right] 
    =
    \mathrm{Tr}\left[\mathrm{Cov}(Y\mid \omega,u)\right].
\end{equation}

Integrating Equation~\ref{eq:bias_variance_control} over the joint state-control space $\Omega\times\mathbb{R}^U$ with respect to measure $\mu$ gives
\begin{align}
  \mathrm{res}\left[
      \mathcal{K}^{(1)}_c,
      \mathcal{K}^{(1)}_{\mathcal{S}^c_m};
      \boldsymbol{\Psi}
  \right]^2
  &=
  \left\|
      \mathcal{K}^{(1)}_c\boldsymbol{\Psi}
      -
      \mathcal{K}^{(1)}_{\mathcal{S}^c_m}\boldsymbol{\Psi}
  \right\|_{L^2(\Omega\times\mathbb{R}^U,\mu;\mathbb{C}^M)}^2
  \nonumber\\
  &\quad
  +
  \int_{\Omega\times\mathbb{R}^U}
  \mathrm{Tr}\left[
      \mathrm{Cov}\left(
          \boldsymbol{\Psi}\left(\Phi_\varepsilon(\omega,u),g(u)\right)
          \middle|\omega,u
      \right)
  \right]
  d\mu(\omega,u),
\end{align}

which proves Proposition~\ref{prop:reskmdc_decomp}.
\end{proof}

Throughout, expectations and covariances in the control-aware setting are understood with respect to the one-step stochastic transition from the current pair $(\omega,u)$, and are integrated over $(\omega,u)$ against the reference measure $\mu$.

\subsection{Counterexample for Remark~\ref{rem:no_general_divergence}}
\label{si-sec:counterexample-rtipping}

Consider the linear forced system
\begin{equation}
    \omega_{t+1}=a\omega_t+u_t,
    \qquad
    u_{t+1}=u_t+r,
    \qquad 0<a<1.
    \label{eq:counterexample_system}
\end{equation}

For fixed $u$, the frozen equilibrium $\omega^\star(u)$ satisfies
\begin{align}
    \omega^\star(u)&=a\omega^\star(u)+u\\
    \omega^\star(u)&=\frac{u}{1-a}
    \label{eq:counterexample_equilibrium}
\end{align}

Next, we define the tracking error
\begin{equation}
    \bar{\omega}_t:=\omega_t-\omega^\star(u_t)=\omega_t-\frac{u_t}{1-a}.
\end{equation}

Then
\begin{align}
    \bar{\omega}_{t+1}
    &= \omega_{t+1}-\omega^\star(u_{t+1}) \nonumber\\
    &= a\omega_t+u_t-\frac{u_t+r}{1-a} \nonumber\\
    &= a\left(\bar{\omega}_t+\frac{u_t}{1-a}\right)+u_t-\frac{u_t+r}{1-a}\nonumber\\
    &= a\bar{\omega}_t-\frac{r}{1-a}.
    \label{eq:counterexample_tracking}
\end{align}

Thus the frozen dynamics remain stable as $0 < a < 1$, but sufficiently large $r$ prevents tracking.

We now consider the following observable dictionary
\begin{equation}
    \boldsymbol{\Psi}(\omega,u)
    :=
    \begin{bmatrix}
        1\\
        \omega\\
        u
    \end{bmatrix}.
\end{equation}

Its one-step lifted evolution is
\begin{equation}
    \boldsymbol{\Psi}(\omega_{t+1},u_{t+1})
    =
    \begin{bmatrix}
        1\\
        \omega_{t+1}\\
        u_{t+1}
    \end{bmatrix}
    =
    \begin{bmatrix}
        1\\
        a\omega_t+u_t\\
        u_t+r
    \end{bmatrix}
    =
    K\boldsymbol{\Psi}(\omega_t,u_t),
\end{equation}

with
\begin{equation}
    K=
    \begin{bmatrix}
        1 & 0 & 0\\
        0 & a & 1\\
        r & 0 & 1
    \end{bmatrix}.
    \label{eq:counterexample_koopman_matrix}
\end{equation}

Hence the augmented dynamics close exactly in the chosen observable space. Because the system is deterministic, both the truncation error and stochastic fluctuation vanish, and therefore
\begin{equation}
    \mathrm{res}\!\left[
      \mathcal{K}^{(1)}_c,
      \mathcal{K}^{(1)}_{\mathcal{S}};
      \boldsymbol{\Psi}
    \right]=0,
\end{equation}

even though the trajectory fails to track the moving equilibrium branch. 

\subsection{Proof of Theorem~\ref{thm:rtip_lower_bound}}
\label{si-sec:proof-rtip_lower_bound}

\begin{proof}
Since each component of the local restoration, $\mathcal{J}(u)\bar{\omega}(\omega,u)$ lies in $\mathcal{S}_m^c$, its unresolved projection onto $\mathcal{Q}_m$ gives zero vector. Applying $\mathcal{Q}_m$ to Equation~\ref{eq:first_moment_decomp} therefore gives
\begin{equation}
    \mathcal{Q}_m[\mathcal{K}^{(1)}_c\bar{\omega}]
    =
    \mathcal{Q}_m \mathcal{J}(u)\bar{\omega}(\omega,u)+\mathcal{Q}_m d+\mathcal{Q}_m\epsilon
    =
    \mathcal{Q}_m d+\mathcal{Q}_m\epsilon.
\end{equation}

Hence, by the reverse triangle inequality,
\begin{equation}
    \left\|
        \mathcal{Q}_m[\mathcal{K}^{(1)}_c\bar{\omega}]
    \right\|_{L^2(\Omega\times\mathbb{R}^U,\mu;\mathbb{R}^N)}
    \ge
    \|\mathcal{Q}_m d\|_{L^2(\Omega\times\mathbb{R}^U,\mu;\mathbb{R}^N)}
    -
    \|\mathcal{Q}_m \epsilon\|_{L^2(\Omega\times\mathbb{R}^U,\mu;\mathbb{R}^N)}.
\end{equation}

From $\bar{\omega}=C\boldsymbol{\Psi}$ and linearity of $C$ and $\mathcal{K}^{(1)}_c$, we have
\begin{equation}
    \mathcal{K}^{(1)}_c\bar{\omega}=C(\mathcal{K}^{(1)}_c\boldsymbol{\Psi}).
\end{equation}

Since $\mathcal{K}^{(1)}_{\mathcal{S}_m^c}\boldsymbol{\Psi}\in(\mathcal{S}_m^c)^M$, each component of $C(\mathcal{K}^{(1)}_{\mathcal{S}_m^c}\boldsymbol{\Psi})$ also lies in $\mathcal{S}_m^c$. By the best-approximation property of $\mathcal{P}_m$,
\begin{align}
    \left\|
        \mathcal{Q}_m[\mathcal{K}^{(1)}_c\bar{\omega}]
    \right\|_{L^2(\Omega\times\mathbb{R}^U,\mu;\mathbb{R}^N)}
    &\le
    \left\|
        \mathcal{K}^{(1)}_c\bar{\omega}
        -
        C\bigl(\mathcal{K}^{(1)}_{\mathcal{S}_m^c}\boldsymbol{\Psi}\bigr)
    \right\|_{L^2(\Omega\times\mathbb{R}^U,\mu;\mathbb{R}^N)} \nonumber\\
    &=
    \left\|
        C\bigl(
            \mathcal{K}^{(1)}_c\boldsymbol{\Psi}
            -
            \mathcal{K}^{(1)}_{\mathcal{S}_m^c}\boldsymbol{\Psi}
        \bigr)
    \right\|_{L^2(\Omega\times\mathbb{R}^U,\mu;\mathbb{R}^N)} \nonumber\\
    &\le
    \|C\|
    \left\|
        \mathcal{K}^{(1)}_c\boldsymbol{\Psi}
        -
        \mathcal{K}^{(1)}_{\mathcal{S}_m^c}\boldsymbol{\Psi}
    \right\|_{L^2(\Omega\times\mathbb{R}^U,\mu;\mathbb{C}^M)}.
\end{align}

Combining the two inequalities gives
\begin{equation}
    \left\|
        \mathcal{K}^{(1)}_c\boldsymbol{\Psi}
        -
        \mathcal{K}^{(1)}_{\mathcal{S}_m^c}\boldsymbol{\Psi}
    \right\|_{L^2(\Omega\times\mathbb{R}^U,\mu;\mathbb{C}^M)}
    \ge
    \frac{1}{\|C\|}
    \left(
        \|\mathcal{Q}_m d\|_{L^2(\Omega\times\mathbb{R}^U,\mu;\mathbb{R}^N)}
        -
        \|\mathcal{Q}_m\epsilon\|_{L^2(\Omega\times\mathbb{R}^U,\mu;\mathbb{R}^N)}
    \right).
\end{equation}

By Proposition~\ref{prop:reskmdc_decomp}, the residual is bounded below by the truncation error because the stochastic fluctuation term is nonnegative. This proves Equation~\ref{eq:drift_lower_bound}. If each component of $\epsilon$ also lies in $\mathcal{S}_m^c$, then $\mathcal{Q}_m\epsilon=0$, then Equation~\ref{eq:drift_lower_bound_clean} follows immediately as a special case.
\end{proof}

\clearpage
\newpage

\section{Prototypical Experiments}\label{si-sec:proto_exp}
Below represents the tipping and non-tipping trajectories across prototypical B-tipping and R-tipping systems.

\begin{figure}[h!]
    \centering
    \begin{subfigure}[h]{0.24\linewidth}
        \centering
        \includegraphics[width=\linewidth]{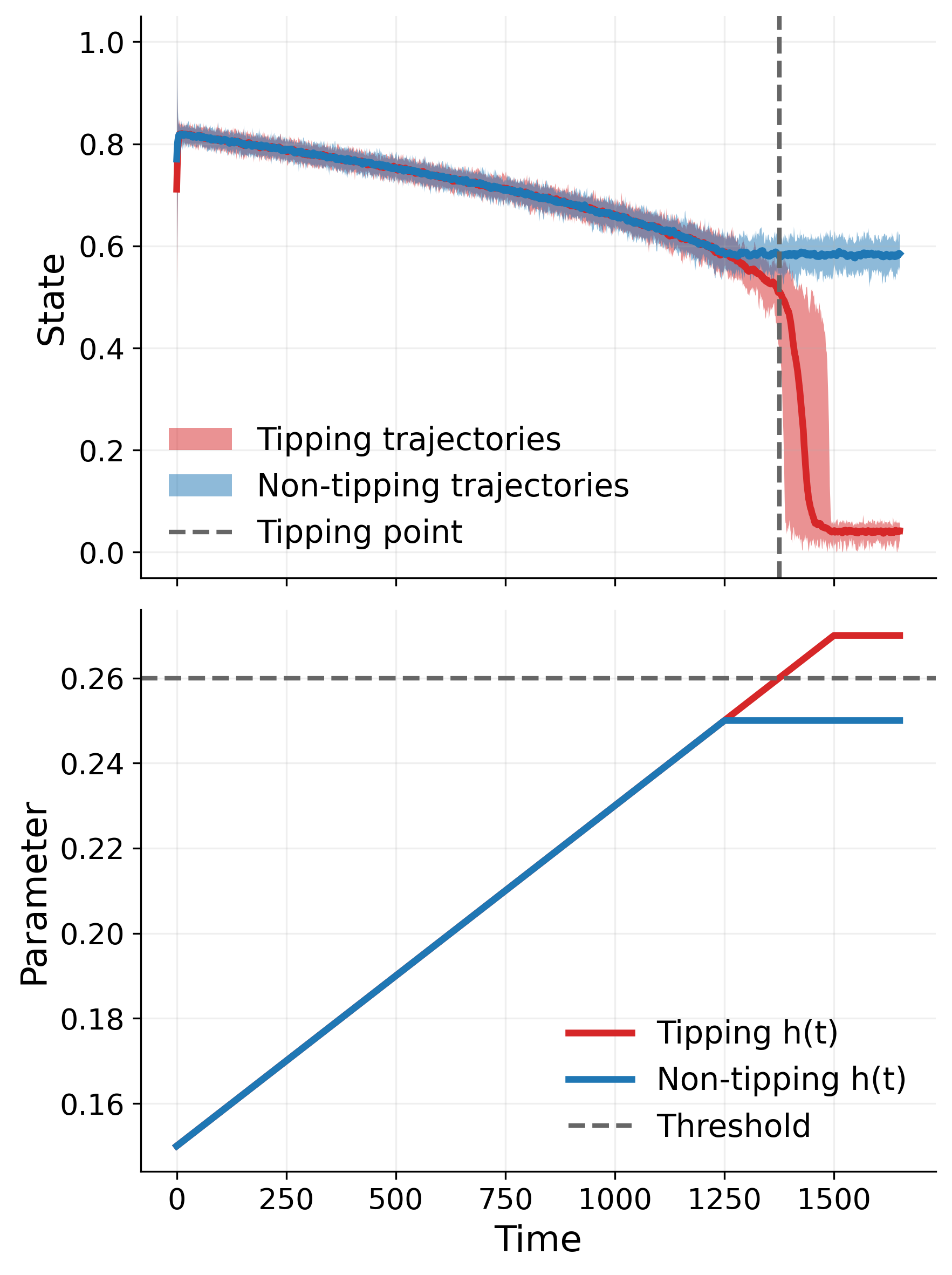}
        \caption{May's harvesting}
        \label{fig:bifurcation_examples_harvesting}
    \end{subfigure}
    \hfill
    \begin{subfigure}[h]{0.24\linewidth}
        \centering
        \includegraphics[width=\linewidth]{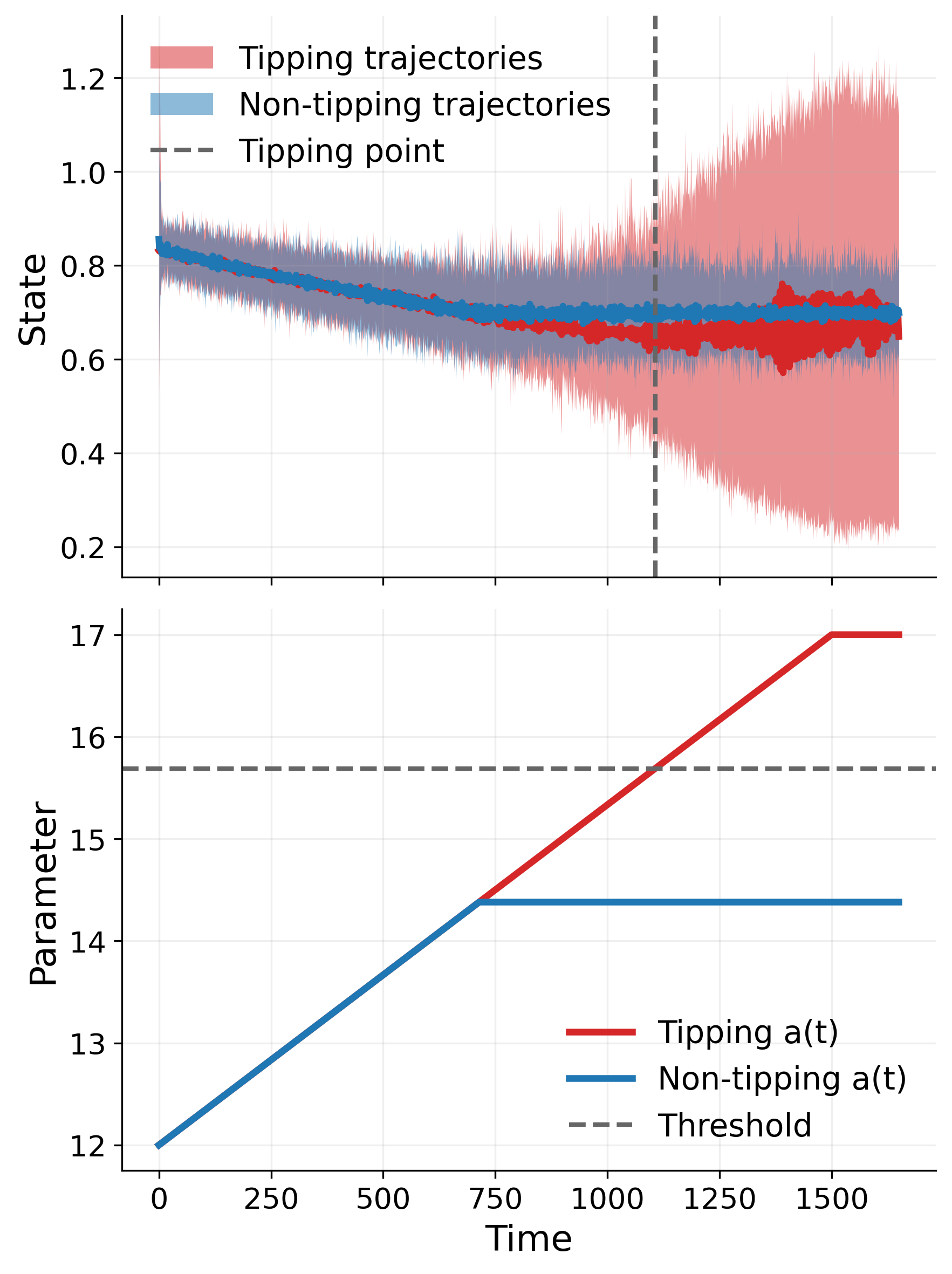}
        \caption{Rosenzweig-MacArthur}
        \label{fig:bifurcation_examples_rm}
    \end{subfigure}
    \hfill
    \begin{subfigure}[h]{0.24\linewidth}
        \centering
        \includegraphics[width=\linewidth]{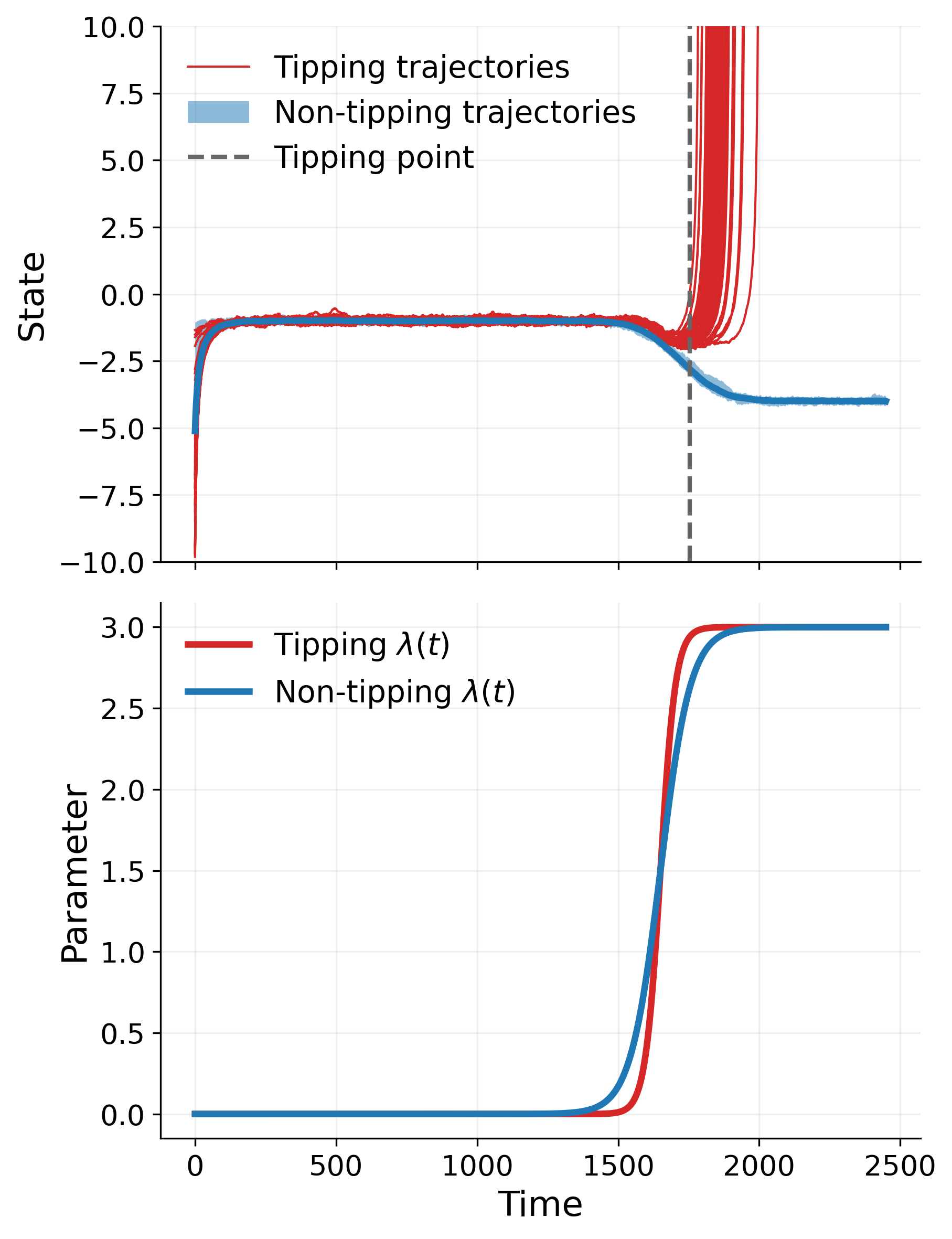}
        \caption{Saddle-node}
        \label{fig:rate_examples_saddle_node}
    \end{subfigure}
    \hfill
    \begin{subfigure}[h]{0.24\linewidth}
        \centering
        \includegraphics[width=\linewidth]{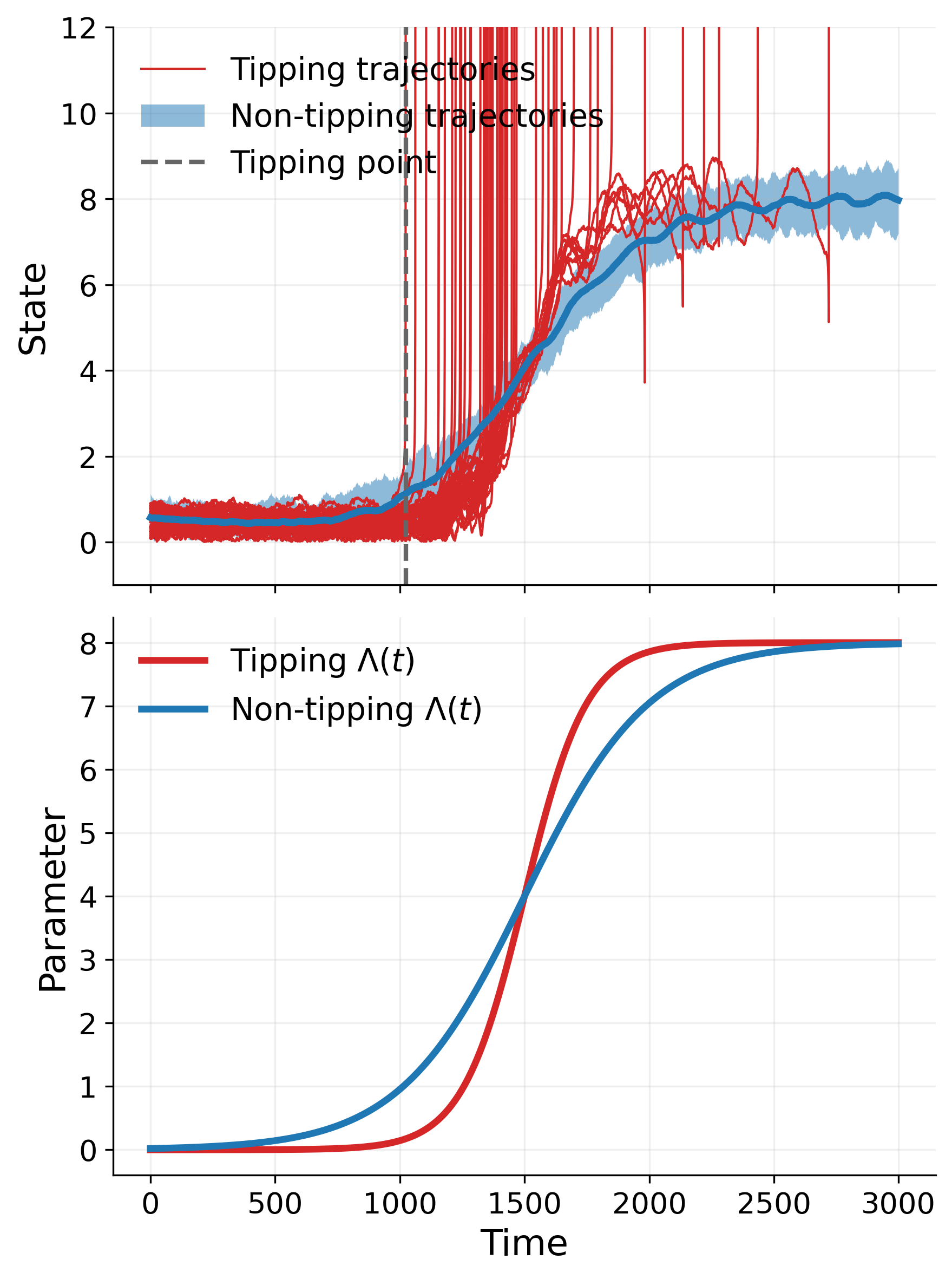}
        \caption{Bautin}
        \label{fig:rate_examples_bautin}
    \end{subfigure}
    \caption{Representative B-tipping (\textbf{a}, \textbf{b}) and R-tipping (\textbf{c}, \textbf{d}) examples. In each panel, the upper plot shows tipping (red) and non-tipping (blue) trajectories, and the lower plot shows the corresponding parameter schedule.}
    \label{fig:bifurcation_examples}
\end{figure}

\subsection{Bifurcation-induced tipping}

\begin{figure}[h!]
    \centering
    \includegraphics[width=\linewidth]{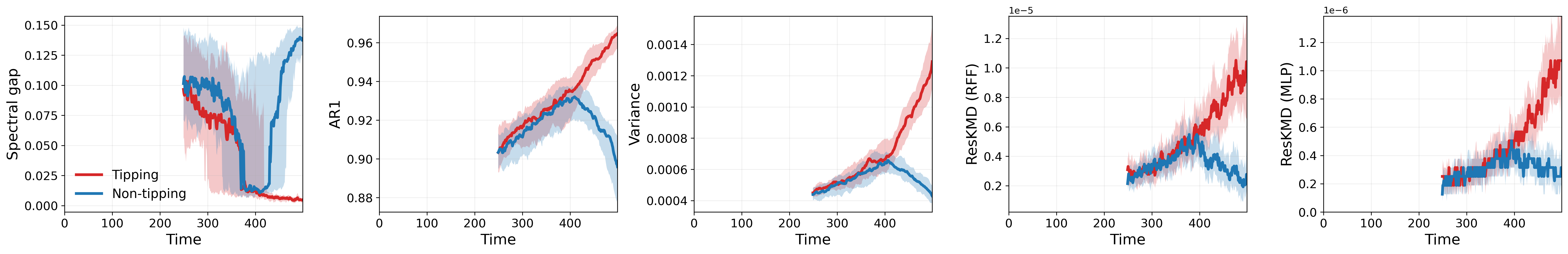}
    \label{fig:bifurcation_ews_harvesting}
    \caption{EWS for B-tipping using May's harvesting model. Panels show the ensemble behavior of five indicators computed on pre-transition sliding windows of length $50\%$: Spectral gap, AR1, variance, ResKMD with RFF, and ResKMD with MLP observable with a hidden size of (16, 8, 4). The shaded region denotes the interquantile range around the median.}
\end{figure}

\subsection{Rate-induced tipping}

\begin{figure}[h!]
    \centering
    \includegraphics[width=\linewidth]{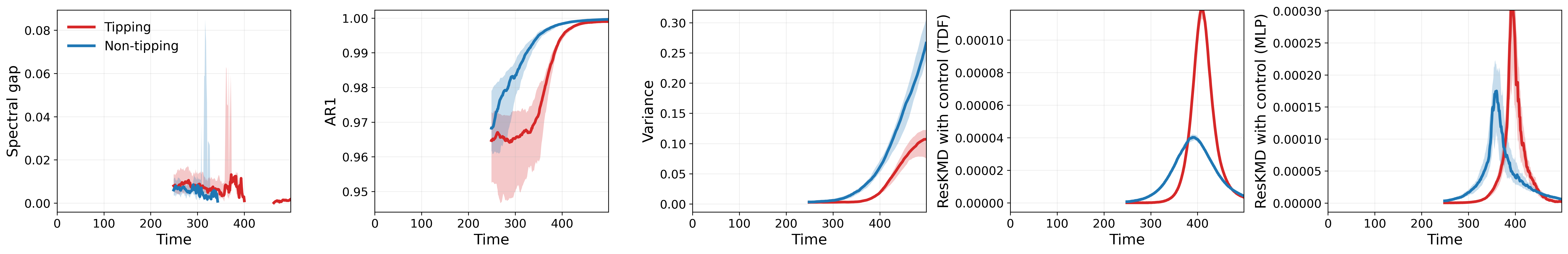}
    \label{fig:rate_ews_saddle_node}
    \caption{EWS for the R-tipping in saddle-node. Panels show the ensemble behavior of five indicators computed on pre-transition sliding windows of length $50\%$: Spectral gap, AR1, variance, and control-aware ResKMD (TDF, MLP observable with a hidden size of (8, 4, 2)). The shaded region denotes the interquantile range around the median.}
\end{figure}

\clearpage
\newpage

\section{Empirical Experiments}\label{si-sec:empirical_exp}
We provide more empirical results showcasing our Koopman-based EWS ability to correctly identify critical transitions. 

\begin{figure}[h!]
    \centering
    \begin{subfigure}[h]{0.9\linewidth}
        \centering
        \includegraphics[width=\linewidth]{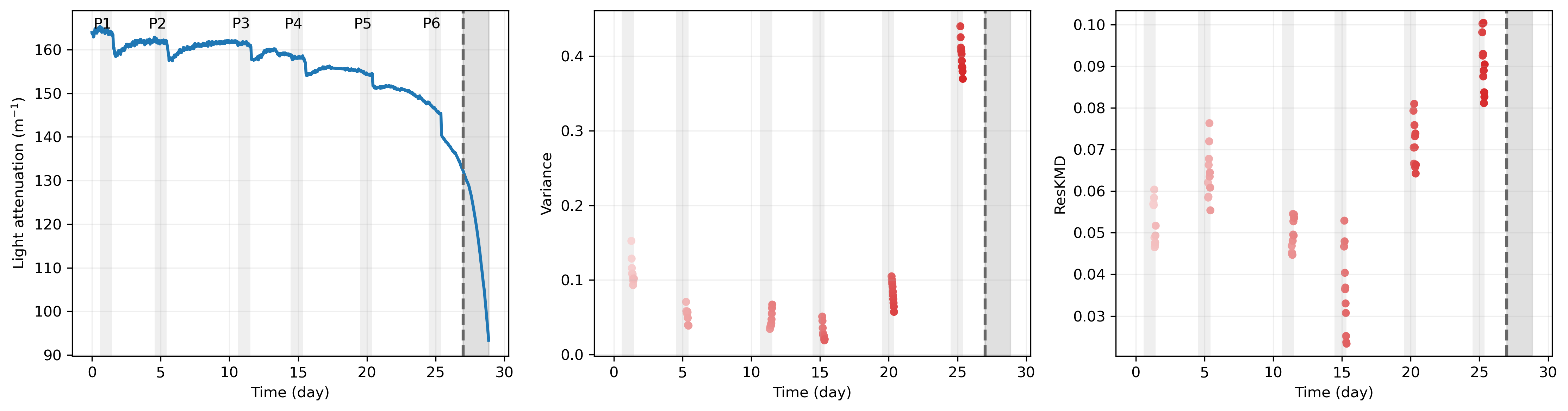}
        \caption{Microcosm experiment}
        \label{fig:empirical_microcosm}
    \end{subfigure}
    \hfill
    \begin{subfigure}[h]{0.9\linewidth}
        \centering
        \includegraphics[width=\linewidth]{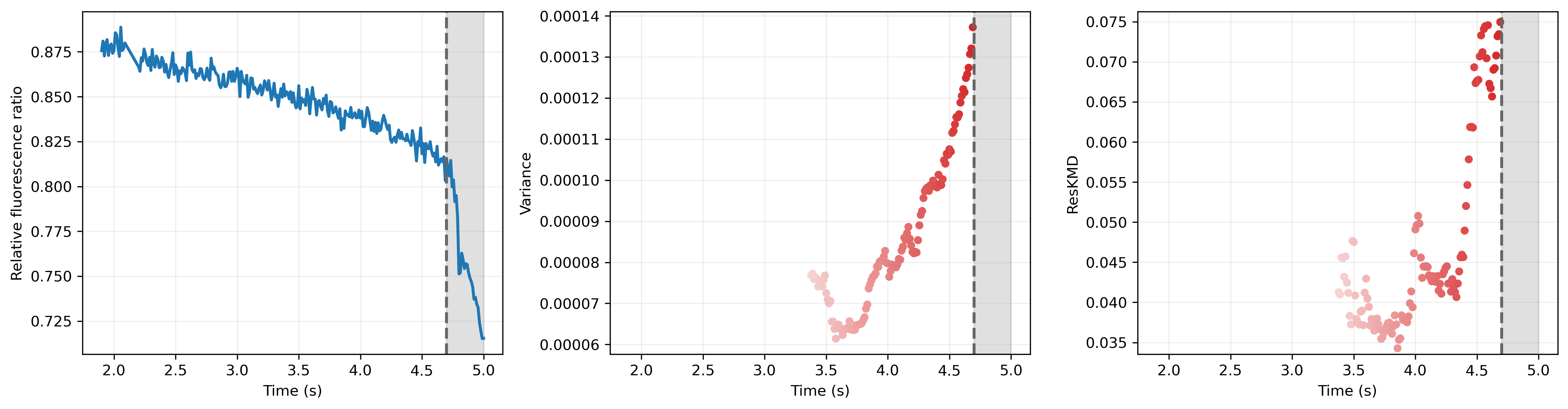}
        \caption{Mitochondria}
        \label{fig:empirical_mitochondria}
    \end{subfigure}
    \hfill
    \caption{Koopman EWS across empirical systems: (\textbf{a}) fold bifurcation in a microcosm experiment where cyanobacteria population undergoes dilution perturbation and increasing light stress, and (\textbf{b}) cytosolic ATP dynamics in living plant cells.}
    \label{si-fig:b_empirical_ews}
\end{figure}

\clearpage
\newpage

\section{Additional Experiments}
\subsection{Coupled AMOC results at $26^\circ$N}
\label{si-sec:amoc_26N}

The main text focuses on $48^\circ$N. For completeness, Figure~\ref{si-fig:amoc_26N} reports the corresponding coupled AMOC results at $26^\circ$N where similar results and conclusions can still be drawn. The control-aware ResKMD again provides better discrimination and earlier detection of early warning signals than AR1 and variance.

\begin{figure}[h!]
    \centering
    \begin{subfigure}[h]{\linewidth}
        \centering
        \includegraphics[width=\linewidth]{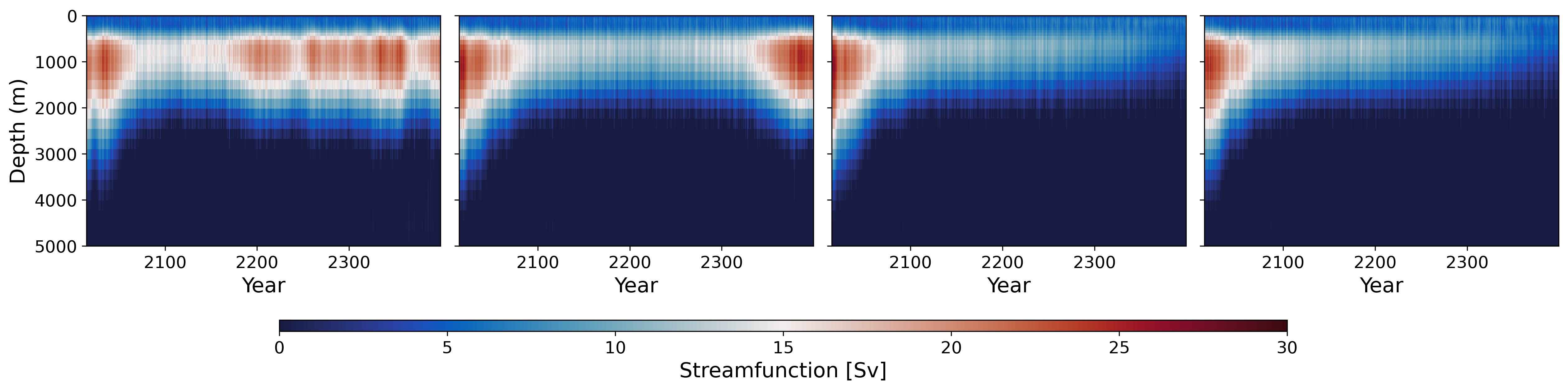}
        \caption{Streamfunction @ $26^\circ$}
        \label{fig:amoc_str_profile_26}
    \end{subfigure}
    \hfill
    \begin{subfigure}[h]{\linewidth}
        \centering
        \includegraphics[width=\linewidth]{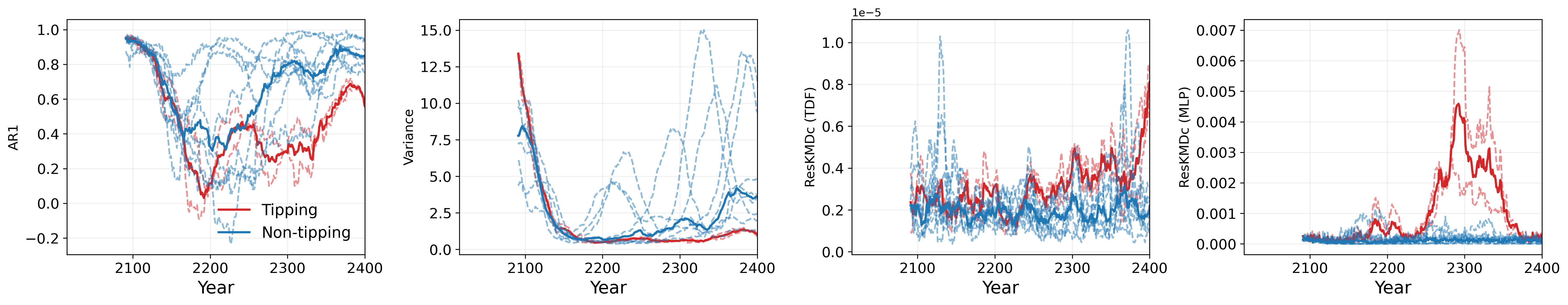}
        \caption{EWS at $26^\circ$N}
        \label{si-fig:amoc_26N_ews}
    \end{subfigure}
    \hfill
    \begin{subfigure}[h]{0.255\linewidth}
        \centering
        \includegraphics[width=\linewidth]{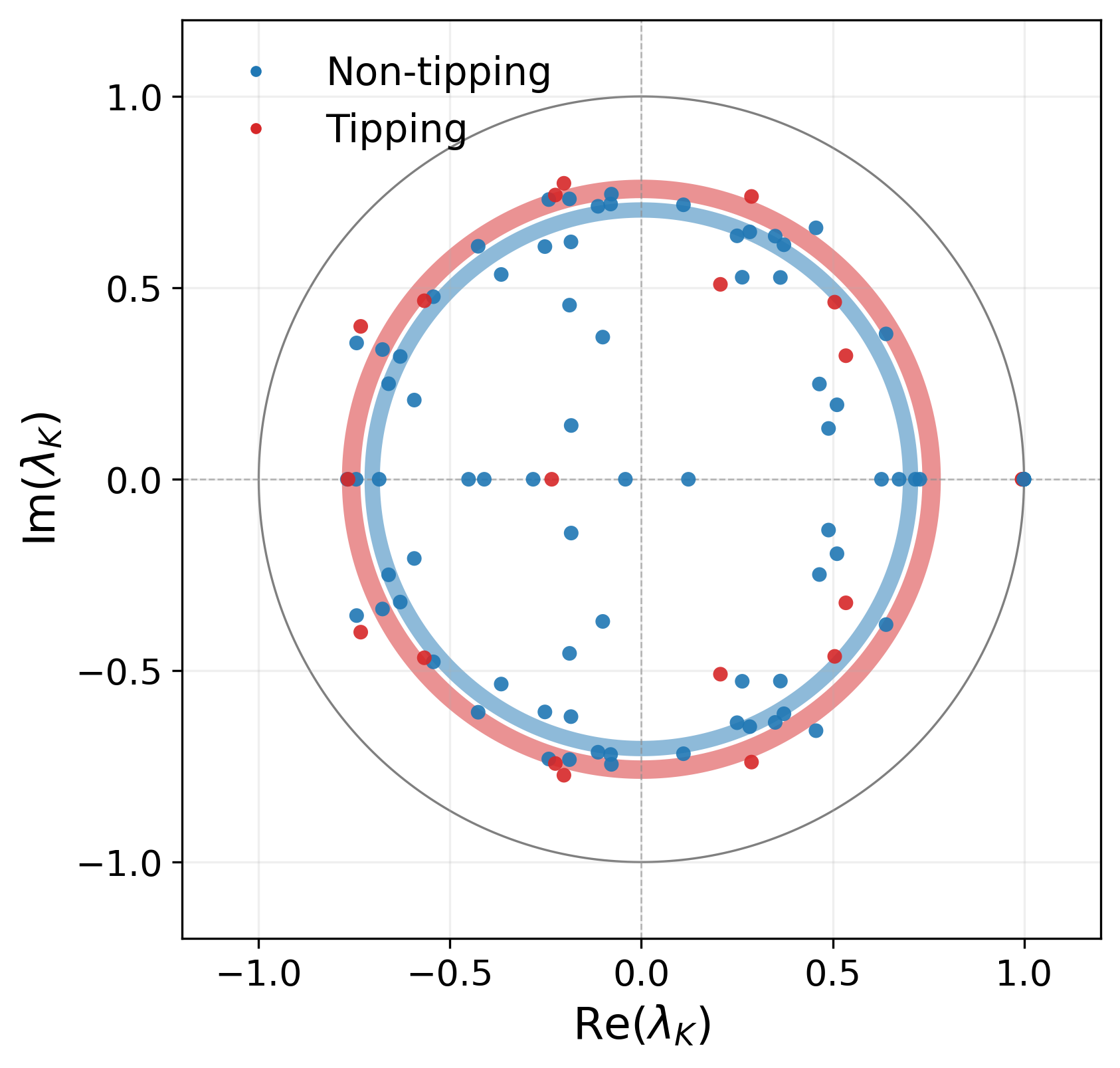}
        \caption{Discrete spectrum}
        \label{si-fig:amoc_26N_disc}
    \end{subfigure}
    \hfill
    \begin{subfigure}[h]{0.240\linewidth}
        \centering
        \includegraphics[width=\linewidth]{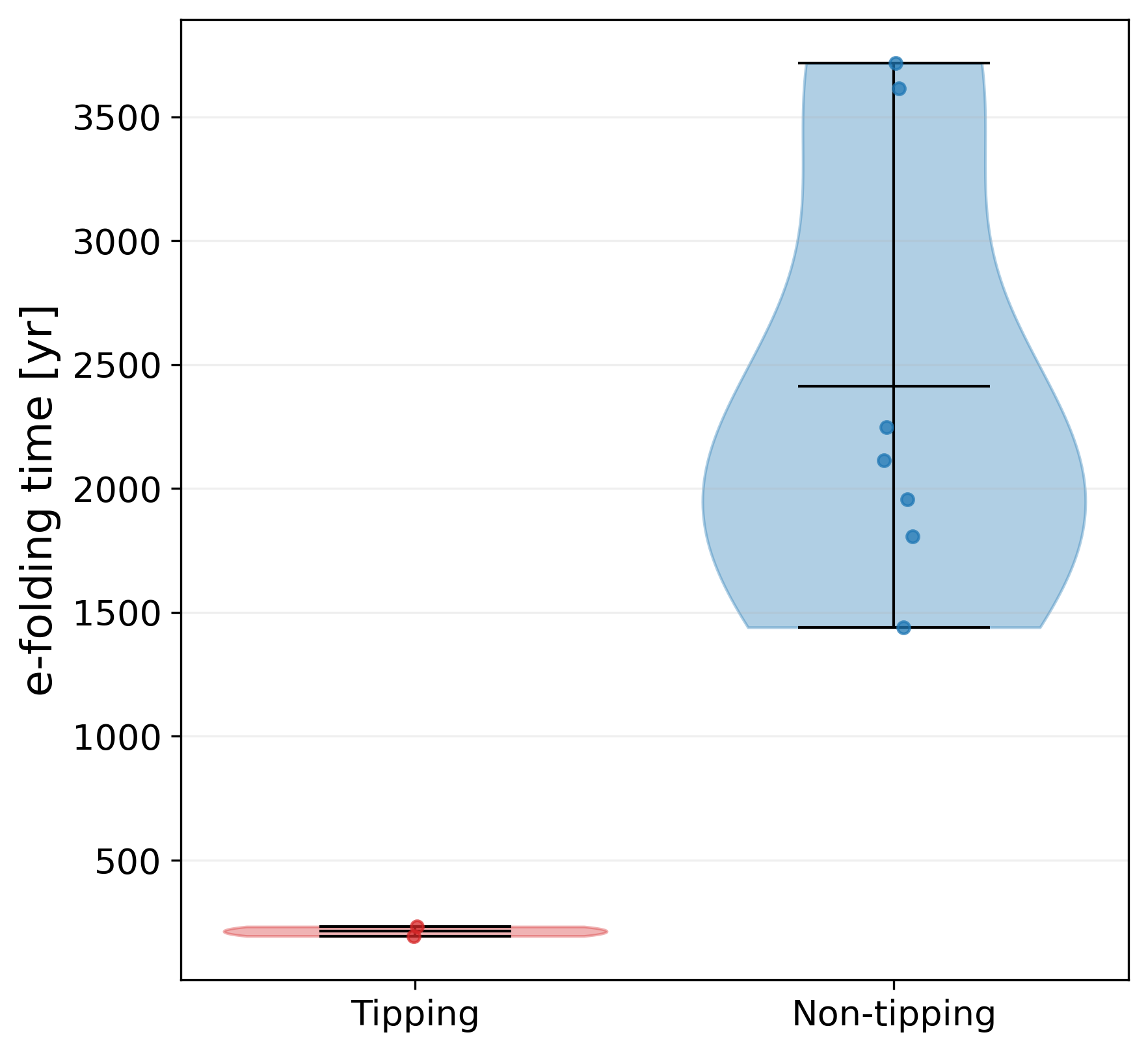}
        \caption{E-folding years}
        \label{si-fig:amoc_26N_efold}
    \end{subfigure}
    \hfill
    \begin{subfigure}[h]{0.41\linewidth}
        \centering
        \includegraphics[width=\linewidth]{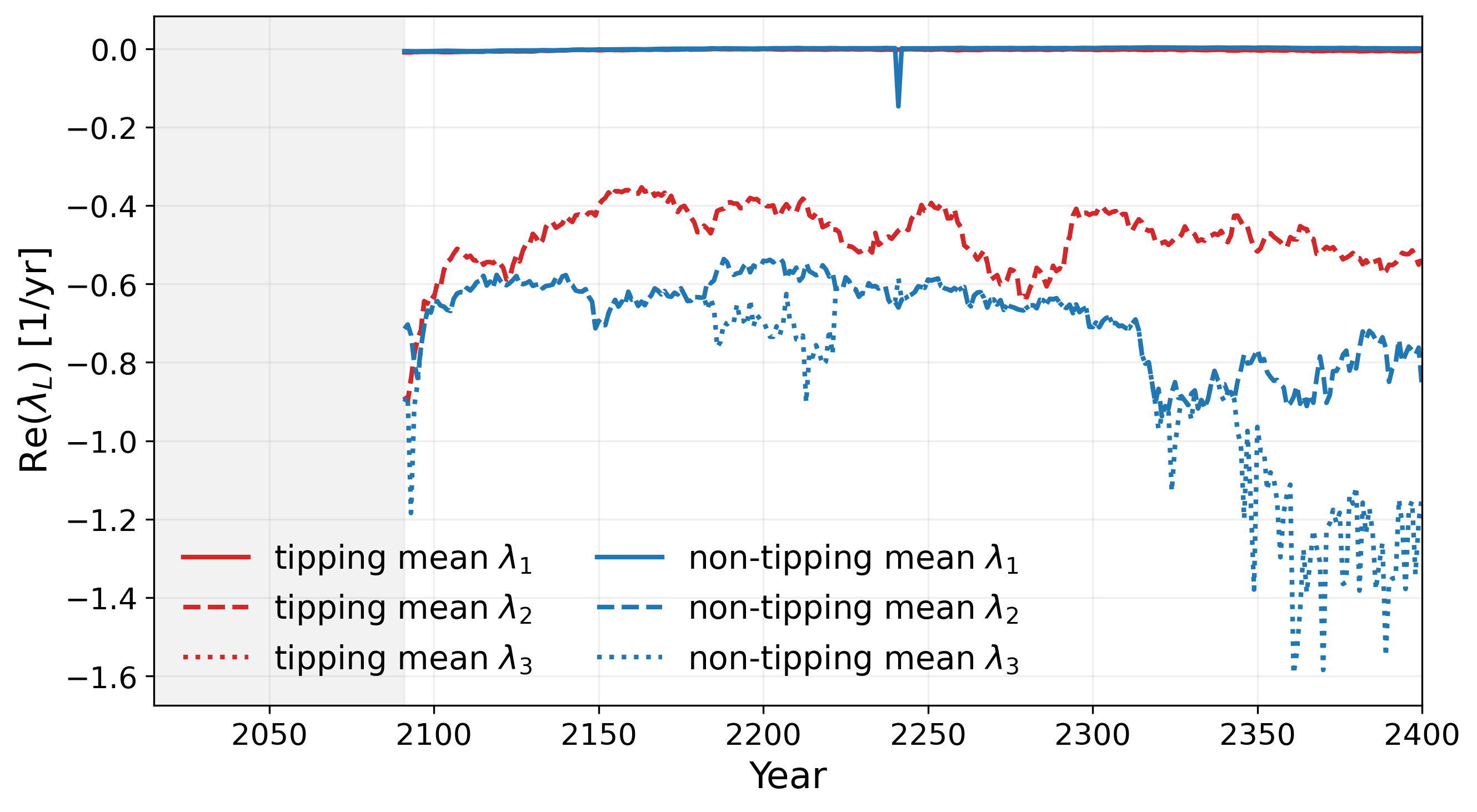}
        \caption{Rolling $\lambda_L$}
        \label{si-fig:amoc_26N_gap}
    \end{subfigure}
    \caption{Supplementary coupled AMOC results at $26^\circ$N. (\textbf{a}) Depth profile of overturning streamfunction for representative coupled AMOC trajectories at $26^\circ$N: the first two columns show non-tipping runs that recover to the stronger branch, whereas the latter two show tipping runs that evolve toward a much weaker circulation state. (\textbf{b}) EWS on pre-transition sliding windows of length $20\%$. (\textbf{c}) Discrete-time Koopman eigenvalues. (\textbf{d}) E-folding times of the slowest decaying non-constant mode. (\textbf{e}) Rolling generator eigenvalues with smallest $|\mathrm{Re}(\lambda_L)|$.}
    \label{si-fig:amoc_26N}
\end{figure}

\clearpage
\newpage

\subsection{AMOC control variable ablation}
\label{si-sec:ablate_amoc_forcing}

To assess sensitivity of the control-aware Koopman residual to the choice of forcing channel, we repeat the coupled AMOC analysis using several alternative control variables. In place of the Denmark Strait sea-ice flux used in the main text, we consider the maximum mixed-layer depth (MLD) in the Labrador Sea, Irminger Sea, Scotland-Iceland Basin, and Greenland-Iceland-Norway Seas. For each forcing variable, we apply the same preprocessing and compute control-aware ResKMD with TDF observables for AMOC strength at both $26^\circ$N and $48^\circ$N.

Figure~\ref{si-fig:ablate_amoc_forcing} compares the forcing trajectories and the corresponding control-aware residuals across the tipping and non-tipping ensembles. The results show that the performance of control-aware ResKMD is sensitive to forcing choice: some forcings produce clearer separation between tipping and non-tipping members than others. This supports the interpretation that control-aware Koopman indicators are most informative when the supplied control captures the dominant mechanism associated with AMOC weakening.

\begin{figure}[h!]
    \centering
    \includegraphics[width=\linewidth]{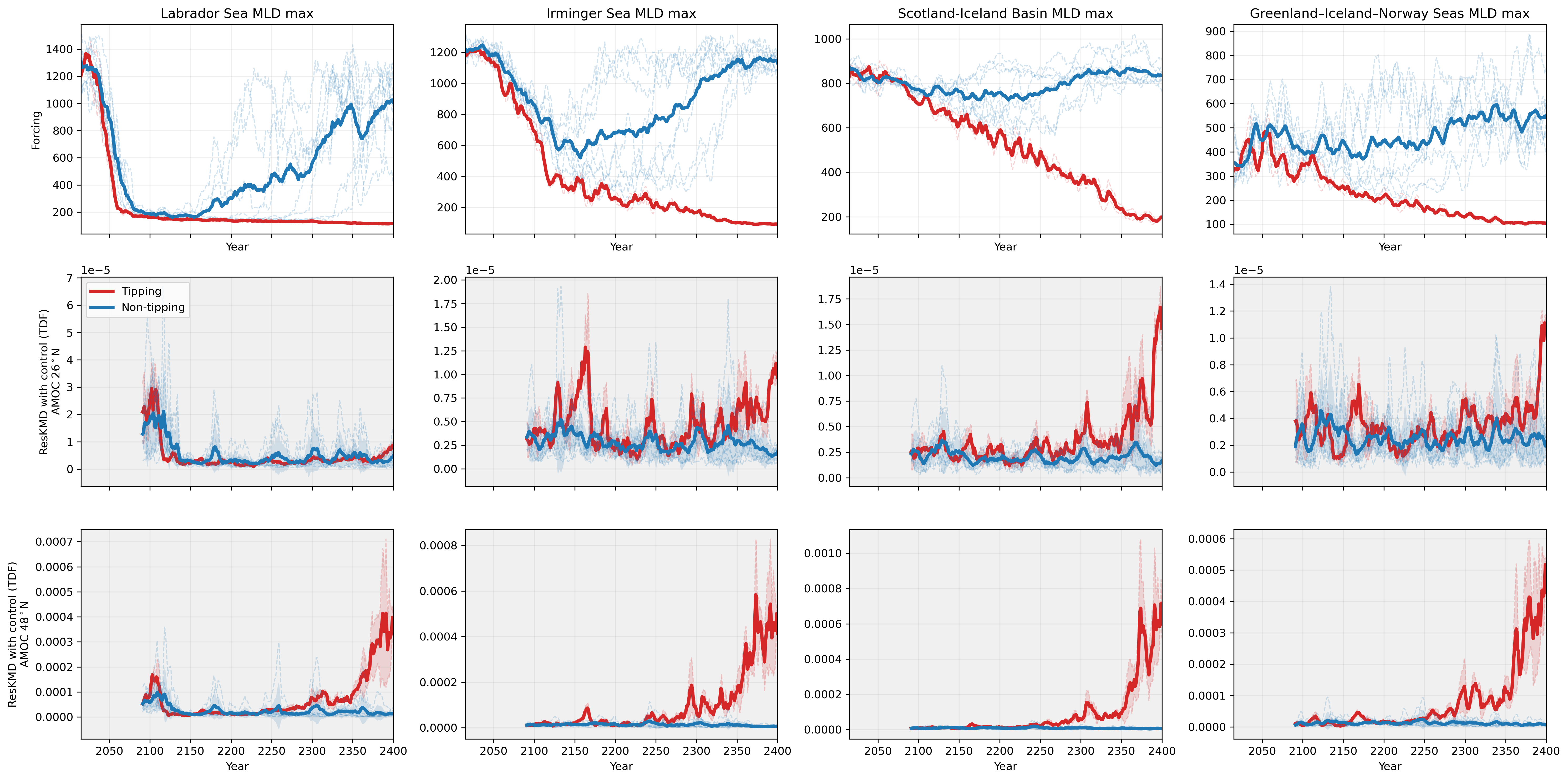}
    \caption{Ablation of forcing variables in control-aware ResKMD for the coupled AMOC runs. Columns correspond to alternative forcing inputs used in the control channel: Labrador Sea MLD maximum, Irminger Sea MLD maximum, Scotland--Iceland Basin MLD maximum, and Greenland--Iceland--Norway Seas MLD maximum. The first row shows the corresponding 10-year rolling-average forcing trajectories, grouped by tipping (red) and non-tipping (blue) ensemble members. The second and third rows show the resulting control-aware ResKMD with TDF observables using sliding windows of length $20\%$ for AMOC strength at $26^\circ$N and $48^\circ$N, respectively.}
    \label{si-fig:ablate_amoc_forcing}
\end{figure}

\end{document}